\documentclass[%
reprint,
nofootinbib,
amsmath,amssymb,
aps,
amsthm,
pra,
]{revtex4-1}

\usepackage{ifthen}
\newboolean{arxiv}
\setboolean{arxiv}{true}
\ifthenelse{\boolean{arxiv}}{\pdfoutput=1}{}

\usepackage{theorem} 
\usepackage{ascmac}
\usepackage{bbm}
\usepackage{bm} 
\usepackage{braket}
\usepackage{dcolumn} 
\usepackage{enumerate}
\usepackage{enumitem}
\usepackage{framed}
\usepackage{graphicx} 
\ifthenelse{\boolean{arxiv}}{\usepackage{hyperref}}{\usepackage[dvipdfmx]{hyperref}}
\usepackage{longtable}
\usepackage{mathtools}
\usepackage{multirow}
\ifthenelse{\boolean{arxiv}}{}{\usepackage{pxjahyper}}
\usepackage{txfonts} 
\usepackage{xparse}

\usepackage{color}
\usepackage[most]{tcolorbox}

\ifthenelse{\boolean{arxiv}}{}{\mathtoolsset{showonlyrefs}}

\newcommand{\hypercolor}{blue}
\hypersetup{
  colorlinks,
  citecolor=\hypercolor,
  linkcolor=\hypercolor,
  urlcolor=\hypercolor,
  bookmarksopen,
  bookmarksopenlevel=4,
  bookmarksnumbered
}

\newboolean{figure}
\setboolean{figure}{true}
\newcommand{\InsertPDF}[2]{\iffigure\includegraphics[scale=#1]{#2}\fi}

\usepackage{bm}

\theorembodyfont{\upshape}
\newtheorem{thm}{Theorem}
\newtheorem{proposition}[thm]{Proposition}
\newtheorem{remark}[thm]{Remark}

\newtheorem{define}{Definition}
\newtheorem{postulateno}{Postulate}

\newtheorem{lemma}[thm]{Lemma}
\newtheorem{cor}[thm]{Corollary}
\newtheorem{ex}{Example}

\newcounter{proof}

\ExplSyntaxOn
\NewDocumentEnvironment{proof}{o}
 {
  \par\medskip
  \noindent
  \textbf{Proof~}
 }
 {\QED\par\smallskip}
\ExplSyntaxOff

\newcounter{postulate}
\renewcommand{\thepostulate}{\arabic{postulate}}
\ExplSyntaxOn
\NewDocumentEnvironment{postulate}{oo}
 {
  \refstepcounter{postulate}
  \begin{postulateno}
  \textbf{\hspace{-0.5em}\IfNoValueTF{#2}{\thepostulate}{#2} ~(\IfNoValueTF{#1}{}{#1})}
 }
 {
  \end{postulateno}
 }
\ExplSyntaxOff

\newcommand{\ha}{\hat{a}}

\newcommand{\hc}{\hat{c}}

\newcommand{\hPi}{\hat{\Pi}}
\newcommand{\hPhi}{\hat{\Phi}}
\newcommand{\hPsi}{\hat{\Psi}}

\newcommand{\hLambda}{\hat{\Lambda}}

\newcommand{\hrho}{\hat{\rho}}

\newcommand{\hsigma}{\hat{\sigma}}

\newcommand{\mC}{\mathcal{C}}
\newcommand{\mD}{\mathcal{D}}

\newcommand{\mG}{\mathcal{G}}
\newcommand{\mH}{\mathcal{H}}
\newcommand{\mI}{\mathcal{I}}

\newcommand{\mN}{\mathcal{N}}

\newcommand{\mP}{\mathcal{P}}

\newcommand{\mS}{\mathcal{S}}
\newcommand{\mT}{\mathcal{T}}

\newcommand{\mX}{\mathcal{X}}

\newcommand{\mZ}{\mathcal{Z}}

\newcommand{\PS}{P_{\rm S}}
\newcommand{\PE}{P_{\rm E}}
\newcommand{\PI}{P_{\rm I}}

\newcommand{\ident}{\hat{1}}

\newcommand{\Real}{\mathbb{R}}
\newcommand{\Complex}{\mathbb{C}}

\newcommand{\Integer}{\mathbb{Z}}
\newcommand{\POVM}{\mM}

\newcommand{\QED}{\hspace*{0pt}\hfill $\blacksquare$}

\DeclareMathOperator{\argmax}{argmax}
\DeclareMathOperator{\argmin}{argmin}
\renewcommand{\i}{\mathbf{i}}
\newcommand{\T}{\mathsf{T}}
\DeclareMathOperator{\Tr}{Tr}
\newcommand{\Trp}[1]{\mathop{\mathrm{Tr}_{#1}}}

\DeclareMathAlphabet{\mymathbb}{U}{BOONDOX-ds}{m}{n}
\newcommand{\zero}{\mymathbb{0}}
\newcommand{\kket}[1]{|#1\rangle\!\rangle}
\newcommand{\bbra}[1]{\langle\!\langle#1|}

\newcommand{\bigop}[1]{\mathop{\vphantom{\sum}\mathchoice{\vcenter{\hbox{\huge $#1$}}}{\vcenter{\hbox{\Large $#1$}}}{#1}{#1}}\displaylimits}

\newcommand{\opt}{\star}

\newcommand{\summ}{\sum_{m=0}^{M-1}}
\newcommand{\sumj}{\sum_{j=0}^{J-1}}
\newcommand{\sumr}{\sum_{r=0}^{R-1}}
\newcommand{\sumg}{\sum_{g \in \mG}}
\newcommand{\sumk}{\sum_{k=0}^{K-1}}
\renewcommand{\ident}{\mathbbm{1}}
\newcommand{\I}{I}
\newcommand{\V}{V}
\newcommand{\W}{W}
\newcommand{\X}{X}
\newcommand{\tV}{{\tilde{\V}}}
\newcommand{\Vt}{{\V_t}}
\newcommand{\Wt}{{\W_t}}
\newcommand{\VT}{{\V_T}}
\newcommand{\WT}{{\W_T}}
\newcommand{\WVT}{{\WT \ot \VT}}
\newcommand{\WV}{{\W \ot \V}}
\newcommand{\WVt}{{\Wt \ot \Vt}}
\newcommand{\NV}{N_\V}
\newcommand{\NW}{N_\W}
\newcommand{\Pos}{\mathsf{Pos}}

\newcommand{\Chn}{\mathsf{Chn}}
\newcommand{\Den}{\mathsf{Den}}
\newcommand{\DenP}{\Den^\mathsf{P}}
\newcommand{\Her}{\mathsf{Her}}
\newcommand{\Uni}{\mathsf{Uni}}
\newcommand{\inter}{\mathrm{int}}
\newcommand{\cE}{\mathcal{E}}
\newcommand{\tg}{\tilde{g}}
\renewcommand{\th}{\tilde{h}}
\newcommand{\hcE}{\hat{\cE}}
\newcommand{\tcE}{\tilde{\cE}}

\renewcommand{\POVM}{\mathsf{POVM}}

\newcommand{\Tester}{\mathsf{Test}}
\newcommand{\ot}{\otimes}
\newcommand{\cross}{\times}
\newcommand{\tst}[1]{{\textstyle #1}}
\newcommand{\Lin}{\mathsf{Lin}}
\newcommand{\Ot}[2]{\tst{\bigotimes_{#1}^{#2}}}
\newcommand{\OtT}{\Ot{t=1}{T}}

\newcommand{\astT}{\tst{\bigast_{t=1}^T}}
\newcommand{\astTPos}{\astT \Pos(\Vt,\Wt)}
\newcommand{\astTHer}{\astT \Her(\Vt,\Wt)}
\newcommand{\astTChn}{\astT \Chn(\Vt,\Wt)}
\newcommand{\astTp}{\tst{\bigast_{t=1}^{T+1}}}
\newcommand{\astTpPos}{\astTp \Pos(\W_{t-1},\Vt)}
\newcommand{\astTpHer}{\astTp \Her(\W_{t-1},\Vt)}
\newcommand{\Pro}[2]{\tst{\prod_{#1}^{#2}}}
\newcommand{\Ad}{\mathrm{Ad}}
\newcommand{\mTh}{\hat{\mT}}
\newcommand{\mTG}{{\mT_\G}}
\newcommand{\mTGh}{\mTh_\G}
\newcommand{\G}{\mathrm{G}}
\renewcommand{\S}{\mS}
\newcommand{\SG}{{\S_\G}}
\renewcommand{\P}{\mP}
\newcommand{\PG}{\P_\G}
\newcommand{\D}{\mD}
\newcommand{\DG}{\D_\G}
\newcommand{\mCG}{{\mC_\G}}
\newcommand{\Prob}{\mathsf{Prob}}
\newcommand{\Probmax}{\Prob_\mathrm{max}}

\newcommand{\ts}{\tilde{s}}
\newcommand{\tit}{\tilde{t}}
\newcommand{\Choi}{\mathsf{C}}
\newcommand{\tChoi}{\tilde{\Choi}}

\newcommand{\tPhi}{\tilde{\Phi}}

\newcommand{\tPsi}{\tilde{\Psi}}
\newcommand{\tTheta}{\tilde{\Theta}}

\newcommand{\lambdaS}{\lambda_\S}
\newcommand{\lambdaSG}{\lambda_\SG}
\newcommand{\lambdamax}{\lambda_\mathrm{max}}
\newcommand{\tU}{\tilde{U}}
\newcommand{\tLambda}{\tilde{\Lambda}}
\newcommand{\Endash}{\tst{\textendash}}
\renewcommand{\c}{\circ}
\newcommand{\g}{{(g)}}
\let\b\relax
\DeclareMathOperator{\b}{\bullet}
\renewcommand{\d}{\diamond}
\renewcommand{\ol}{\overline}
\DeclareMathOperator{\co}{\mathsf{co}}
\DeclareMathOperator{\clco}{\ol{\co}}
\DeclareMathOperator{\coni}{\mathsf{coni}}
\DeclareMathOperator{\clconi}{\ol{\coni}}
\newcommand{\inv}[1]{\bar{#1}}
\let\ast\relax
\DeclareMathOperator{\ast}{\circledast}
\newcommand{\bigast}{\bigop{\ast}}
\newcommand{\Sa}{S_{\!\mathrm{a}}}
\newcommand{\pinc}{p_\mathrm{inc}}
\newcommand{\pNP}{p_\mathrm{false}}
\newcommand{\x}{\mathrm{x}}
\newcommand{\y}{\mathrm{y}}
\newcommand{\z}{\mathrm{z}}
\newcommand{\range}[2]{\{#1,\dots,#2\}}
\newcommand{\termdef}{\textit}

\setlist[enumerate]{label=\arabic*), leftmargin=3em, itemsep=0pt, parsep=0pt, labelwidth=5em}

\let\protect\relax
{\catcode`\|=\active
  \xdef\Craket{\protect\expandafter\noexpand\csname Craket \endcsname}
  \expandafter\gdef\csname Craket \endcsname#1{\begingroup
     \ifx\SavedDoubleVert\relax
       \let\SavedDoubleVert\|\let\|\BraDoubleVert
     \fi
     \mathcode`\|32768\let|\BraVert
     \left({#1}\right)\endgroup}
}

\definecolor{memo}{RGB}{128,0,255}
\definecolor{gray}{RGB}{128,128,128}

\newcommand{\Discard}[1]{}

\usepackage{amsmath}	
\begin{document}

\preprint{APS/123-QED}

\title{Generalized quantum process discrimination problems}

\affiliation{%
 Quantum Information Science Research Center, Quantum ICT Research Institute, Tamagawa University,
 Machida, Tokyo 194-8610, Japan
}%

\author{Kenji Nakahira}
\affiliation{%
 Quantum Information Science Research Center, Quantum ICT Research Institute, Tamagawa University,
 Machida, Tokyo 194-8610, Japan
}%

\author{Kentaro Kato}
\affiliation{%
 Quantum Information Science Research Center, Quantum ICT Research Institute, Tamagawa University,
 Machida, Tokyo 194-8610, Japan
}%

\date{\today}

\begin{abstract}
 We study a broad class of quantum process discrimination problems
 that can handle many optimization strategies
 such as the Bayes, Neyman-Pearson, and unambiguous strategies,
 where each process can consist of multiple time steps and can have an internal memory.
 Given a collection of candidate processes, our task is to find a discrimination strategy,
 which may be adaptive and/or entanglement-assisted, that maximizes a given objective function
 subject to given constraints.
 Our problem can be formulated as a convex problem.
 Its Lagrange dual problem with no duality gap and necessary and sufficient conditions
 for an optimal solution are derived.
 We also show that if a problem has a certain symmetry and at least one optimal solution exists,
 then there also exists an optimal solution with the same type of symmetry.
 A minimax strategy for a process discrimination problem is also discussed.
 As applications of our results, we provide some problems in which
 an adaptive strategy is not necessary for optimal discrimination.
 We also present an example of single-shot channel discrimination
 for which an analytical solution can be obtained.
\end{abstract}

\maketitle



\section{Introduction}

A quantum process, which is a mathematical object that models the probabilistic description of quantum phenomena, plays a fundamental role in quantum information theory.
Identifying a quantum process is of great importance to characterize the behavior of quantum devices.
We focus on the situation in which a process is known to belong to a given finite collection of processes; our goal is to determine which one is used.
This problem often arises, e.g., in quantum communication, quantum metrology, and quantum cryptography.

Quantum states can be regarded as a special case of quantum processes.
Since the seminal works of Helstrom, Holevo, and Yuen {\it et al.}
\cite{Hel-1969,Hol-1973,Yue-Ken-Lax-1975} appeared in the end of the 1960's
and 1970's, quantum state discrimination has been extensively investigated
\cite{Ban-Kur-Mom-Hir-1997,Bar-2001,Jez-Reh-Fiu-2002,Cho-Hsu-2003,
Eld-For-2001,Kat-Hir-2003,Eld-Meg-Ver-2004,Tys-2010,Bae-2013,Nak-Usu-2013-group}.
This problem can be formulated as a semidefinite programming (SDP) problem
(e.g., \cite{Bel-1975,Eld-Meg-Ver-2003}),
which allows us to easily analyze properties of optimal discrimination.
Many optimization strategies can be considered,
among which it is necessary to choose a suitable one depending on the problem being solved.
Possibly the simplest practical strategy is
to find discrimination maximizing the average success probability,
which is often called minimum-error discrimination.
The Bayes strategy \cite{Hol-1973,Yue-Ken-Lax-1975,Hel-1976} and
the Neyman-Pearson strategy \cite{Hel-1976,Hol-1982-Prob,Par-1997} are also frequently used.
As other strategies, discrimination maximizing the average success probability
has been investigated subject to several constraints: for example,
errors are not allowed \cite{Iva-1987,Che-Bar-1998}
(which is called optimal unambiguous discrimination),
the average error probability does not exceed a fixed value
\cite{Tou-Ada-Ste-2007,Hay-Has-Hor-2008,Sug-Has-Hor-Hay-2009},
and the average inconclusive (or failure) probability is fixed
\cite{Che-Bar-1998-inc,Eld-2003-inc,Fiu-Jez-2003}
(which is referred to as optimal inconclusive discrimination).
In the case in which the prior probabilities of the states are unknown,
to optimize discrimination, several strategies based on the minimax criterion have been investigated
\cite{Hir-Ike-1982,Osa-Ban-Hir-1996,Dar-Sac-Kah-2005,Kat-2012,Nak-Kat-Usu-2013-minimax}.
Moreover, a generalized state discrimination problem,
which can handle all of the above mentioned strategies, was proposed
\cite{Nak-Kat-Usu-2015-general}.
In these studies, necessary and sufficient conditions for optimal discrimination
have been formulated.
These results help us to find analytical and/or numerical optimal solutions.

A quantum process discrimination problem is more general and
often more difficult to solve than a state discrimination problem.
States, effects, measurements, channels, and superchannels are all special cases of quantum processes.
In this paper, we are concerned with the task of discriminating quantum processes
each of which can consist of multiple time steps and can have an internal memory.
Process discrimination (in particular in the cases of single-shot and multi-shot channels,
including measurements) has been an active area of research for at least the past two decades.
Discrimination of two quantum processes with maximum average success probability
has been widely studied \cite{Aci-2001,Sac-2005,Sac-2005-EB,Li-Qiu-2008,Mat-Pia-Wat-2010,Sed-Zim-2014,
Pir-Lup-2017,Puc-Paw-Kra-Kuk-2018,Pir-Lau-Lup-Per-2019}.
Optimal unambiguous discrimination \cite{Wan-Yin-2006,Zim-Hei-2008,Zim-Sed-2010,Reh-Far-Jeo-Shi-2018},
optimal inconclusive discrimination \cite{Sed-Zim-2014},
and the Neyman-Pearson strategy \cite{Maf-Kou-Gha-Par-2019,Hir-2021} have also been investigated.
It is well known that the problem of finding minimum-error discrimination between two channels
can be formulated as an SDP problem \cite{Gil-Lan-Nie-2005,Gut-Wat-2007,Wat-2009}.
In the more general case of more than two processes that can consist of multiple time steps
with or without memory, the problem has been shown to be formulated as an SDP problem \cite{Chi-2012}
(see also \cite{Zim-2008,Jen-Pla-2016} for the case of single-step processes).
Note that such a problem can handle adaptive (feedback-assisted)
and/or entanglement-assisted discrimination.
However, in particular in the case of multi-step processes,
only a few optimization strategies have ever been reported;
these results cannot readily be applied to many other optimization strategies.
Moreover, the properties of optimal discrimination are not known except for some special cases.

In this paper, we address generalized process discrimination problems,
which are applicable to a broad class of optimization strategies
including all of the above mentioned ones.
Our approach can significantly reduce the required efforts for analyzing
this class of process discrimination problems
compared to analyzing these problems separately.
We show that our discrimination problems are formulated as convex problems,
which are a generalization of SDP problems.
Convex problems are well-understood, and thus our formulation allows us to
easily investigate the properties of optimal discrimination.
Note that the problems addressed in this paper can be interpreted as an extension of
generalized state discrimination problems treated in Ref.~\cite{Nak-Kat-Usu-2015-general}.
However, the techniques used in Ref.~\cite{Nak-Kat-Usu-2015-general} cannot
directly be used for our problems; process discrimination problems are much harder to
analyze than state discrimination problems.

The paper is organized as follows.
In Sec.~\ref{sec:channel}, we provide a generalized process discrimination problem,
which is formulated as a convex problem with a so-called quantum tester.
In Sec.~\ref{sec:opt}, we provide its Lagrange dual problem and
show that the optimal values of the primal and dual problems coincide.
Also, necessary and sufficient conditions for a tester to be optimal are given.
Moreover, we derive necessary and sufficient conditions that the optimal value
remain unchanged even when a certain additional constraint is imposed.
In Sec.~\ref{sec:symmetry}, it is shown that if a problem has a certain symmetry
and an optimal solution exists, then there also exists an optimal solution
having the same type of symmetry.
In Sec.~\ref{sec:minimax}, we introduce a minimax version of a process discrimination problem.
In Sec.~\ref{sec:example}, some examples are given to demonstrate
how to apply our results to solve a problem.

\section{Process discrimination problems} \label{sec:channel}

\subsection{Notation}

We first introduce some notation.
$\Real$, $\Real_+$, and $\Complex$ denote, respectively,
the sets of all real, nonnegative real, and complex numbers.
The complex conjugate of $z \in \Complex$ is denoted by $z^*$.
For each finite-dimensional complex Hilbert space
(which we also call a system) $\V$, let $\NV$ be its dimension.
We will identify a one-dimensional system with $\Complex$.
For each matrix $X$ on $\V$, let $X^\dagger$ and $X^\T$ be, respectively,
the Hermitian transpose and the transpose of $X$ (in the standard basis of $\V$).
Let $\Her_\V$ and $\Pos_\V$ be, respectively, the sets of all Hermitian and
positive semidefinite matrices on $\V$.
$\Her_\V$ is an $\NV^2$-dimensional real Hilbert space with the inner product defined by
$\braket{X,Y} \coloneqq \Tr(XY)$ $~(X,Y \in \Her_\V)$.
A positive semidefinite matrix is called pure if it has rank one.
We will denote by $\Her(\V,\W)$ the set of all linear maps from $\Her_\V$ to $\Her_\W$,
every element of which is called Hermitian-preserving.
Let $\Pos(\V,\W)$ and $\Chn(\V,\W)$ be, respectively, the sets of all completely positive (CP) maps
and all trace-preserving CP maps from $\Her_\V$ to $\Her_\W$.
Moreover, let $\Den_\V$ be the set of all positive semidefinite matrices with unit trace
(i.e., density matrices) on $\V$
and $\DenP_\V$ be the set of all pure elements in $\Den_\V$.
For a set $\mX$ in a real vector space, let $\Lin(\mX)$ be the smallest real vector space
containing $\mX$.
Obviously, we have $\Chn(\V,\W) \subset \Pos(\V,\W) \subset \Her(\V,\W)$,
$\DenP_\V \subset \Den_\V \subset \Pos_\V \subset \Her_\V$,
$\Lin[\Pos(\V,\W)] = \Her(\V,\W)$, and $\Lin(\Pos_\V) = \Her_\V$.
We can identify $\Chn(\Complex,\V)$ with $\Den_\V$, $\Pos(\Complex,\V)$ with $\Pos_\V$,
and $\Her(\Complex,\V)$ with $\Her_\V$.
$\I_\V$ and $\ident_\V$, respectively, denote the identity matrix on $\V$
and the identity map on $\Her_\V$.
$\zero$ denotes a zero matrix.
In quantum theory, each single-step process is described by a CP map.
In particular, a single-step process described by a trace-preserving CP map is called a quantum channel.
Any quantum state, which is described by a density matrix,
and any quantum measurement, which is described by a positive operator-valued measure (POVM),
can be regarded as special cases of quantum channels.
Fix a natural number $M \ge 2$ and denote by $\POVM_\V$ the set of all POVMs with $M$
elements on a system $\V$.
Throughout this paper, we consider only measurements with a finite number of outcomes.
Given a set $\mX$, let $\inter(\mX)$, $\mX^*$, $\co \mX$, and $\coni \mX$
be the interior, the dual cone, the convex hull
[i.e., $\co \mX \coloneqq \{ \sum_i p_i x_i : p_i \in \Real_+, ~\sum_i p_i = 1, ~x_i \in \mX \}$],
and the (convex) conical hull
[i.e., $\coni \mX \coloneqq \{ \sum_i p_i x_i : p_i \in \Real_+, ~x_i \in \mX \}$] of $\mX$.
We denote the closure of $\mX$ by $\ol{\mX}$, $\ol{\co \mX}$ by $\clco \mX$,
and $\ol{\coni \mX}$ by $\clconi \mX$.
For a given natural number $T$, let $\tV \coloneqq \WVT \ot \cdots \ot \W_1 \ot \V_1$.
For any $X, Y \in \Her_\V$, let $X \ge Y$ (or $Y \le X$) denote $X - Y \in \Pos_\V$.
For any natural number $n$, let $\mI_n \coloneqq \range{0}{n-1}$.
$\delta_{n,n'}$ denotes the Kronecker delta.
Let $\Uni_\V$ be the set of all unitary and anti-unitary operators on $\V$.
For any $U \in \Uni_\V$, the linear map $\Ad_U \in \Her(\V,\V)$ is defined as%
\footnote{$U$ is an anti-unitary operator on $\V$ if and only if
there exists a unitary operator $\tU \in \Uni_\V$ such that
$\Ad_U(X) = \Ad_{\tU}(X^\T)$ $~(X \in \Her_\V)$.
If $U$ is anti-unitary, then $\Ad_U$ is not CP.}%
\begin{alignat}{2}
 \Ad_U(X) &\coloneqq U X U^\dagger, &\quad &X \in \Her_\V.
\end{alignat}
$\Trp{\V}$ denotes the partial trace over $\V$.

\subsection{Quantum processes, testers, and combs}

\subsubsection{Processes and testers}

We shall introduce a quantum process (or a quantum network) and a quantum tester
\cite{Chi-Dar-Per-2008,Chi-Dar-Per-2009,Chi-2012}
(see also a quantum strategy \cite{Gut-Wat-2007}).
Let us consider the connection of
$T$ linear maps
$\{ \hc^{(t)} \in \Her(\W'_{t-1} \ot \Vt, \W'_t \ot \Wt) \}_{t=1}^T$
as shown in Fig.~\ref{fig:tester},
where $\W'_0 \coloneqq \Complex$ and $\W'_T \coloneqq \Complex$.
We mathematically express this process as%
\footnote{Although a linear map $\hat{x} \in \Her(\V,\W)$ is not CP in general,
we will, by abuse of language, refer to $\hat{x}$ as a (single-step) process.
Also, we refer to $\hc$ as a process.}%
\begin{alignat}{1}
 \hc &\coloneqq \hc^{(T)} \ast \hc^{(T-1)} \ast \cdots \ast \hc^{(1)},
 \label{eq:network}
\end{alignat}
where $\ast$ denotes the connection of processes,
which is called the link product \cite{Chi-Dar-Per-2008}.
$\hc$ has definite causal order; for any $t$ and $t'$ with $t < t'$,
signalling from $\hc^{(t')}$ to $\hc^{(t)}$ is impossible
[i.e., $\hc^{(t)}$ is not in the causal future of $\hc^{(t')}$].
$\W'_1,\dots,\W'_{T-1}$ are internal systems of process $\hc$.
Any memoryless process can be expressed in the form of Eq.~\eqref{eq:network}
with $\W'_1 = \dots = \W'_T = \Complex$.
Let $\astTHer$ be the set of all processes $\hc$ expressed in the form of Eq.~\eqref{eq:network}.
As a special case, if $\hc^{(1)} = \cdots = \hc^{(T)}$ holds, then
$\hc$ of Eq.~\eqref{eq:network} is denoted by $[\hc^{(1)}]^{\ast T}$.
Also, let $\astTPos$ and $\astTChn$ be, respectively,
the sets of all processes $\hc$ expressed in the form of Eq.~\eqref{eq:network}
with $\hc^{(t)} \in \Pos(\W'_{t-1} \ot \Vt, \W'_t \ot \Wt)$ and
$\hc^{(t)} \in \Chn(\W'_{t-1} \ot \Vt, \W'_t \ot \Wt)$ for each $t \in \range{1}{T}$.
$\astTChn \subset \astTPos \subset \astTHer$ obviously holds.

A collection of processes expressed in the form
\begin{alignat}{1}
 \hPhi &\coloneqq \{ \hPhi_m \}_{m=0}^{M-1}, \nonumber \\
 \hPhi_m &\coloneqq \hPi_m \ast \hsigma_T \ast \hsigma_{T-1} \ast \cdots \ast \hsigma_1
 \label{eq:hPhi}
\end{alignat}
with $T$ channels
$\{ \hsigma_t \in \Chn(\W_{t-1} \ot \V'_{t-1}, \Vt \ot \V'_t) \}_{t=1}^T$
(where $\W_0 \coloneqq \Complex$ and $\V'_0 \coloneqq \Complex$) and
a measurement $\hPi \coloneqq \{ \hPi_m \}_{m=0}^{M-1} \in \POVM_{\WT \ot \V'_T}$
is called a \termdef{quantum tester}.
It follows that $\hPhi_m \in \astTpPos$ holds, where $\V_{T+1} \coloneqq \Complex$.
Let $\mTGh$ be the set of all testers $\hPhi$ representable in the form of Eq.~\eqref{eq:hPhi}.
We will call each element $\hPhi_m$ of a tester $\hPhi$ a \termdef{tester element}.
In the special case of $T = 1$, a tester is often referred to as a process POVM \cite{Zim-2008}.
A process $\hc$ and a tester element $\hPhi_m$ can be connected as in Fig.~\ref{fig:tester},
which is mathematically expressed by
\begin{alignat}{1}
 \braket{\hPhi_m, \hc} &\coloneqq
 \hPi_m \c [\hc^{(T)} \ot \ident_{\V'_T}] \c \cdots \c [\hc^{(2)} \ot \ident_{\V'_2}]
 \nonumber \\
 &\quad \c [\ident_{\W'_1} \ot \hsigma_2] \c [\hc^{(1)} \ot \ident_{\V'_1}] \c \hsigma_1 \in \Real,
 \label{eq:inner_prod}
\end{alignat}
where $\c$ denotes the map composition.
\begin{figure}[bt]
 \centering
 \InsertPDF{1.0}{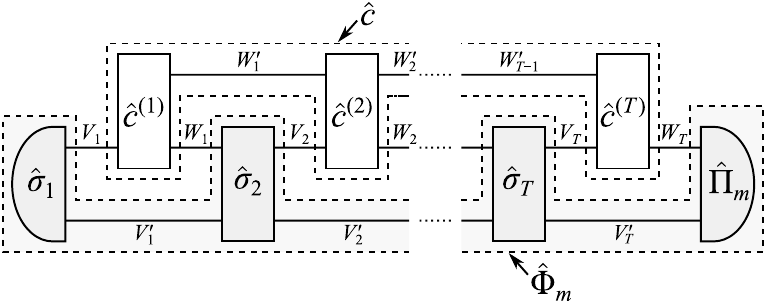}
 \caption{Quantum process
 $\hc \coloneqq \hc^{(T)} \ast \hc^{(T-1)} \ast \cdots \ast \hc^{(1)}$ and tester
 $\hPhi \coloneqq \{ \hPhi_m \coloneqq \hPi_m \ast \hsigma_T \ast \hsigma_{T-1} \ast \cdots
 \ast \hsigma_1 \}_{m=0}^{M-1}$.}
 \label{fig:tester}
\end{figure}

For any two processes $\hc, \hc' \in \astTHer$ and $q, q' \in \Real$,
$q \hc + q' \hc'$ is the element of $\astTHer$ uniquely characterized by
\begin{alignat}{2}
 \braket{\hPhi_m, q \hc + q' \hc'} &= q \braket{\hPhi_m, \hc} + q' \braket{\hPhi_m, \hc'}
\end{alignat}
for any tester element $\hPhi_m$.
Thus, $\astTHer$ can be considered as a real Hilbert space;
$\astTpHer$ is its dual space.

\subsubsection{Choi-Jamio{\l}kowski representations}

Quantum processes and testers can be conveniently mathematically described
in the so-called Choi-Jamio{\l}kowski representations
\cite{Cho-1975,Jam-1972,Jia-Luo-Fu-2013,Chi-Dar-Per-2008-memory}.
Specifically, the Choi-Jamio{\l}kowski representation
of a process $\hc \in \astTHer$, denoted by $\Choi_{\hc}$,
is given as Fig.~\ref{fig:Choi_process_tester}(a),
where $\hPsi_t \coloneqq \kket{\I_\Vt}\bbra{\I_\Vt} \in \Pos_{\Vt \ot \Vt}$,
$\kket{\I_\Vt} \coloneqq \sum_{i=1}^{N_\Vt} \ket{i} \ot \ket{i} \in \Vt \ot \Vt$
($\{ \ket{i} \}_{i=1}^{N_\Vt}$ is the standard basis of $\Vt$),
and $\bbra{\I_\Vt} \coloneqq \kket{\I_\Vt}^\dagger$.
Also, the Choi-Jamio{\l}kowski representation of a tester element
$\hPhi_m \in \astTpPos$, denoted by $\tChoi_{\hPhi_m}$,
is given as Fig.~\ref{fig:Choi_process_tester}(b),
where $\hPsi_t^\dagger \coloneqq \bbra{\I_\Vt} \cdot \Endash \cdot \kket{\I_\Vt}
\in \Pos(\Vt \ot \Vt,\Complex)$.
Both $\Choi$ and $\tChoi$ are well-defined as linear maps.
We can see that $\Choi:\astTHer \to \Her_\tV$ and
$\tChoi:\astTpHer \to \Her(\tV,\Complex)$ are surjective.
For each system $\V$, we often identify any $X \in \Her_\V$ with
$\braket{X,\Endash} \in \Her(\V,\Complex)$%
\footnote{As an example, we consider a POVM element $\Pi_0 \in \Pos_\V$.
In quantum theory, $\Pi_0$ is often identified with the linear map
$\braket{\Pi_0,\Endash} = \Tr(\Pi_0 \cdot \Endash) \in \Her(\V,\Complex)$.},
in which case $\tChoi$ can be regarded as
a map from $\astTpHer$ to $\Her_\tV$.
For the sake of brevity, we will denote the Choi-Jamio{\l}kowski representations
of processes and testers as the same letter without the hat symbol,
e.g., for each $\hc \in \astTHer$ and $\hPhi_m \in \astTpPos$, let
\begin{alignat}{2}
 c &\coloneqq \Choi_{\hc}, &\quad \Phi_m &\coloneqq \tChoi_{\hPhi_m}.
\end{alignat}
For convenience and without confusion,
we will also call $c$ and $\Phi$ a process and a tester, respectively.
We can easily verify
\begin{alignat}{1}
 \braket{\Phi_m, c} &= \braket{\hPhi_m, \hc}.
 \label{eq:Phi_c}
\end{alignat}
In the special case of $\W'_1 = \dots = \W'_{T-1} = \Complex$,
it follows that the Choi-Jamio{\l}kowski representation of each quantum process
$\hc \coloneqq \hc^{(T)} \ast \hc^{(T-1)} \ast \cdots \ast \hc^{(1)}$
$~[\hc^{(t)} \in \Chn(\Vt, \Wt)]$
is written by $c = c^{(T)} \ot c^{(T-1)} \ot \cdots \ot c^{(1)}$
[where $c^{(t)}$ is the Choi-Jamio{\l}kowski representation of $\hc^{(t)}$].
\begin{figure}[bt]
 \centering
 \InsertPDF{1.0}{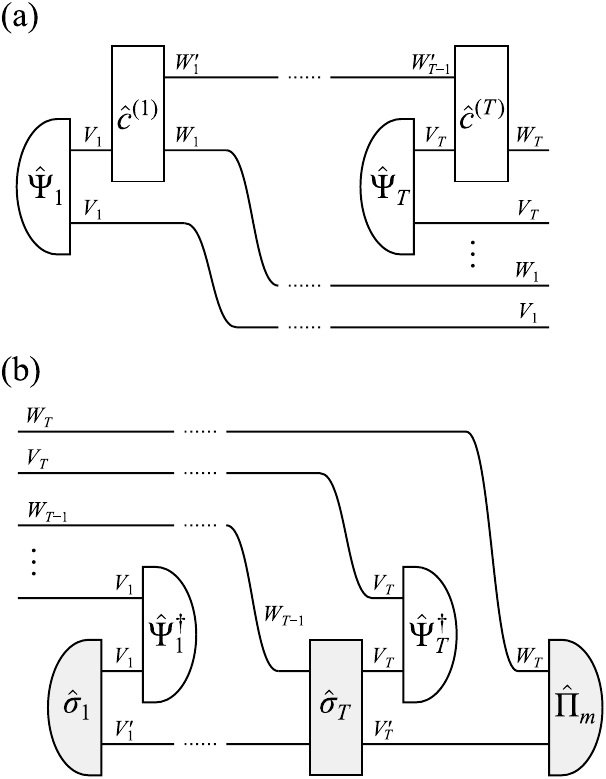}
 \caption{Choi-Jamio{\l}kowski representations of (a) a quantum process
 $\hc \coloneqq \hc^{(T)} \ast \hc^{(T-1)} \ast \cdots \ast \hc^{(1)}$ and
 (b) a quantum tester element
 $\hPhi_m \coloneqq \hPi_m \ast \hsigma_T \ast \hsigma_{T-1} \ast \cdots \hsigma_1$,
 where $\hPsi_t \coloneqq \kket{\I_\Vt}\bbra{\I_\Vt}$
 and $\hPsi_t^\dagger \coloneqq \bbra{\I_\Vt} \cdot \Endash \cdot \kket{\I_\Vt}$.}
 \label{fig:Choi_process_tester}
\end{figure}

\subsubsection{Combs}

Each element of $\astTChn$ is called a \termdef{quantum comb} \cite{Chi-Dar-Per-2008}
(also known as a supermap or a quantum strategy \cite{Gut-Wat-2007}).
For each comb $\hc$, we will also call $c \coloneqq \Choi_{\hc}$ a comb.
Let $\OtT \Chn_\WVt$ (or simply, $\Chn_\tV$) be the set of all combs $c \in \Pos_\tV$, i.e.,
\begin{alignat}{1}
 \Chn_\tV &\coloneqq \OtT \Chn_\WVt \coloneqq \left\{ \Choi_{\hc} : \hc \in \astTChn \right\}.
\end{alignat}
We can identify $\Chn_{\V \ot \Complex}$ with $\Den_\V$
and $\Chn_{\Complex \ot \V}$ with $\{ \I_\V \}$.
$c \in \Pos_\tV$ is a comb if and only if there exists
$\left\{ c_t \in \Pos_{\WVt \ot \cdots \ot \W_1 \ot \V_1} \right\}_{t=1}^T$
such that \cite{Chi-Dar-Per-2008,Chi-Dar-Per-2009}
\begin{alignat}{2}
 c_T &= c, \nonumber \\
 \Trp{\Wt} c_t &= \I_\Vt \ot c_{t-1}, &\quad &\forall t \in \range{2}{T}, \nonumber \\
 \Trp{\W_1} c_1 &= \I_{\V_1}.
 \label{eq:Comb}
\end{alignat}
For each comb $c$, $\{ c_t \}_{t=1}^T$ satisfying Eq.~\eqref{eq:Comb} is
uniquely determined by $c_T \coloneqq c$ and
\begin{alignat}{1}
 c_t &\coloneqq \frac{1}{N_{\V_{t+1}}} \Trp{\W_{t+1} \ot \V_{t+1}} c_{t+1}, \quad t \in \range{1}{T-1}.
\end{alignat}

Let
\begin{alignat}{1}
 \mTG &\coloneqq \left\{ \{ \tChoi_{\hPhi_m} \}_{m=0}^{M-1} : \hPhi \in \mTGh \right\}, \nonumber \\
 \SG &\coloneqq \Ot{t=1}{T+1} \Chn_{\Vt \ot \W_{t-1}}.
 \label{eq:TGSG}
\end{alignat}
Note that $\SG = \left\{ \I_\WT \ot \tau : \tau \in \OtT \Chn_{\Vt \ot \W_{t-1}} \right\}$ holds
from $\V_{T+1} = \Complex$.
$\Phi \coloneqq \{ \Phi_m \}_{m=0}^{M-1} \subset \Pos_\tV$ is in $\mTG$
if and only if $\summ \Phi_m \in \SG$ \cite{Gut-Wat-2007,Chi-Dar-Per-2008}.
Thus, we have
\begin{alignat}{1}
 \mTG &= \left\{ \Phi \in \mCG : \summ \Phi_m \in \SG \right\}, \nonumber \\
 \mCG &\coloneqq \Pos_\tV^M.
 \label{eq:TG}
\end{alignat}
We can easily verify
\begin{alignat}{2}
 \braket{\varphi,c} &= 1, &\quad &\forall c \in \Chn_\tV, ~ \varphi \in \SG,
 \label{eq:Phim_c_1}
\end{alignat}
which implies that, for every $c \in \Chn_\tV$ and $\Phi \in \mTG$,
$\{ \braket{\Phi_m, c} \}_{m=0}^{M-1}$ is a probability distribution.
Thus, $\Phi \in \mTG$ can be regarded as a map from combs to probability distributions.

\subsection{Discrimination problems}

To simplify the discussion, we first restrict ourselves to $T$-shot channel discrimination problems.
Let us consider the problem of determining which of $R$ known quantum channels,
$\{ \hLambda_r \}_{r=0}^{R-1} \subset \Pos(\V,\W)$, is used.
This problem is depicted as Fig.~\ref{fig:tester_channel_discrimination},
which can be seen as a special case of Fig.~\ref{fig:tester},
where $\hLambda_r$ is a given channel and $\Vt \coloneqq \V$ and $\W_t \coloneqq \W$
for each $t \in \range{1}{T}$.
To discriminate the channels, we first prepare an input state
$\hsigma_1 \in \Den_{\V_1 \ot \V'_1}$,
and then the channels $\hLambda_r \ot \ident_{\V'_1}, \hsigma_2,
\hLambda_r \ot \ident_{\V'_2}, \dots, \hsigma_T,
\hLambda_r \ot \ident_{\V'_T}$ are sequentially applied.
We finally perform a measurement
$\hPi \coloneqq \{ \hPi_m \}_{m=0}^{M-1} \in \POVM_{\W_T \ot \V'_T}$.
There exist many criteria for discriminating quantum channels.
When using the minimum-error criterion, we set $M \coloneqq R$ and
try to find a tester
$\hPhi \coloneqq \{ \hPhi_m \coloneqq \hPi_m \ast \hsigma_T \ast \cdots \ast \hsigma_1 \}_{m=0}^{M-1}
\in \mTGh$ that maximizes the average success probability
$\PS(\hPhi) \coloneqq \sum_{r=0}^{R-1} p_r \braket{\hPhi_r, \hLambda_r^{\ast T}}$,
where $p_r$ is the prior probability of the channel $\hLambda_r$.
This problem can be written as
\begin{alignat}{1}
 \begin{array}{ll}
  \mbox{maximize} & \displaystyle \PS(\hPhi) \\
  \mbox{subject~to} & \hPhi \in \mTGh. \\
 \end{array}
 \label{prob:PPi_success}
\end{alignat}
\begin{figure}[bt]
 \centering
 \InsertPDF{1.0}{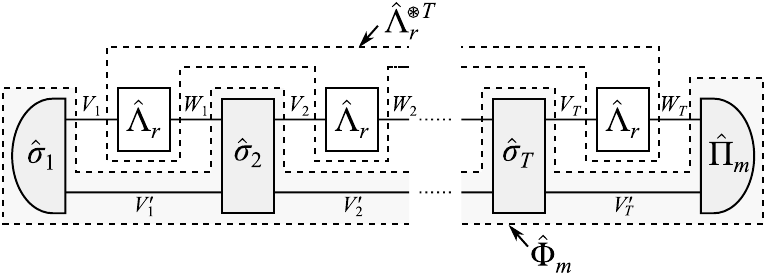}
 \caption{$T$-shot channel discrimination.
 $\hLambda_r$ is a channel and $\{\hPhi_m\}_{m=0}^{M-1}$ is a tester.}
 \label{fig:tester_channel_discrimination}
\end{figure}

The above discussion easily extends to discrimination of more general processes,
e.g., multi-shot subchannel discrimination or discrimination of processes
each of which consists of multiple time steps.
We give three typical examples.

\begin{ex} \label{ex:memoryless_comb}
 The first example is the problem of discriminating quantum memoryless combs
 $\{ \hcE_r \}_{r=0}^{R-1}$,
 where each $\hcE_r$ is characterized by the connection of $T$ channels
 $\hLambda^{(1)}_r,\dots,\hLambda^{(T)}_r$, i.e.,
 \begin{alignat}{1}
  \hcE_r &\coloneqq \hLambda^{(T)}_r \ast \cdots \ast \hLambda^{(1)}_r \in \astTChn,
 \end{alignat}
 where $\hLambda^{(t)}_r \in \Chn(\Vt, \Wt)$.
 One can see that $T$-shot discrimination of quantum channels $\{ \hLambda_r \}_{r=0}^{R-1}$
 is a special case of this model with $\hLambda^{(t)}_r = \hLambda_r$.
 Another special case is quantum change point problems
 (see Refs.~\cite{Sen-Bag-Cal-Chi-2016,Sen-Mar-Mun-2017} in the case of $\hLambda^{(t)}_r$ being a state,
 i.e., $\Vt = \Complex$).
 In change point problems, a channel $\hLambda_0 \in \Chn(\V,\W)$ is prepared
 until some unspecified point $r$,
 after which another channel $\hLambda_1 \in \Chn(\V,\W)$ is prepared.
 We want to determine the change point $r$ as accurately as possible.
 This situation corresponds to the case in which
 $\Vt = \V$, $\Wt = \W$, $R = T + 1$, and $\hLambda^{(t)}_r = \hLambda_{\iota_r(t)}$
 $~(r \in \mI_R)$ hold, where $\iota_r(t) = 1$ for $t > r$, else $0$.
 A third special case is discrimination of the order in which the channels
 $\hLambda_1,\dots,\hLambda_T \in \Chn(\V,\W)$ are applied.
 Assume that each of the channels is applied once and only once; then,
 this situation corresponds to the case
 $\Vt = \V$, $\Wt = \W$, $R = T!$, and $\hLambda^{(t)}_r = \hLambda_{\gamma_r(t)}$,
 where $\gamma_r$ is the permutation on $\range{1}{T}$ determined by $r \in \mI_R$.
\end{ex}

\begin{ex}[Comparison of quantum channels] \label{ex:Comp}
 The second example is the problem of comparing quantum channels,
 which is an extension of quantum state comparison
 \cite{Ber-Hil-2005,Kle-Kam-Bru-2005,Pan-Wu-2011,Hay-Has-Hor-2018}
 and quantum measurement comparison \cite{Zim-Hei-Sed-2009}.
 Suppose that $K$ unknown quantum channels are given, each of which is
 randomly chosen from $L$ known channels $\hLambda_0,\dots,\hLambda_{L-1}$
 with the probabilities $u_0,\dots,u_{L-1}$.
 We want to determine whether they are identical or not.
 This problem is reduced to the problem of
 discriminating the following two channels
 \begin{alignat}{1}
  \tLambda_0 &\coloneqq p_0^{-1} \sum_{l=0}^{L-1} (u_l\hLambda_l)^{\ot K}, \nonumber \\
  \tLambda_1 &\coloneqq p_1^{-1} \left[ \left( \sum_{l=0}^{L-1} u_l\hLambda_l \right)^{\ot K}
  - \sum_{l=0}^{L-1} (u_l\hLambda_l)^{\ot K} \right],
 \end{alignat}
 where $p_0 \coloneqq \sum_{l=0}^{L-1} u_l^K$ and $p_1 \coloneqq 1 - p_0$ are the prior probabilities
 of $\tLambda_0$ and $\tLambda_1$.
\end{ex}

\begin{ex}[Discrimination of patterns]
 The third example is the problem of
 discriminating spatial and temporal patterns encoded in quantum channels.
 Assume that a comb
 \begin{alignat}{1}
  \hcE_x &\coloneqq \left[ \bigotimes_{k=1}^K \hLambda_{x^{(T)}_k} \right] \ast \cdots
  \ast \left[ \bigotimes_{k=1}^K \hLambda_{x^{(1)}_k} \right]
 \end{alignat}
 is given, where $\hLambda_0,\dots,\hLambda_{L-1} \in \Chn(\V,\W)$ are some channels
 and $\hcE_x$ is uniquely determined by a two-dimensional pattern
 $x \coloneqq \{ x^{(t)}_k \}_{(t,k)=(1,1)}^{(T,K)}$,
 each of whose entries $x^{(t)}_k$ is in $\mI_L$ (see Fig.~\ref{fig:pattern}).
 Also, assume that $x$ belongs to one of $R$ mutually exclusive subsets
 $\mX_0,\dots,\mX_{R-1}$ of $\mI_L^{TK}$.
 We want to determine which of $\mX_0,\dots,\mX_{R-1}$ the pattern $x$ belongs to.
 One can see this problem as the problem of discriminating $R$ channels
 $\{ \sum_{x \in \mX_r} p_x \hcE_x \}_{r=0}^{R-1}$,
 where $p_x$ is the prior probability of $\hcE_x$.
 This problem can be applied to various spatial and temporal patterns.
 The memoryless comb discrimination shown in Example~\ref{ex:memoryless_comb}
 can be seen as an example of this problem with $K \coloneqq 1$.
 One can easily see that quantum comb comparison, shown in Example~\ref{ex:Comp}
 is also an example of this problem,
 which corresponds to $\mX_0 \coloneqq
 \{ x \in \mI_L^{TK} : \exists l \in \mI_L, ~x^{(t)}_k = l ~(\forall t,k) \}$,
 $\mX_1 \coloneqq
 \{ x \in \mI_L^{TK} : x^{(1)}_k = x^{(2)}_k = \dots = x^{(T)}_k ~(\forall k) \} \setminus \mX_0$
 (where $\setminus$ is the set difference operation),
 and $p_x \coloneqq \prod_{k=0}^{K-1} u_{x^{(1)}_k}$.
 A third example is the problem of discriminating pulse-position modulated channels
 \cite{Zhu-Pir-2020-entangle},
 which corresponds to $R \coloneqq K$ and
 $\mX_r \coloneqq \{ x \in \mI_L^{TK} : x^{(t)}_k = \delta_{k,r+1} ~(\forall t,k) \}$
 $~(r \in \mI_R)$.
 A fourth example is the problem of determining whether $\hLambda_0$ has occurred or not,
 which corresponds to $R \coloneqq 2$,
 $\mX_0 \coloneqq \{ x \in \mI_L^{TK} : \exists t,k, ~x^{(t)}_k = 0 \}$,
 and $\mX_1 \coloneqq \mI_L^{TK} \setminus \mX_0$.
 \begin{figure}[bt]
  \centering
  \InsertPDF{1.0}{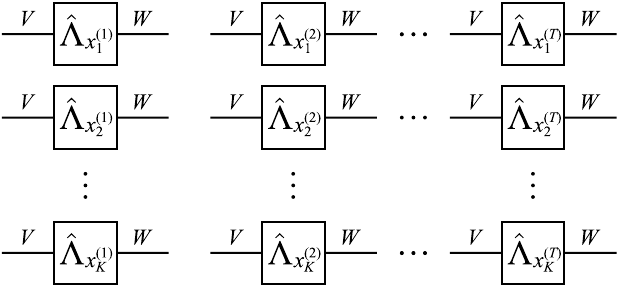}
  \caption{Spatial and temporal pattern
  $x \coloneqq \{ x^{(t)}_k \in \mI_L \}_{(t,k)=(1,1)}^{(T,K)}$
  encoded in quantum channels.
  $\hLambda_0,\dots,\hLambda_{L-1}$ are channels.}
  \label{fig:pattern}
 \end{figure}
\end{ex}

\subsection{Formulation}

\subsubsection{Unrestricted testers}

In this paper, to analyze a wide range of process discrimination problems,
we consider a problem written in the following form:
\begin{alignat}{1}
 \begin{array}{ll}
  \mbox{maximize} & \displaystyle \summ \braket{\hPhi_m, \hc_m} \\
  \mbox{subject~to} & \hPhi \in \mTGh, \quad
   \displaystyle \summ \braket{\hPhi_m, \ha_{j,m}} \le b_j ~(\forall j \in \mI_J), \\
 \end{array}
 \label{prob:PPi}
\end{alignat}
where $\{ \hc_m \}_{m=0}^{M-1}, \{ \ha_{j,m} \}_{(j,m)=(0,0)}^{(J-1,M-1)} \subset \astTHer$
and $\{ b_j \}_{j=0}^{J-1} \in \Real^J$ are constants determined by the problem.
$J$ is a nonnegative integer.
Problem~\eqref{prob:PPi_success} is obviously the special case of Problem~\eqref{prob:PPi}
with $M \coloneqq R$, $J \coloneqq 0$, and $\hc_m \coloneqq p_r \hLambda_r^{\ast T}$.
In the special case of $T = 1$ and $\V_1 = \Complex$,
it follows from $\astTHer \cong \Her_{\W_1}$ that
Problem~\eqref{prob:PPi} is the generalized quantum state discrimination problem
described in Ref.~\cite{Nak-Kat-Usu-2015-general}.
Throughout this paper, for simplicity of discussion,
in any optimization problem that maximizes (respectively, minimizes) an objective function,
the optimal value is set to $-\infty$ (respectively, $\infty$) if there is no feasible solution.

We often use
\begin{alignat}{2}
 c_m &\coloneqq \Choi_{\hc_m} \in \Her_\tV, &\quad
 a_{j,m} &\coloneqq \Choi_{\ha_{j,m}} \in \Her_\tV,
\end{alignat}
instead of $\hc_m$ and $\ha_{j,m}$,
which enables us to simplify the formulation of process discrimination problems.
Let
\begin{alignat}{1}
 \PG &\coloneqq \left\{ \Phi \in \mTG : \eta_j(\Phi) \le 0 ~(\forall j \in \mI_J)
 \right\},
 \label{eq:Testerc}
\end{alignat}
where
\begin{alignat}{1}
 \eta_j(\Phi) &\coloneqq \summ \braket{\Phi_m,a_{j,m}} - b_j \in \Real.
 \label{eq:eta}
\end{alignat}
Problem~\eqref{prob:PPi} is rewritten by the following SDP problem:
\begin{alignat}{1}
 \begin{array}{ll}
  \mbox{maximize} & P(\Phi) \coloneqq
   \displaystyle \summ \braket{\Phi_m,c_m} \\
  \mbox{subject~to} & \Phi \in \PG.
 \end{array}
 \tag{$\mathrm{P_G}$} \label{prob:PG}
\end{alignat}

\subsubsection{Restricted testers}

We are often concerned with a process discrimination problem in which
the available testers are restricted to belong to a certain subset of all possible testers
in quantum mechanics.
Very recently, a general formulation of restricted problems of finding minimum-error testers
has been discussed in Ref.~\cite{Nak-2021-restricted}.
For examples of such restricted problems, the reader can refer to Ref.~\cite{Nak-2021-restricted}.
We will extend this work to a broad class of optimization criteria.
We impose the additional constraint $\Phi \in \mT$,
where $\mT$ is a nonempty convex subset of $\mTG$.
This problem is formulated as%
\footnote{We do not assume that $\mT$ is closed,
which is inspired by the fact there exists an important subset of all possible testers that is not closed,
e.g., the set of local operations and classical communication \cite{Chi-Cui-Lo-2012}.
While an optimal solution to Problem~\eqref{prob:P} may not exist,
its optimal value, $\sup_{\Phi \in \P} P(\Phi)$, is always uniquely determined.}
\begin{alignat}{1}
 \begin{array}{ll}
  \mbox{maximize} & P(\Phi) \\
  \mbox{subject~to} & \Phi \in \P,
 \end{array}
 \tag{$\mathrm{P}$} \label{prob:P}
\end{alignat}
where $\P \coloneqq \PG \cap \mT$, i.e.,
\begin{alignat}{1}
 \P &\coloneqq \left\{ \Phi \in \mT : \eta_j(\Phi) \le 0 ~(\forall j \in \mI_J) \right\}.
 \label{eq:TestercC}
\end{alignat}
Problem~\eqref{prob:PG} can be viewed as the special case of Problem~\eqref{prob:P}
with $\mT \coloneqq \mTG$.
Problem~\eqref{prob:P} is not an SDP problem in general,
but is a convex problem since $\P$ is convex.
The assumption of the convexity of $\mT$ implies that
any probabilistic mixture of any pair of testers $\Phi^{(1)},\Phi^{(2)} \in \mT$,
$\{ p \Phi^{(1)}_m + (1-p) \Phi^{(2)}_m \}_{m=0}^{M-1}$
$~(\forall 0 < p < 1)$, is in $\mT$.
In this paper, we also assume
\begin{alignat}{1}
 \ol{\P} &= \left\{ \Phi \in \ol{\mT} : \eta_j(\Phi) \le 0 ~(\forall j \in \mI_J) \right\}.
 \label{eq:TestercC2}
\end{alignat}
If $\mT$ is closed, then Eq.~\eqref{eq:TestercC2} always holds.
These assumptions hold in many practical situations.
Let us choose a closed convex cone $\mC$ and a closed convex set $\S$ such that%
\footnote{Such $\mC$ and $\S$ always exist.
Indeed,
$\mC \coloneqq \{ \{ p \Phi_m \}_{m=0}^{M-1} : p \in \Real_+, \Phi \in \ol{\mT} \}$ and
$\S \coloneqq \{ \summ \Phi_m : \Phi \in \ol{\mT} \}$
satisfy Eq.~\eqref{eq:T}.}
\begin{alignat}{1}
 \ol{\mT} &= \left\{ \Phi \in \mC : \summ \Phi_m \in \S \right\},
 \quad \mC \subseteq \mCG, \quad \S \subseteq \SG.
 \label{eq:T}
\end{alignat}
Equation~\eqref{eq:TG} is the special case of Eq.~\eqref{eq:T} with
$\mC = \mCG$ and $\S = \SG$.
Note that if the feasible set $\P$ is not empty, then
at least one optimal solution exists.

\subsubsection{Examples}

We provide three simple examples of Problem~\eqref{prob:P}.
For more information, see Sec.~II of Ref.~\cite{Nak-Kat-Usu-2015-general},
which provides several other examples in the case of state discrimination.
\begin{ex}[Optimal inconclusive discrimination] \label{ex:inconclusive}
 The first example is the problem of finding optimal inconclusive discrimination
 of quantum combs.
 This is an extension of the problem of finding optimal inconclusive state discrimination
 \cite{Iva-1987,Die-1988,Per-1988}.
 In this problem, we want to discriminate $R$ combs
 $\hcE_0,\dots,\hcE_{R-1} \in \astTChn$
 with maximum average success probability subject to the constraint that
 the average inconclusive probability is equal to a constant value $\pinc$ with $0 \le \pinc \le 1$.
 We try to find an optimal tester $\hPhi \coloneqq \{ \hPhi_m \}_{m=0}^{M-1} \in \mTGh$
 with $M \coloneqq R + 1$.
 The element $\hPhi_r$ with $r < R$ corresponds to the identification
 of the comb $\hcE_r$, whereas $\hPhi_R$ corresponds to the inconclusive answer.
 The average success and inconclusive probabilities are, respectively, written as
 \begin{alignat}{2}
  \PS(\hPhi) &\coloneqq \sumr p_r \mathrm{Pr}(r|\hcE_r), &\quad
  \PI(\hPhi) &\coloneqq \sumr p_r \mathrm{Pr}(R|\hcE_r),
 \end{alignat}
 where $\mathrm{Pr}(m|\hcE_r) \coloneqq \braket{\hPhi_m, \hcE_r}$
 and $p_r$ is the prior probability of the comb $\hcE_r$.
 Thus, the problem is formulated as
 \begin{alignat}{1}
  \begin{array}{ll}
   \mbox{maximize} & \PS(\hPhi) \\
   \mbox{subject~to} & \hPhi \in \mTGh, ~\PI(\hPhi) = \pinc. \\
  \end{array}
  \tag{$\mathrm{P_{inc}}$} \label{prob:PPiInc}
 \end{alignat}
 From Eq.~\eqref{eq:Phi_c}, we have $\mathrm{Pr}(m|\hcE_r) = \braket{\Phi_m, \cE_r}$.
 The optimal value of the problem does not change if we replace
 the constraint $\PI(\hPhi) = \pinc$ by $\PI(\hPhi) \ge \pinc$;
 indeed, in this case, we can easily verify that any optimal solution $\hPhi$
 must satisfy $\PI(\hPhi) = \pinc$.
 Therefore, this problem is rewritten as Problem~\eqref{prob:PG} with
 \begin{alignat}{1}
  M &\coloneqq R + 1, \nonumber \\
  J &\coloneqq 1, \nonumber \\
  c_m &\coloneqq
  \begin{dcases}
   p_m \cE_m, & m < R, \\
   \zero, & m = R, \\
  \end{dcases} \nonumber \\
  a_{0,m} &\coloneqq
  \begin{dcases}
   \zero, & m < R, \\
   - \sumr p_r \cE_r, & m = R, \\
  \end{dcases} \nonumber \\
  b_0 &\coloneqq - \pinc.
  \label{eq:inc_abc}
 \end{alignat}
 $\PG$ is not empty for any $0 \le \pinc \le 1$.
 In the case of $T = 1$ and $\V_1 = \Complex$,
 this problem reduces to the SDP problem given by Ref.~\cite{Eld-2003-inc}.

 In the special case of $\pinc = 0$, Problem~\eqref{prob:PPiInc} is
 equivalent to the problem of finding minimum-error discrimination,
 i.e., $\hPhi \in \mTGh$ that maximizes $\PS(\hPhi)$,
 in which case, without loss of generality, we can assume $\hPhi_R = \zero$.
 Thus, this problem is written as Problem~\eqref{prob:PG} with
 $M \coloneqq R$, $J \coloneqq 0$, and $c_r \coloneqq p_r \cE_r$ $~(r \in \mI_R)$.

 In another special case in which $\pinc$ is sufficiently large,
 the average error probability, $\PE(\hPhi) \coloneqq 1 - \PS(\hPhi) - \PI(\hPhi)$,
 of an optimal solution becomes zero.
 Unambiguous (or error-free) discrimination, which satisfies $\PE(\hPhi) = 0$,
 is called optimal if it maximizes the average success probability
 (or, equivalently, minimizes the average inconclusive probability).
 The problem of finding optimal unambiguous discrimination can be
 formulated as
 \begin{alignat}{1}
  \begin{array}{ll}
   \mbox{maximize} &
    \displaystyle \PS(\hPhi) - \lim_{\kappa \to \infty} \kappa \PE(\hPhi) \\
   \mbox{subject~to} & \hPhi \in \mTGh. \\
  \end{array}
  \tag{$\mathrm{P_{unamb}}$} \label{prob:PPiUnamb}
 \end{alignat}
 One can easily verify that an optimal solution satisfies $\PE(\hPhi) = 0$.
 Problem~\eqref{prob:PPiUnamb} is rewritten as Problem~\eqref{prob:PG} with
 $M \coloneqq R + 1$, $J \coloneqq 0$,
 $c_r \coloneqq p_r \cE_r - \kappa \sum_{r' \neq r} p_{r'} \cE_{r'}$ $~(r \in \mI_R)$,
 $c_R \coloneqq \zero$, and $\kappa \to \infty$.
 Note that this problem is also formulated as
 \begin{alignat}{1}
  \begin{array}{ll}
   \mbox{maximize} &
    \displaystyle \PS(\hPhi) \\
   \mbox{subject~to} & \hPhi \in \mTGh, ~\PS(\hPhi) + \PI(\hPhi) = 1, \\
  \end{array}
 \end{alignat}
 which is rewritten by Problem~\eqref{prob:PG} with
 \begin{alignat}{1}
  M &\coloneqq R + 1, \nonumber \\
  J &\coloneqq 1, \nonumber \\
  c_m &\coloneqq
  \begin{dcases}
   p_m \cE_m, & m < R, \\
   \zero, & m = R, \\
  \end{dcases} \nonumber \\
  a_{0,m} &\coloneqq
  \begin{dcases}
   - p_m \cE_m, & m < R, \\
   - \sumr p_r \cE_r, & m = R, \\
  \end{dcases} \nonumber \\
  b_0 &\coloneqq - 1.
  \label{eq:unamb_abc}
 \end{alignat}
 In the case of $T = 1$ and $\V_1 = \Complex$,
 this problem reduces to the SDP problem given by Ref.~\cite{Eld-Sto-Has-2004}.
\end{ex}
\begin{ex}[Neyman-Pearson strategy]
 The second example is an optimal process discrimination problem under the Neyman-Pearson criterion,
 whose state discrimination version has been extensively investigated
 \cite{Hel-1976,Hol-1982-Prob,Par-1997}.
 Let us consider the problem of discriminating two combs $\hcE_0$ and $\hcE_1$.
 This criterion attempts to maximize the detection probability $\mathrm{Pr}(1|\hcE_1)$
 while the false-alarm probability $\mathrm{Pr}(1|\hcE_0)$ is less than or equal to
 a constant value $\pNP$ with $0 \le \pNP \le 1$,
 where $\mathrm{Pr}(m|\hcE_r) \coloneqq \braket{\hPhi_m, \hcE_r}$.
 This problem can be formulated as
 \begin{alignat}{1}
  \begin{array}{ll}
   \mbox{maximize} &
    \displaystyle \mathrm{Pr}(1|\hcE_1) \\
   \mbox{subject~to} & \hPhi \in \mTGh, ~\mathrm{Pr}(1|\hcE_0) \le \pNP, \\
  \end{array}
  \tag{$\mathrm{P_{NP}}$} \label{prob:PPiNP}
 \end{alignat}
 which is rewritten by Problem~\eqref{prob:PG} with
 \begin{alignat}{3}
  M &\coloneqq 2, &\quad J &\coloneqq 1, &\quad
  c_m &\coloneqq \delta_{m,1} \cE_1, \nonumber \\
  a_{0,m} &\coloneqq \delta_{m,1} \cE_0, &\quad b_0 &\coloneqq \pNP.
  \label{eq:NP_abc}
 \end{alignat}
 $\PG$ is not empty for any $0 \le \pNP \le 1$.
\end{ex}

\begin{ex}[Restricted testers] \label{ex:restricted_testers}
 We can consider a process discrimination problem under the inconclusive and Neyman-Pearson
 strategies in which testers are restricted to belong to a subset of $\mTh$ of $\mTGh$.
 Let us consider the former case.
 This problem is formulated as%
 \footnote{We should note that Problem~\eqref{prob:PPiIncR}
 is not exactly equivalent to Problem~\eqref{prob:PPiInc} with $\mTGh$ replaced by $\mTh$.
 Indeed, any $\hPhi \in \mTh$ may satisfy $\PI(\hPhi) > \pinc$,
 in which case there is no feasible solution to the latter problem.
 However, the latter problem can also be formulated in the form of Problem~\eqref{prob:P}
 since $\PI(\hPhi) = \pinc$ is equivalent to $\PI(\hPhi) \ge \pinc$ and $\PI(\hPhi) \le \pinc$.}
 \begin{alignat}{1}
  \begin{array}{ll}
   \mbox{maximize} & \PS(\hPhi) \\
   \mbox{subject~to} & \hPhi \in \mTh, ~\PI(\hPhi) \ge \pinc, \\
  \end{array}
  \label{prob:PPiIncR}
 \end{alignat}
 which is rewritten as Problem~\eqref{prob:P} with Eq.~\eqref{eq:inc_abc} and
 \begin{alignat}{1}
  \mT &\coloneqq \left\{ \{ \tChoi_{\hPhi_m} \}_{m=0}^{M-1} : \hPhi \in \mTh \right\}.
  \label{eq:bTh2bT}
 \end{alignat}

 As a concrete example, let us assume that testers are restricted to the form of
 Fig.~\ref{fig:tester_restricted}.
 Such a tester, consisting of two sequentially connected single-shot testers,
 is interpreted as a tester performed by Alice and Bob
 in which only one-way classical communication from Alice to Bob is allowed.
 Specifically, in such a tester, Alice prepares a state $\hrho_\mathrm{A}$,
 performs a measurement $\{\hPi^\mathrm{A}_i\}_i$,
 and sends her outcome $i$ to Bob.
 Based on her result $i$, Bob then prepares a state $\hrho^{(i)}_\mathrm{B}$ and
 performs a measurement $\{\hPi^{(i)}_m\}_m$.
 It is seen that $\mT$ satisfies Eq.~\eqref{eq:T} with
 \begin{alignat}{1}
  \mC &\coloneqq \left\{ \left\{ \sum_i B^{(i)}_r \ot A_i \right\}_{r=0}^R :
  A_i \in \Pos_{\W_1 \ot \V_1}, ~\{ B^{(i)}_r \}_r \in \Tester_{\W_2,\V_2} \right\},
  \label{eq:C_chsequential}
 \end{alignat}
 and $\S \coloneqq \SG$, where $\Tester_{\W_2,\V_2}$ is the set of all testers
 $\{B_r\}_{r=0}^R \subset \Pos_{\W_2 \ot \V_2}$ with $R+1$ outcomes
 [i.e., $\{B_r\}_r$ satisfies $\sum_{r=0}^R B_r = \I_{\W_2} \ot \rho$
 for some $\rho \in \Den_{\V_2}$].
 \begin{figure}[bt]
  \centering
  \InsertPDF{1.0}{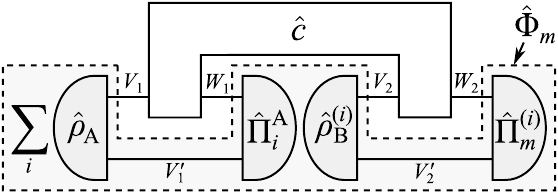}
  \caption{Process discrimination with two sequentially connected single-shot testers.
  The tester $\{ \hPhi_m \}_m$ consits of
  $\hrho_\mathrm{A} \in \Den_{\V_1 \ot \V'_1}$,
  $\{ \hPi^\mathrm{A}_i \}_i$,
  $\hrho^{(i)}_\mathrm{B} \in \Den_{\V_2 \ot \V'_2}$,
  and $\{ \hPi^{(i)}_m \}_m \in \POVM_{\W_2 \ot \V'_2}$,
  where $\{ \hPi^\mathrm{A}_i \}_i$ is a measurement of $\W_1 \ot \V'_1$.}
  \label{fig:tester_restricted}
 \end{figure}

 Similarly, in the case of the Neyman-Pearson strategy, the problem is written as
 Problem~\eqref{prob:PPiNP} with $\mTGh$ replaced by $\mTh$, i.e.,
 \begin{alignat}{1}
  \begin{array}{ll}
   \mbox{maximize} &
    \displaystyle \mathrm{Pr}(1|\hcE_1) \\
   \mbox{subject~to} & \hPhi \in \mTh, ~\mathrm{Pr}(1|\hcE_0) \le \pNP, \\
  \end{array}
  \tag{$\mathrm{P_{NP}}$} \label{prob:PPiNPR}
 \end{alignat}
 which is also formulated as Problem~\eqref{prob:P} with Eq.~\eqref{eq:NP_abc}
 and $\mT$ of Eq.~\eqref{eq:bTh2bT}.
\end{ex}

\section{Optimal solutions to process discrimination problems} \label{sec:opt}

In this section, we derive the Lagrange dual problem of Problem~\eqref{prob:P}
that has no duality gap.
Also, necessary and sufficient conditions for a tester to be optimal are given.
We also give necessary and sufficient conditions that the optimal value remain unchanged
even when a certain additional constraint is imposed.
These results are useful for obtaining analytical and/or numerical optimal solutions.

\subsection{Dual problems}

The following theorem holds (proved in Appendix~\ref{append:dual})%
\footnote{A diagrammatic representation of dual problems
(in the minimum-error case) can be seen in Ref.~\cite{Nak-2020-diagram},
which allows us to gain an intuitive understanding of an operational interpretation.}.
\begin{thm} \label{thm:dual}
 Assume that Problem~\eqref{prob:P} is given.
 Let $\mC$ and $\S$ be a closed convex cone and a closed convex set satisfying Eq.~\eqref{eq:T}.
 The optimal value of Problem~\eqref{prob:P} coincides with that of
 the following optimization problem:
 \begin{alignat}{1}
  \begin{array}{ll}
   \mbox{minimize} & \displaystyle D_\S(\chi, q) \coloneqq \lambdaS(\chi) + \sumj q_j b_j \\
   \mbox{subject~to} & (\chi, q) \in \D \\
  \end{array}
  \tag{$\mathrm{D}$} \label{prob:D}
 \end{alignat}
 with $\chi \in \Her_\tV$ and $q \coloneqq \{ q_j \}_{j=0}^{J-1} \in \Real_+^J$,
 where
 \begin{alignat}{1}
  \lambdaS(\chi) &\coloneqq \sup_{\varphi \in \S} \braket{\varphi, \chi}, \nonumber \\
  \D &\coloneqq \left\{ (\chi, q) \in \Her_\tV \times \Real_+^J :
  \{\chi - z_m(q)\}_{m=0}^{M-1} \in \mC^* \right\}, \nonumber \\
  z_m(q) &\coloneqq c_m - \sumj q_j a_{j,m} \in \Her_\tV.
  \label{eq:z}
 \end{alignat}
\end{thm}

One can easily see that Problem~\eqref{prob:D},
which is the Lagrange dual problem of Problem~\eqref{prob:P},
is a convex problem.
Problem~\eqref{prob:D} is often easier to solve than Problem~\eqref{prob:P}.
Note that $\{ y_m \in \Her_\tV \}_{m=0}^{M-1} \in \mC^*$ is equivalent to
\begin{alignat}{1}
 \summ \braket{\Phi_m,y_m} &\ge 0, \quad \forall \Phi \in \mC.
\end{alignat}
It is easily seen that the function $\lambdaS$ is convex and positively homogeneous of degree 1
[i.e., $\lambdaS(r \chi) = r \lambdaS(\chi)$ holds for any $r \in \Real_+$
and $\chi \in \Her_\tV$].
From Eq.~\eqref{eq:Phim_c_1}, we have
\begin{alignat}{2}
 \braket{\varphi,\chi} &= \lambdaSG(\chi) = \lambdaS(\chi),
 &\quad \forall \varphi \in \SG, ~\chi \in \Lin(\Chn_\tV).
 \label{eq:LinChn_lambda}
\end{alignat}

As a special case of Problem~\eqref{prob:D}, the dual of Problem~\eqref{prob:PG} is given by
\begin{alignat}{1}
 \begin{array}{ll}
  \mbox{minimize} & \displaystyle D_\SG(\chi, q) \\
  \mbox{subject~to} & (\chi, q) \in \DG \\
 \end{array}
 \tag{$\mathrm{D_G}$} \label{prob:DG}
\end{alignat}
with $(\chi, q)$, where
\begin{alignat}{1}
 \DG &\coloneqq \left\{ (\chi, q) \in \Her_\tV \times \Real_+^J :
 \chi \ge z_m(q) ~(\forall m \in \mI_M) \right\}.
 \label{eq:DG}
\end{alignat}
Theorem~\ref{thm:dual} immediately yields that the optimal values of Problems~\eqref{prob:PG} and
\eqref{prob:DG} coincide.

Note that one can consider two any sets $\mC$ $~(\subseteq \mCG)$ and
$\S$ $~(\subseteq \SG)$ such that
\begin{alignat}{1}
 \ol{\mT} &= \left\{ \Phi \in \clconi \mC : \summ \Phi_m \in \clco \S \right\},
 \label{eq:Tex}
\end{alignat}
instead of Eq.~\eqref{eq:T}.
In this case, one can easily verify
$\lambda_{\clco \S}(\chi) = \lambdaS(\chi)$ and $(\clconi \mC)^* = \mC^*$,
which indicates that Theorem~\ref{thm:dual} works without any changes.
In what follows, for simplicity, we assume that
$\mC$ and $\S$ are, respectively, a closed convex cone and a closed convex set.

\begin{ex}[Optimal inconclusive discrimination]
 By substituting Eq.~\eqref{eq:inc_abc} into Problem~\eqref{prob:DG},
 the dual of Problem~\eqref{prob:PPiInc} is immediately obtained as
 \begin{alignat}{1}
  \begin{array}{ll}
   \mbox{minimize} & \lambdaSG(\chi) - q \pinc \\
   \mbox{subject~to} & \chi \ge p_r \cE_r ~(\forall r \in \mI_{R+1})
  \end{array}
  \tag{$\mathrm{D_{inc}}$} \label{prob:DInc}
 \end{alignat}
 with $\chi \in \Her_\tV$ and $q \in \Real_+$,
 where $p_R \coloneqq q$ and $\cE_R \coloneqq \sumr p_r \cE_r$.
 Any feasible solution $\chi$ is in $\Pos_\tV$.
 It is easily seen that there exists an optimal solution $(\chi,q)$ such that $q \le 1$%
 \footnote{Proof: Arbitrarily choose $q > 1$; then,
 since $q \cE_R \ge \cE_R \ge p_r \cE_r$ holds for each $r \in \mI_R$,
 $(\chi,q) \in \DG$ is equivalent to $\chi \ge q \cE_R$.
 This gives $(\cE_R,1) \in \DG$.
 Arbitrarily choose $\chi$ satisfying $(\chi,q) \in \DG$;
 then, it suffices to show $D_\SG(\chi,q) \ge D_\SG(\cE_R,1)$.
 From Eq.~\eqref{eq:Phim_c_1}, we have $\lambdaSG(\cE_R) = 1$.
 Thus, $D_\SG(\chi,q) = \lambdaSG(\chi) - q \pinc
 \ge \lambdaSG(q \cE_R) - q \pinc = q (1 - \pinc)
 \ge 1 - \pinc = D_\SG(\cE_R,1)$ holds.}.
 In the special case of $\pinc = 0$, which corresponds to the minimum-error strategy,
 the dual problem is written as
 \begin{alignat}{1}
  \begin{array}{ll}
   \mbox{minimize} & \lambdaSG(\chi) \\
   \mbox{subject~to} & \chi \ge p_r \cE_r ~(\forall r \in \mI_R).
  \end{array}
  \label{prob:D_success}
 \end{alignat}
 Also, the dual of Problem~\eqref{prob:PPiUnamb} is
 \begin{alignat}{1}
  \begin{array}{ll}
   \mbox{minimize} & \lambdaSG(\chi) \\
   \mbox{subject~to} & \displaystyle \lim_{\kappa \to \infty}
    \left( \chi - p_r \cE_r +
     \kappa \sum_{\substack{r'=0 \\ r' \neq r}}^{R-1} p_{r'} \cE_{r'} \right) \ge \zero
    ~(\forall r \in \mI_R). \\
  \end{array}
 \end{alignat}
 Note that this constraint is rewritable as
 $\tTheta_r (\chi - p_r \cE_r) \tTheta_r \ge \zero$ $~(\forall r \in \mI_R)$,
 where $\tTheta_r$ is the orthogonal projection matrix onto the null space of
 $\sum_{r' \neq r} p_{r'} \cE_{r'}$.
\end{ex}

\begin{ex}[Neyman-Pearson strategy]
 By substituting Eq.~\eqref{eq:NP_abc} into Problem~\eqref{prob:DG},
 the dual of Problem~\eqref{prob:PPiNP} is obtained as
 \begin{alignat}{1}
  \begin{array}{ll}
   \mbox{minimize} &
    \displaystyle \lambdaSG(\chi) + q \pNP \\
   \mbox{subject~to} & \chi \ge \cE_1 - q \cE_0
  \end{array}
 \end{alignat}
 with $\chi \in \Pos_\tV$ and $q \in \Real_+$.
\end{ex}

\begin{ex}[Restricted testers]
 By substituting Eq.~\eqref{eq:inc_abc} into Problem~\eqref{prob:D},
 the dual of Problem~\eqref{prob:PPiIncR} is obtained as
 \begin{alignat}{1}
  \begin{array}{ll}
   \mbox{minimize} & \lambdaS(\chi) - q \pinc \\
   \mbox{subject~to} & \{ \chi - p_r \cE_r \}_{r=0}^R \in \mC^*
  \end{array}
  \label{prob:DIncR}
 \end{alignat}
 with $\chi \in \Her_\tV$ and $q \in \Real_+$,
 where $p_R \coloneqq q$ and $\cE_R \coloneqq \sumr p_r \cE_r$.
 In the special case of $\mC$ given by Eq.~\eqref{eq:C_chsequential},
 Problem~\eqref{prob:DIncR} is rewritten by
 \begin{alignat}{1}
  \begin{array}{ll}
   \mbox{minimize} & \lambdaSG(\chi) - q \pinc \\
   \mbox{subject~to} & \displaystyle \Trp{\W_2 \ot \V_2}
    \left[ \sum_{r=0}^R B_r (\chi - p_r \cE_r) \right] \ge \zero \\
   &\quad (\forall \{B_r\}_{r=0}^R \in \Tester_{\W_2,\V_2}). \\
  \end{array}
  \label{prob:D_chsequential}
 \end{alignat}
\end{ex}

We can show the following Proposition (proved in Appendix~\ref{append:chi_comb}).
\begin{proposition} \label{pro:chi_comb_feasible}
 For any $(\chi',q) \in \D$, there exists $(\chi, q) \in \D$
 satisfying $\lambdaSG(\chi) = \lambdaSG(\chi')$ and $\chi \in \Lin(\Chn_\tV)$.
\end{proposition}
This proposition immediately yields the following corollary (proof omitted).
\begin{cor} \label{cor:chi_comb}
 If $\S = \SG$ holds, then
 for any optimal solution $(\chi',q)$ to Problem~\eqref{prob:D},
 there also exists an optimal solution $(\chi, q)$ to Problem~\eqref{prob:D}
 satisfying $\chi \in \Lin(\Chn_\tV)$.
\end{cor}

\subsection{Conditions for optimality}

The following theorem provides necessary and sufficient conditions for
a tester to be optimal for Problem~\eqref{prob:P}
(proved in Appendix~\ref{append:Phi_nas}).

\begin{thm} \label{thm:Phi_nas}
 $\Phi \in \P$ and $(\chi, q) \in \D$ are, respectively, optimal for
 Problems~\eqref{prob:P} and \eqref{prob:D} if and only if they satisfy
 \begin{alignat}{2}
  q_j \eta_j(\Phi) &= 0, &\quad &\forall j \in \mI_J, \nonumber \\
  \summ \braket{\Phi_m, \chi - z_m(q)} &= 0, \nonumber \\
  \summ \braket{\Phi_m, \chi} &= \lambdaS(\chi).
  \label{eq:channel_nas}
 \end{alignat}
\end{thm}

We consider the case $\mC = \mCG$; then, since $\chi \ge z_m(q)$ holds,
the second line of Eq.~\eqref{eq:channel_nas} is equivalent to
\begin{alignat}{1}
 [\chi - z_m(q)]\Phi_m &= \zero, \quad \forall m \in \mI_M,
 \label{eq:channel_nas_mCG}
\end{alignat}
which follows from $XY = \zero \Leftrightarrow \braket{Y,X} = 0$
for any $X,Y \in \Pos_\tV$.
Moreover, let us consider $\Phi \in \P$ such that $\summ \Phi_m$ is of full rank.
Let
\begin{alignat}{1}
 \chi^\Phi(q) &\coloneqq \left[ \summ z_m(q) \Phi_m \right]
 \left( \summ \Phi_m \right)^{-1};
 \label{eq:chi_Phi}
\end{alignat}
then, it follows that $\chi = \chi^\Phi(q)$ holds for any $(\chi, q) \in \D$
satisfying Eq.~\eqref{eq:channel_nas_mCG}.
This immediately yields the following two corollaries.

\begin{cor} \label{cor:chi}
 Let us consider Problem~\eqref{prob:P} with $\mC = \mCG$.
 Assume that there exists an optimal solution $\Phi$ such that $\summ \Phi_m$ is of full rank.
 Then, any optimal solution $(\chi, q)$ to Problem~\eqref{prob:D} satisfies
 $\chi = \chi^\Phi(q)$.
 If, in addition, $\S = \SG$ (i.e., $\mT = \mTG$) holds,
 then $\chi^\Phi(q) \in \Lin(\Chn_\tV)$ holds.
\end{cor}
\begin{proof}
 From Theorem~\ref{thm:Phi_nas},
 any optimal solution $(\chi, q)$ to Problem~\eqref{prob:D} satisfies
 Eq.~\eqref{eq:channel_nas_mCG},
 which gives $\chi = \chi^\Phi(q)$.
 In the case of $\S = \SG$, from Corollary~\ref{cor:chi_comb},
 there exists an optimal solution $(\chi', q)$ to Problem~\eqref{prob:D}
 such that $\chi'$ is in $\Lin(\Chn_\tV)$.
 Again from Theorem~\ref{thm:Phi_nas}, we have $\chi' = \chi^\Phi(q)$.
\end{proof}

\begin{cor} \label{cor:Phi_nas_full}
 Let us consider Problem~\eqref{prob:P} with $\mC = \mCG$.
 Assume that there exists an optimal solution to Problem~\eqref{prob:D}.
 Arbitrarily choose $\Phi \in \P$ such that $\summ \Phi_m$ is of full rank;
 then, $\Phi$ is optimal for Problem~\eqref{prob:P} if and only if
 there exists $q \in \Real_+^J$ such that
 \begin{alignat}{2}
  q_j \eta_j(\Phi) &= 0, &\quad &\forall j \in \mI_J, \nonumber \\
  \chi^\Phi(q) &\ge z_m(q), &\quad &\forall m \in \mI_M, \nonumber \\
  \summ \braket{\Phi_m, \chi^\Phi(q)} &= \lambdaS[\chi^\Phi(q)].
  \label{eq:channel_nas2}
 \end{alignat}
\end{cor}
\begin{proof}
 ``If'':
 Let $\chi \coloneqq \chi^\Phi(q)$; then,
 $(\chi, q) \in \D$ holds from the second line of Eq.~\eqref{eq:channel_nas2}.
 $\chi^\Phi(q) \in \Her_\tV$ obviously holds from $\chi^\Phi(q) \ge z_m(q)$.
 From Theorem~\ref{thm:Phi_nas}, it suffices to show Eq.~\eqref{eq:channel_nas}.
 The first and third lines of Eq.~\eqref{eq:channel_nas} obviously hold.
 $\summ [\chi^\Phi(q) - z_m(q)] \Phi_m = \zero$ holds from Eq.~\eqref{eq:chi_Phi}.
 Taking the trace of this equation yields the second line of Eq.~\eqref{eq:channel_nas}.

 ``Only if'':
 Let $(\chi, q)$ be an optimal solution to Problem~\eqref{prob:D};
 then, $\chi = \chi^\Phi(q)$ holds from Corollary~\ref{cor:chi}.
 Thus, Eq.~\eqref{eq:channel_nas2} holds from Theorem~\ref{thm:Phi_nas} and
 $(\chi, q) \in \D$.
\end{proof}

\begin{ex}[Optimal inconclusive discrimination]
 We can show, by substituting Eq.~\eqref{eq:inc_abc} into Eq.~\eqref{eq:channel_nas},
 that necessary and sufficient conditions for $\Phi \in \PG$ and $(\chi, q) \in \DG$
 to be, respectively, optimal for Problems~\eqref{prob:PPiInc} and \eqref{prob:DInc} are
 \begin{alignat}{2}
  q (\braket{\Phi_R, \cE_R} - \pinc) &= 0, \nonumber \\
  (\chi - p_r \cE_r) \Phi_r &= \zero, &\quad &\forall r \in \mI_{R+1}, \nonumber \\
  \sum_{r=0}^R \braket{\Phi_r, \chi} &= \lambdaSG(\chi),
  \label{eq:channel_nas_inc}
 \end{alignat}
 where $p_R \coloneqq q$ and $\cE_R \coloneqq \sumr p_r \cE_r$.
 It is easily seen that the first line is rewritten by
 $\PI(\hPhi) ~(= \braket{\Phi_R, \cE_R}) = \pinc$.
 Also, Corollary~\ref{cor:Phi_nas_full} gives that,
 for any $\Phi \in \PG$ such that $\sum_{r=0}^R \Phi_r$ is of full rank,
 $\Phi$ is optimal for Problem~\eqref{prob:PPiInc} if and only if
 there exists $q \in \Real_+$ such that
 \begin{alignat}{2}
  \braket{\Phi_R, \cE_R} &= \pinc, \nonumber \\
  \chi^\Phi(q) &\ge p_r \cE_r, &\quad &\forall r \in \mI_{R+1}, \nonumber \\
  \sum_{r=0}^R \braket{\Phi_r, \chi^\Phi(q)} &= \lambdaSG[\chi^\Phi(q)]
 \end{alignat}
 [recall that an optimal solution to Problem~\eqref{prob:DInc} always exists].
 In the special case of $T = 1$ and $\V_1 = \Complex$,
 in which case each $\cE_r$, denoted by $\rho_r$, is a quantum state,
 necessary and sufficient conditions for $\Phi \in \PG \subseteq \POVM_{\W_1}$ and $(\chi, q) \in \DG$
 to be optimal are
 \begin{alignat}{2}
  \braket{\Phi_R, \rho_R} &= \pinc, \nonumber \\
  (\chi - p_r \rho_r) \Phi_r &= \zero, &\quad &\forall r \in \mI_{R+1},
  \label{eq:channel_nas_inc_state}
 \end{alignat}
 where $p_R \coloneqq q$ and $\rho_R \coloneqq \sumr p_r \rho_r$.
 The third line of Eq.~\eqref{eq:channel_nas} always holds
 from $\sum_{r=0}^R \braket{\Phi_r, \chi} = \Tr \chi = \lambdaSG(\chi)$
 [note that $\summ \Phi_m = \I_{\W_1}$ holds for any $\Phi \in \POVM_{\W_1}$].
 Also, from Corollary~\ref{cor:Phi_nas_full},
 $\Phi \in \PG$ is optimal if and only if
 \begin{alignat}{2}
  \braket{\Phi_R, \rho_R} &= \pinc, \nonumber \\
  \chi^\Phi(q) &\ge p_r \rho_r, &\quad &\forall r \in \mI_{R+1}
  \label{eq:channel_nas2_inc_state}
 \end{alignat}
 holds, where $\chi^\Phi(q)$ of Eq.~\eqref{eq:chi_Phi} is written as
 $\chi^\Phi(q) \coloneqq \sum_{r=0}^R p_r \rho_r \Phi_r = \sumr p_r \rho_r (\Phi_r + q \Phi_R)$.
\end{ex}

Table~\ref{tab:summary} summarizes the formulation
of the process discrimination problems.
The general formulation and the cases of optimal inconclusive discrimination of quantum combs
$\{ \hcE_r \}_{r=0}^{R-1} \subset \astTChn$
and quantum states $\{ \rho_r \}_{r=0}^{R-1} \subset \Den_\W$, respectively,
with $\mT = \mTG$ are shown.
In these examples, $\{ p_r \}_{r=0}^{R-1}$ is the prior probabilities.
\begin{table*}[tbp]
 \caption{Formulation of the generalized process discrimination problems.}
 \label{tab:summary}
 \begin{center}
  \tabcolsep = 0.4em
  \begin{tabular}{p{0.5em}lll}
   \hline \hline
   & \multicolumn{1}{c}{Primal problems}
       & \multicolumn{1}{c}{Dual problems}
           & \multicolumn{1}{c}{
               \begin{minipage}[t]{5cm}
                Necessary and sufficient conditions for \\
                $\Phi \in \P$ and $(\chi,q) \in \D$ to be optimal
               \end{minipage}} \\ \hline
   \multicolumn{4}{l}{\rule[0pt]{0pt}{10pt}Basic formulation} \vspace{-1.6em} \\
   & \begin{minipage}[t]{5cm}
      \begin{eqnarray*}
       \begin{array}{ll}
        \mbox{maximize} & \displaystyle \summ \braket{\Phi_m, c_m} \\
        \mbox{subject~to} & \Phi \in \mT, \\
        & \displaystyle \summ \braket{\Phi_m, a_{j,m}} \le b_j ~(\forall j \in \mI_J) \\
       \end{array} \\
       \llap{\eqref{prob:P}}
      \end{eqnarray*}
     \end{minipage}
   & \begin{minipage}[t]{5cm}
      \begin{eqnarray*}
       \begin{array}{ll}
        \mbox{minimize} & \displaystyle \lambdaS(\chi) + \sumj q_j b_j \\
        \mbox{subject~to} & (\chi, q) \in \Her_\tV \times \Real_+^J, \\
        & \{\chi - z_m(q)\}_{m=0}^{M-1} \in \mC^* \\
       \end{array} \\
       \llap{\eqref{prob:D}}
      \end{eqnarray*}
     \end{minipage}
   & \begin{minipage}[t]{5cm}
      \begin{eqnarray*}
        \begin{aligned}
         q_j \eta_j(\Phi) &= 0 ~(\forall j \in \mI_J), \\
         \summ \braket{\Phi_m, \chi - z_m(q)} &= 0, \\
         \displaystyle \summ \braket{\Phi_m, \chi} &= \lambdaS(\chi) \\
        \end{aligned} \\
       \llap{\eqref{eq:channel_nas}}
      \end{eqnarray*}
     \end{minipage} \\
   \multicolumn{4}{l}{\rule[0pt]{0pt}{10pt}Example~1: Optimal inconclusive discrimination of combs
   $\{ \hcE_r \}_{r=0}^{R-1} \subset \astTChn$ with the prior probabilities $\{ p_r \}_{r=0}^{R-1}$}
   \vspace{-1.6em} \\
   & \begin{minipage}[t]{5cm}
      \begin{eqnarray*}
       \begin{array}{ll}
        \mbox{maximize} & \displaystyle \sumr \braket{\Phi_r, p_r \cE_r} \\
        \mbox{subject~to} & \Phi \in \mTG, \\
        & \displaystyle \sumr \braket{\Phi_R, p_r \cE_r} = \pinc \\
        \end{array} \\
       \llap{\eqref{prob:PPiInc}}
      \end{eqnarray*}
     \end{minipage}
   & \begin{minipage}[t]{5cm}
      \begin{eqnarray*}
       \begin{array}{ll}
        \mbox{minimize} & \lambdaSG(\chi) - q \pinc \\
        \mbox{subject~to} & (\chi, q) \in \Her_\tV \times \Real_+, \\
        & \chi \ge p_r \cE_r ~(\forall r \in \mI_{R+1}), \\
        \end{array} \\
       \llap{\eqref{prob:DInc}}
      \end{eqnarray*}
      where $p_R \coloneqq q$ and $\cE_R \coloneqq \sumr p_r \cE_r$
     \end{minipage}
   & \begin{minipage}[t]{5cm}
      \begin{eqnarray*}
       \begin{aligned}
        \braket{\Phi_R, \cE_R} &= \pinc, \\
        (\chi - p_r \cE_r) \Phi_r &= \zero ~(\forall r \in \mI_{R+1}), \\
        \displaystyle \sum_{r=0}^R \braket{\Phi_r, \chi} &= \lambdaSG(\chi) \\
        \end{aligned} \\
       \llap{\eqref{eq:channel_nas_inc}}
      \end{eqnarray*}
     \end{minipage} \\
   \multicolumn{4}{l}{\rule[0pt]{0pt}{10pt}Example~2: Optimal inconclusive discrimination of states
   $\{ \rho_r \}_{r=0}^{R-1} \subset \Den_\W$ with the prior probabilities $\{ p_r \}_{r=0}^{R-1}$}
   \vspace{-1.6em} \\
   & \begin{minipage}[t]{5cm}
      \begin{eqnarray*}
       \begin{array}{ll}
        \mbox{maximize} & \displaystyle \sumr \braket{\Phi_r, p_r \rho_r} \\
        \mbox{subject~to} & \Phi \in \POVM_\W, \\
        & \displaystyle \sumr \braket{\Phi_R, p_r \rho_r} = \pinc \\
       \end{array}
      \end{eqnarray*}
     \end{minipage}
   & \begin{minipage}[t]{5cm}
      \begin{eqnarray*}
       \begin{array}{ll}
        \mbox{minimize} & \Tr \chi - q \pinc \\
        \mbox{subject~to} & (\chi, q) \in \Her_\W \times \Real_+, \\
        & \chi \ge p_r \rho_r ~(\forall r \in \mI_{R+1}), \\
       \end{array}
      \end{eqnarray*}
      where $p_R \coloneqq q$ and $\rho_R \coloneqq \sumr p_r \rho_r$
     \end{minipage}
   & \begin{minipage}[t]{5cm}
      \begin{eqnarray*}
       \begin{aligned}
        \braket{\Phi_R, \rho_R} &= \pinc, \\
        (\chi - p_r \rho_r) \Phi_r &= \zero ~(\forall r \in \mI_{R+1}) \\
       \end{aligned} \\
       \llap{\eqref{eq:channel_nas_inc_state}}
      \end{eqnarray*}
     \end{minipage} \\
   \hline \hline
  \end{tabular}
 \end{center}
\end{table*}

The following theorem provides necessary and sufficient conditions that
the optimal value remain unchanged even when an additional constraint is imposed.
\begin{thm} \label{thm:opt_two}
 Let $\mT_1$ and $\mT_2$ be nonempty convex sets satisfying $\mT_1 \subseteq \mT_2 \subseteq \mTG$.
 For each $i \in \{1,2\}$, let us choose a closed convex cone $\mC_i$ and a closed convex set $\S_i$
 such that $\ol{\mT_i} = \left\{ \Phi \in \mC_i : \summ \Phi_m \in \S_i \right\}$,
 $\mC_1 \subseteq \mC_2 \subseteq \mCG$, and $\S_1 \subseteq \S_2 \subseteq \SG$
 [see Eq.~\eqref{eq:T}].
 Problem~\eqref{prob:D} with $(\mC,\S) = (\mC_i,\S_i)$ is denoted by
 Problem~($\mathrm{D}_i$).
 Assume that the feasible set of Problem~\eqref{prob:P} with $\mT = \mT_1$ is not empty
 and that an optimal solution to Problem~($\mathrm{D}_2$) exists.
 We consider the following four statements:
 \begin{enumerate}[label=(\arabic*)]
  \item The optimal value of Problem~\eqref{prob:P} with $\mT = \mT_1$ is
        the same as that with $\mT = \mT_2$ [or, equivalently,
        the optimal values of Problems~($\mathrm{D}_1$) and ($\mathrm{D}_2$) are
        the same].
  \item There exists an optimal solution $(\chi^\opt, q^\opt)$ to Problem~($\mathrm{D}_1$)
        such that it is a feasible solution to Problem~($\mathrm{D}_2$)
        and satisfies $\lambda_{\S_1}(\chi^\opt) = \lambda_{\S_2}(\chi^\opt)$.
  \item Any optimal solution to Problem~($\mathrm{D}_2$) is
        optimal for Problem~($\mathrm{D}_1$).
  \item There exists an optimal solution $(\chi^\opt, q^\opt)$ to Problem~($\mathrm{D}_1$)
        such that it is a feasible solution to Problem~($\mathrm{D}_2$)
        and satisfies $\chi^\opt \in \Lin(\Chn_\tV)$.
 \end{enumerate}
 Then, $(1) \Leftrightarrow (2) \Rightarrow (3)$ always holds.
 Also, if $\S_2 = \SG$ holds, then $(1)$--$(4)$ are all equivalent.
\end{thm}
\begin{proof}
 We start with some preliminary remarks.
 For each $i \in \{1,2\}$, let $D_i^\opt$ and $\D_i$ be, respectively,
 the optimal value and the feasible set of Problem~($\mathrm{D}_i$).
 From $\S_1 \subseteq \S_2$, $\lambda_{\S_1}(\chi) \le \lambda_{\S_2}(\chi)$ holds for
 any $\chi \in \Her_\tV$.
 Thus, we have
 \begin{alignat}{2}
  D_1^\opt &\le D_{\S_1}(\chi, q) &&\le D_{\S_2}(\chi, q), \quad \forall (\chi, q) \in \D_1.
  \label{eq:opt_two_le}
 \end{alignat}
 Also, $D_{\S_1}(\chi, q) = D_{\S_2}(\chi, q)$ is equivalent to
 $\lambda_{\S_1}(\chi) = \lambda_{\S_2}(\chi)$.

 We first show $(1) \Rightarrow (3)$, $(1) \Rightarrow (2)$, and $(2) \Rightarrow (1)$.

 $(1) \Rightarrow (3)$:
 Choose any optimal solution $(\chi^\opt, q^\opt)$ to Problem~($\mathrm{D}_2$).
 Since $\D_2 \subseteq \D_1$ holds from $\mC_1 \subseteq \mC_2$,
 $(\chi^\opt, q^\opt) \in \D_1$ holds.
 We also have $D_1^\opt = D_2^\opt = D_{\S_2}(\chi^\opt, q^\opt)$.
 Thus, from Eq.~\eqref{eq:opt_two_le} with $(\chi, q)$ replaced by $(\chi^\opt, q^\opt)$,
 $D_{\S_1}(\chi^\opt, q^\opt) = D_1^\opt$ must hold.
 Therefore, $(\chi^\opt, q^\opt)$ is optimal for Problem~($\mathrm{D}_1$).

 $(1) \Rightarrow (2)$:
 Let $(\chi^\opt, q^\opt)$ be any optimal solution to Problem~($\mathrm{D}_2$).
 Since Statement~(3) holds, $(\chi^\opt, q^\opt)$ is optimal for Problem~($\mathrm{D}_1$).
 $\lambda_{\S_1}(\chi^\opt) = \lambda_{\S_2}(\chi^\opt)$ obviously holds
 from $D_{\S_1}(\chi^\opt,q^\opt) = D_{\S_2}(\chi^\opt,q^\opt)$.

 $(2) \Rightarrow (1)$:
 From Eq.~\eqref{eq:opt_two_le} with $(\chi, q)$ replaced by $(\chi^\opt, q^\opt)$,
 we have $D_1^\opt = D_{\S_1}(\chi^\opt, q^\opt) = D_{\S_2}(\chi^\opt, q^\opt) \ge D_2^\opt$.
 Thus, since $D_1^\opt \le D_2^\opt$ always holds, we have $D_1^\opt = D_2^\opt$.

 We next assume $\S_2 = \SG$ and show $(3) \Rightarrow (4)$ and $(4) \Rightarrow (2)$.

 $(3) \Rightarrow (4)$:
 From Corollary~\ref{cor:chi_comb}, there exists an optimal solution
 $(\chi^\opt, q^\opt)$ to Problem~($\mathrm{D}_2$) satisfying $\chi^\opt \in \Lin(\Chn_\tV)$.
 From Statement~(3), $(\chi^\opt, q^\opt)$ is optimal for Problem~($\mathrm{D}_1$).

 $(4) \Rightarrow (2)$:
 $\lambda_{\S_1}(\chi^\opt) = \lambdaSG(\chi^\opt)$ obviously holds
 from Eq.~\eqref{eq:LinChn_lambda}.
\end{proof}

\section{Symmetry} \label{sec:symmetry}

We now focus on a process discrimination problem that
has a certain symmetry.
We show that, in such a problem, if at least one optimal solution exists, then
there exists an optimal solution having the corresponding symmetry.
This symmetric property can reduce the number of degrees of freedom
and allows us to easily obtain analytical optimal solutions.
This can also lead to computationally efficient algorithms for finding optimal solutions.

\subsection{Group action}

As a preliminary, we recall a group action.
Let $\mG$ be a group with the identity element $e$.
Assume that the order of $\mG$, denoted as $|\mG|$, is greater than one
since the case $|\mG| = 1$ is trivial.
A \termdef{group action} of $\mG$ on a set $\mT$,
$\{ g \b \Endash:\mT \to \mT \}_{g \in \mG}$, is a set of maps on $\mT$ satisfying
\begin{alignat}{2}
 (gh) \b x &= g \b (h \b x), &\quad &\forall g, h \in \mG, ~x \in \mT, \nonumber \\
 e \b x &= x, &\quad &\forall x \in \mT.
\end{alignat}
Let $\inv{g}$ be the inverse of $g$.
For each $g \in \mG$, since $\inv{g} \b (g \b x) = x$ holds for any $x \in \mT$,
$g \b \Endash$ is bijective.
In this paper, group actions on $\mI_K$ $~(K \ge 1)$ and $\Her_\tV$ are considered.

Let us first consider an action of $\mG$ on $\mI_K$,
$\{ g \b \Endash:\mI_K \to \mI_K \}_{g \in \mG}$.
A trivial example is $g \b k \coloneqq k$ $~(\forall g \in \mG, k \in \mI_K)$.
Another example is $g \b k \coloneqq g \oplus_K k$ $~(\forall g \in \mG, k \in \mI_K)$,
where $\oplus_K$ denotes addition modulo $K$
and $\mG \coloneqq \Integer_K \coloneqq \range{0}{K-1}$ is the cyclic group
with the multiplication $g h \coloneqq g \oplus_K h$ $~(\forall g,h \in \mG)$.

Let us next consider an action of $\mG$ on the real Hilbert space $\Her_\tV$,
$\left\{ g \b \Endash:\Her_\tV \to \Her_\tV \right\}_{g \in \mG}$.
We are only concerned with a linearly isometric action, i.e., each $g \b \Endash$
is linear and satisfies
\begin{alignat}{1}
 \braket{g \b x, g \b y} &= \braket{x,y}, \quad \forall g \in \mG, ~x, y \in \Her_\tV.
 \label{eq:isometry}
\end{alignat}
A typical example is an action expressed in the form
\begin{alignat}{2}
 g \b \Endash &\coloneqq \Ad_{U_g},
 \label{eq:action_uni}
\end{alignat}
where $\mG \ni g \mapsto U_g \in \Uni_\tV$ is
a projective unitary or projective anti-unitary representation%
\footnote{$\mG \ni g \mapsto U_g \in \Uni_\V$ is
called a projective unitary or projective anti-unitary representation of $\mG$
if $\Ad_{U_e} = \ident_\V$ and
$\Ad_{U_g} \c \Ad_{U_{g'}} = \Ad_{U_{gg'}}$ hold for any $g,g' \in \mG$.
In this case, $\Ad_{U_{\inv{g}}} = \Ad_{U_g^\dagger}$ holds.}
(which we will simply call a \termdef{projective representation}) of $\mG$.
Another example is an action expressed in the form
\begin{alignat}{2}
 g \b \Endash &\coloneqq \Ad_{U_{g,T} \ot U'_{g,T} \ot \cdots \ot U_{g,1} \ot U'_{g,1}},
 \label{eq:action_sep}
\end{alignat}
where, for each $t \in \range{1}{T}$,
$\mG \ni g \mapsto U_{g,t} \in \Uni_\Wt$ and
$\mG \ni g \mapsto U'_{g,t} \in \Uni_\Vt$ are projective representations of $\mG$.
For instance, the partial transposes $(\Endash)^{\T_{\Wt}}$ and $(\Endash)^{\T_{\Vt}}$
$~(t \in \range{1}{T})$ can be expressed in the form of Eq.~\eqref{eq:action_sep}.

\subsection{Symmetric discrimination problems}

\begin{define}
 Let $\mG$ be a group.
 We will call Problems~\eqref{prob:P} and \eqref{prob:D}
 \termdef{$\mG$-symmetric} if the following conditions hold:
 (a) there exist group actions of $\mG$ on $\mI_M$, $\mI_J$, and $\Her_\tV$;
 (b) the action of $\mG$ on $\Her_\tV$ is linearly isometric;
 and (c)
 \begin{alignat}{2}
  \Phi^\g &\in \mT, &\quad &\forall g \in \mG, ~\Phi \in \mT, \nonumber \\
  \Phi^\g &\in \mC, &\quad &\forall g \in \mG, ~\Phi \in \mC, \nonumber \\
  g \b \varphi &\in \S, &\quad &\forall g \in \mG, ~\varphi \in \S, \nonumber \\
  g \b a_{j,m} &= a_{g \b j, g \b m}, &\quad &\forall g \in \mG, ~j \in \mI_J, ~m \in \mI_M, \nonumber \\
  b_j &= b_{g \b j}, &\quad &\forall g \in \mG, ~j \in \mI_J
  \label{eq:sym_eval_ab}
 \end{alignat}
 and
 \begin{alignat}{2}
  g \b c_m &= c_{g \b m}, &\quad &\forall g \in \mG, ~m \in \mI_M
  \label{eq:sym_eval_c}
 \end{alignat}
 hold%
 \footnote{In this case, since the map $\Phi \mapsto \Phi^\g$ is invertible,
 $\{ \Phi^\g : \Phi \in \mC \} = \mC$ must hold.
 Also, since the map $\varphi \mapsto g \b \varphi$ is invertible,
 $\{ g \b \varphi : \varphi \in \S \} = \S$ must hold.},
 where
 \begin{alignat}{1}
  \Phi^\g &\coloneqq \{ \Phi^\g_m \coloneqq \inv{g} \b \Phi_{g \b m} \}_{m=0}^{M-1},
  \quad g \in \mG, ~\Phi \in \mC.
  \label{eq:Phig}
 \end{alignat}
\end{define}
If $\mT$ is closed, then the first line of Eq.~\eqref{eq:sym_eval_ab}
is derived from its second and third lines.

A large class of process discrimination problems having certain symmetries
can be formulated as Problem~\eqref{prob:P} with $\mG$-symmetric.
Indeed, in the case of minimum-error state discrimination,
cyclic states \cite{Bel-1975,Ban-Kur-Mom-Hir-1997},
three mirror-symmetric states \cite{And-Bar-Gil-Hun-2002},
linear codes with binary letter-states \cite{Usu-Tak-Hat-Hir-1999},
geometrically uniform (or compound geometrically uniform) states \cite{Eld-Meg-Ver-2004},
and self-symmetric states \cite{Nak-Usu-2012-self}
can be treated within this framework.
Four examples are given as follows
(other examples in the case of state discrimination can be seen
in Sec.~III of Ref.~\cite{Nak-Usu-2013-group}):

\begin{ex}[Optimal inconclusive discrimination] \label{ex:sym_inc}
 We consider Problem~\eqref{prob:PPiIncR}, i.e.,
 the problem of discriminating quantum combs $\{ \hcE_r \}_{r=0}^{R-1}$
 under the inconclusive strategy in which testers are restricted to belong to a subset $\mTh$
 of $\mTGh$.
 Let $p_r$ be the prior probability of the comb $\hcE_r$.
 Since this problem is rewritten as Problem~\eqref{prob:P} with Eq.~\eqref{eq:inc_abc},
 it follows that for some group $\mG$, this problem and its dual problem [i.e., Problem~\eqref{prob:DIncR}]
 are $\mG$-symmetric if and only if
 \begin{alignat}{3}
  \Phi^\g &\in \mC, &\quad g \b \varphi &\in \S, &\quad \varpi_g(R) &= R, \nonumber \\
  p_r &= p_{\varpi_g(r)}, &\quad g \b \cE_r &= \cE_{\varpi_g(r)}
 \end{alignat}
 holds for any $g \in \mG$, $\Phi \in \mC$, $\varphi \in \S$, and $r \in \mI_R$,
 where the action $\{ g \b \Endash \}_{g \in \mG}$ of $\mG$ on $\mI_M$
 ($M \coloneqq R + 1$) is denoted by $\{ \varpi_g(\Endash) \}_{g \in \mG}$.
 Note that the action of $\mG$ on $\mI_J = \mI_1$ is uniquely determined by $g \b 0 = 0$.
 Recall that Problems~\eqref{prob:PPiInc} and \eqref{prob:DInc} are
 the particular case of $\mC = \mCG$ and $\S = \SG$.
\end{ex}

\begin{ex}
 Let us consider Problem~\eqref{prob:PPiInc} with
 $\hcE_r \coloneqq \hLambda_r^{\ast T}$ and
 $\hLambda_0,\dots,\hLambda_{R-1} \in \Chn(\V,\W)$.
 Assume that the prior probabilities are equal and that
 \begin{alignat}{1}
  \hLambda_r &= \Ad_{U^r} \c \hLambda_0, \quad \forall r \in \mI_R
  \label{eq:PPiInc_sym_cycl}
 \end{alignat}
 holds, where $U$ is a unitary operator on $\W$ satisfying $U^R = \I_\W$
 and $U^r \neq \I_\W$ for each $1 \le r < R$.
 Let $\Integer_R \coloneqq \range{0}{R-1}$ be the cyclic group.
 We consider the actions of $\Integer_R$ on $\mI_M$ ($M \coloneqq R + 1$)
 and $\Her_\tV$ given, respectively, by
 \begin{alignat}{1}
  g \b m &\coloneqq
  \begin{dcases}
   g \oplus_R m, & m < R, \\
   R, & m = R, \\
  \end{dcases} \nonumber \\
  g \b \Endash &\coloneqq \Ad_{U^g \ot \I_\V \ot U^g \ot \I_\V \ot \cdots \ot U^g \ot \I_\V}
 \end{alignat}
 for each $g \in \Integer_R$; then,
 $g \b \cE_r = \cE_{g \b r}$ holds for any $g \in \Integer_R$ and $r \in \mI_R$.
 Thus, one can easily verify from Example~\ref{ex:sym_inc} that this problem is $\Integer_R$-symmetric.
\end{ex}

\begin{ex} \label{ex:sym_self}
 Let us consider Problem~\eqref{prob:PPiInc} with
 $\hcE_r \coloneqq \hLambda_r^{\ast T}$ and
 $\hLambda_0,\dots,\hLambda_{R-1} \in \Chn(\V,\W)$.
 Let $\mH$ be a group and assume that
 \begin{alignat}{1}
  \hLambda_r &= \Ad_{U_h} \c \hLambda_r \c \Ad_{\tU_h},
  \quad \forall h \in \mH, ~r \in \mI_R
  \label{eq:PPiInc_sym_self}
 \end{alignat}
 [or, equivalently, $\Lambda_r = \Ad_{U_h \ot \tU_h^\T}(\Lambda_r)$]
 holds for some projective representations
 $\mH \ni h \mapsto U_h \in \Uni_\W$ and $\mH \ni h \mapsto \tU_h \in \Uni_\V$.
 The prior probabilities are arbitrarily chosen.
 Note that a channel $\hLambda_r$ satisfying Eq.~\eqref{eq:PPiInc_sym_self}
 is sometimes called covariant.
 Let us consider the $T$-fold direct product of $\mH$,
 \begin{alignat}{2}
  \mH^T &\coloneqq \{ (h_1,\dots,h_T) : h_1,\dots,h_T \in \mH \},
  \label{eq:sym_self_HT}
 \end{alignat}
 and its group actions on $\mI_M$ ($M \coloneqq R + 1$) and $\Her_\tV$ defined as
 \begin{alignat}{2}
  g \b m &\coloneqq m, &\quad &g \in \mH^T, ~m \in \mI_M, \nonumber \\
  (h_1,\dots,h_T) \b \Endash &\coloneqq \Ad_{U_{h_T} \ot \tU_{h_T}^\T \ot \cdots \ot U_{h_1} \ot \tU_{h_1}^\T},
  &\quad &(h_1,\dots,h_T) \in \mH^T;
  \label{eq:sym_self_HT_action}
 \end{alignat}
 then, $g \b \cE_r = \cE_{g \b r} ~(= \cE_r)$ holds for any $g \in \mH^T$ and $r \in \mI_R$.
 Thus, it is easily seen from Example~\ref{ex:sym_inc} that this problem is $\mH^T$-symmetric
 if $\Phi^\g$ belongs to $\mCG$ for any $\Phi \in \mCG$.
\end{ex}

\begin{ex} \label{ex:sym_permute}
 As an example of a problem with restricted testers,
 let us consider Problem~\eqref{prob:PPiIncR} with $R \coloneqq T!$,
 $\V_1 = \cdots = \VT \eqqcolon \V$,
 $\W_1 = \cdots = \WT \eqqcolon \W$,
 $\hcE_{\gamma_r} \coloneqq \hLambda_{\gamma_r(T)} \ast \cdots \ast \hLambda_{\gamma_r(1)}$,
 and $\hLambda_1,\dots,\hLambda_T \in \Chn(\V,\W)$,
 where $\gamma_r$ is the permutation on $\range{1}{T}$ determined by $r \in \mI_R$.
 For simplicity, we focus on the case $T = 2$,
 i.e., the problem of discriminating $\hcE_0 \coloneqq \hLambda_1 \ast \hLambda_0$
 and $\hcE_1 \coloneqq \hLambda_0 \ast \hLambda_1$.
 Assume that the prior probabilities are equal
 and that testers are restricted to nonadaptive ones.
 In this case, Eq.~\eqref{eq:T} with
 \begin{alignat}{1}
  \mC &\coloneqq \mCG, \nonumber \\
  \S &\coloneqq \left\{ (\ident_{\W_2} \ot \cross_{\W_1,\V_2} \ot \ident_{\V_1})
  (\I_{\W_2 \ot \W_1} \ot \rho) : \rho \in \Den_{\V_2 \ot \V_1} \right\}
  \label{eq:C_nonadaptive}
 \end{alignat}
 holds, where $\cross_{\V,\W}$ is the process that swaps two systems $\V$ and $\W$.
 Let $\mG \coloneqq \{ e, \tg \}$, where
 $\left\{ g \b \Endash:\Her_\tV \to \Her_\tV \right\}_{g \in \mG}$
 is the linear action characterized by
 $\tg \b \Endash \coloneqq \cross_{\W_2 \ot \V_2,\W_1 \ot \V_1}$.
 Note that this action can be expressed in the form of Eq.~\eqref{eq:action_uni}.
 Since $\cE_0 = \Lambda_1 \ot \Lambda_0$ and $\cE_1 = \Lambda_0 \ot \Lambda_1$ holds,
 $g \b \cE_r = \cE_{g \b r}$ holds for each $g \in \mG$ and $r \in \mI_2$.
 Thus, in the case of inconclusive strategy, one can easily verify from Example~\ref{ex:sym_inc}
 that the problem is $\mG$-symmetric.
 The same discussion can be applied to the case $T > 2$.
\end{ex}

\subsection{Symmetric solutions}

Let us fix a group $\mG$.
For any $\Phi \in \mC$, let
\begin{alignat}{1}
 \Phi^\d &\coloneqq \left\{ \Phi^\d_m \coloneqq \frac{1}{|\mG|}
 \sumg \Phi^\g_m \right\}_{m=0}^{M-1},
 \label{eq:Phib}
\end{alignat}
where $\Phi^\g$ is defined by Eq.~\eqref{eq:Phig}.
It follows that $\Phi^\d$ has the symmetry property
\begin{alignat}{1}
 g \b \Phi^\d_m &= \Phi^\d_{g \b m}, \quad \forall g \in \mG, ~m \in \mI_M,
 \label{eq:sym_Phi}
\end{alignat}
which follows from
\begin{alignat}{1}
 g \b \Phi^\d_m &= \frac{1}{|\mG|} \sum_{h \in \mG} g \b \Phi^{(h)}_m
 = \frac{1}{|\mG|} \sum_{h' \in \mG} \inv{h'} \b \Phi_{(h'g) \b m}
 = \Phi^\d_{g \b m},
\end{alignat}
where $h' \coloneqq h \inv{g}$.
Similarly, for any $(\chi, q) \in \Her_\tV \times \Real_+^J$, let
\begin{alignat}{2}
 \chi^\d &\coloneqq \frac{1}{|\mG|} \sumg \chi^\g, &\quad
 q^\d &\coloneqq \left\{ q^\d_j \coloneqq \frac{1}{|\mG|} \sumg q^\g_j \right\}_{j=0}^{J-1},
 \nonumber \\
 \chi^\g &\coloneqq g \b \chi, &\quad
 q^\g &\coloneqq \{ q_j^\g \coloneqq q_{\inv{g} \b j} \}_{j=0}^{J-1}.
 \label{eq:chib}
\end{alignat}
From
\begin{alignat}{1}
 g \b \chi^\d &= \frac{1}{|\mG|} \sum_{h \in \mG} g \b \chi^{(h)}
 = \frac{1}{|\mG|} \sum_{h \in \mG} (gh) \b \chi = \chi^\d, \nonumber \\
 q^\d_{g \b j} &= \frac{1}{|\mG|} \sum_{h \in \mG} q^{(h)}_{g \b j}
 = \frac{1}{|\mG|} \sum_{h \in \mG} q_{(\inv{h}g) \b j} = q^\d_j,
\end{alignat}
$(\chi^\d,q^\d)$ has the symmetry property
\begin{alignat}{2}
 g \b \chi^\d &= \chi^\d, &\quad &\forall g \in \mG, \nonumber \\
 q^\d_j &= q^\d_{g \b j}, &\quad &\forall g \in \mG, ~j \in \mI_J.
 \label{eq:sym_chi}
\end{alignat}

\begin{lemma} \label{lemma:sym_kappa}
 If $\mT$, $\mC$, $\S$, $\{ a_{j,m} \}_{(j,m)=(0,0)}^{(J-1,M-1)} \subset \Her_\tV$, and
 $\{ b_j \}_{j=0}^{J-1} \in \Real^J$ satisfy Eq.~\eqref{eq:sym_eval_ab}, then
 $\Phi^\g, \Phi^\d \in \P$ holds for any $\Phi \in \P$ and $g \in \mG$.
\end{lemma}
\begin{proof}
 Arbitrarily choose $\Phi \in \P$.
 It follows from $\Phi \in \mT$ and the first line of Eq.~\eqref{eq:sym_eval_ab} that
 $\Phi^\g, \Phi^\d \in \mT$ holds.
 We have that for any $g \in \mG$ and $j \in \mI_J$,
 \begin{alignat}{1}
  \summ \braket{\Phi^\g_m, a_{j,m}}
  &= \summ \braket{\inv{g} \b \Phi_{g \b m}, a_{j,m}}
  = \summ \braket{\Phi_{g \b m}, g \b a_{j,m}} \nonumber \\
  &= \summ \braket{\Phi_{g \b m}, a_{g \b j,g \b m}}
  \le b_{g \b j} = b_j,
  \label{eq:sym_Pi_Phig}
 \end{alignat}
 where the inequality follows from the map $g \b \Endash : \mI_M \to \mI_M$
 being bijective.
 Thus, we have $\Phi^\g \in \P$.
 Since $\P$ is convex, we have $\Phi^\d \in \P$.
\end{proof}

\begin{thm} \label{thm:sym_Phi}
 Let $\mG$ be a group.
 Assume that Problem~\eqref{prob:P} is $\mG$-symmetric; then
 $\Phi^\d \in \P$ and $P(\Phi^\d) = P(\Phi)$ hold for any $\Phi \in \P$.
\end{thm}
\begin{proof}
 $\Phi^\d \in \P$ holds from Lemma~\ref{lemma:sym_kappa}.
 We have that for any $g \in \mG$,
 \begin{alignat}{1}
  P[\Phi^\g] &= \summ \braket{\Phi^\g_m, c_m}
  = \summ \braket{\inv{g} \b \Phi_{g \b m}, c_m} \nonumber \\
  &= \summ \braket{\Phi_{g \b m}, g \b c_m}
  = \summ \braket{\Phi_{g \b m}, c_{g \b m}} = P(\Phi).
 \end{alignat}
 Thus, we have
 \begin{alignat}{2}
  P(\Phi^\d) &= \frac{1}{|\mG|} \sumg P[\Phi^\g] &&= P(\Phi).
 \end{alignat}
\end{proof}
Considering the case of $\Phi$ being optimal for Problem~\eqref{prob:P},
we immediately obtain the following corollary as a special case of Theorem~\ref{thm:sym_Phi}
(proof omitted).
\begin{cor} \label{cor:sym_Phi}
 Let $\mG$ be a group.
 Assume that Problem~\eqref{prob:P} is $\mG$-symmetric.
 Then, for any optimal solution, $\Phi$, to Problem~\eqref{prob:P},
 $\Phi^\d$ is also optimal for Problem~\eqref{prob:P}.
\end{cor}

In the case of Problem~\eqref{prob:P} being $\mG$-symmetric,
this corollary guarantees that if at least one optimal solution exists, then
there also exists an optimal solution with the symmetry property of Eq.~\eqref{eq:sym_Phi}.
This corollary also implies that the optimal value remains unchanged even if
we impose the additional constraint of Eq.~\eqref{eq:sym_Phi}
(with $\Phi^\d$ replaced by $\Phi$).
Problem~\eqref{prob:P} with this constraint is still convex.

\begin{thm} \label{thm:sym_chi}
 Let $\mG$ be a group.
 Assume that Problem~\eqref{prob:D} is $\mG$-symmetric; then,
 $(\chi^\d, q^\d) \in \D$ and $D_\S(\chi^\d, q^\d) \le D_\S(\chi, q)$ hold
 for any $(\chi, q) \in \D$.
\end{thm}
\begin{proof}
 We have that for any $m \in \mI_M$,
 \begin{alignat}{1}
  z_m(q^\d) &= c_m - \sumj q_j^\d a_{j,m}
  = \frac{1}{|\mG|} \sumg \left[ c_m - \sumj q_j^\g a_{j,m} \right] \nonumber \\
  &= \frac{1}{|\mG|} \sumg \left[ c_m - \sum_{j'=0}^{J-1} q_{j'} a_{g \b j',m} \right] \nonumber \\
  &= \frac{1}{|\mG|} \sumg g \b \left[ c_{\inv{g} \b m} - \sum_{j'=0}^{J-1} q_{j'} a_{j',\inv{g} \b m}
  \right] \nonumber \\
  &= \frac{1}{|\mG|} \sumg g \b z_{\inv{g} \b m}(q),
 \end{alignat}
 where $j' \coloneqq \inv{g} \b j$.
 This yields
 \begin{alignat}{1}
  \summ \braket{\Phi_m, \chi^\d - z_m(q^\d)}
  &= \frac{1}{|\mG|} \summ \sumg \braket{\Phi_m, g \b [\chi - z_{\inv{g} \b m}(q)]} \nonumber \\
  &= \frac{1}{|\mG|} \sum_{m'=0}^{M-1} \sumg \braket{\Phi^\g_{m'}, \chi - z_{m'}(q)} \ge 0
  \label{eq:sym_chid}
 \end{alignat}
 for any $\Phi \in \mC$, where $m' \coloneqq \inv{g} \b m$.
 The inequality follows from $\Phi^\g \in \mC$
 and $\{ \chi - z_{m'}(q) \}_{m'=0}^{M-1} \in \mC^*$.
 Therefore, $\{ \chi^\d - z_m(q^\d) \}_{m=0}^{M-1} \in \mC^*$, i.e.,
 $(\chi^\d,q^\d) \in \D$ holds.
 Moreover, we have
 \begin{alignat}{1}
  D_\S(\chi^\d, q^\d) &= \lambdaS\left[ \frac{1}{|\mG|} \sumg \chi^\g \right]
  + \sumj \frac{1}{|\mG|} \sumg q^\g_j b_j \nonumber \\
  &\le \frac{1}{|\mG|} \sumg \left[ \lambdaS[\chi^\g] + \sumj q^\g_j b_j \right]
  \nonumber \\
  &= \frac{1}{|\mG|} \sumg \left[ \lambdaS(\chi) + \sumj q_{\inv{g} \b j} b_{\inv{g} \b j} \right]
  \nonumber \\
  &= \frac{1}{|\mG|} \sumg D_\S(\chi, q) = D_\S(\chi, q),
 \end{alignat}
 where the second line follows since $\lambdaS$ is convex.
 The third line follows from
 \begin{alignat}{1}
  \lambdaS[\chi^\g] &= \sup_{\varphi \in \S} \braket{\varphi, g \b \chi}
  = \sup_{\varphi \in \S} \braket{\inv{g} \b \varphi, \chi}
  = \lambdaS(\chi).
 \end{alignat}
\end{proof}
We immediately obtain the following corollary as a special case of Theorem~\ref{thm:sym_chi}
(proof omitted).
\begin{cor} \label{cor:sym_chi}
 Let $\mG$ be a group.
 Assume that Problem~\eqref{prob:D} is $\mG$-symmetric.
 Then, for any optimal solution, $(\chi, q)$, to Problem~\eqref{prob:D},
 $(\chi^\d, q^\d)$ is also optimal for Problem~\eqref{prob:D}.
\end{cor}

In the case of Problem~\eqref{prob:D} being $\mG$-symmetric,
this corollary says that there exists an optimal solution with the symmetry property
of Eq.~\eqref{eq:sym_chi} whenever an optimal solution exists.
Also, the optimal value does not change even if we impose the additional constraint
of Eq.~\eqref{eq:sym_chi} [with $(\chi^\d,q^\d)$ replaced by $(\chi,q)$].
Problem~\eqref{prob:D} with this constraint is still convex.

\subsection{Sufficient conditions that a tester with maximally entangled pure states can be optimal}
\label{subsec:symmetry_entangled}

We will call a tester expressed as in Fig.~\ref{fig:tester_Bell}
a \termdef{tester with maximally entangled pure states},
where $\tPsi_1,\dots,\tPsi_T$ are maximally entangled pure states
and $\hPi \coloneqq \{ \hPi_m \}_{m=0}^{M-1}$ is a measurement.
Such a tester $\hPhi$ is expressed by Eq.~\eqref{eq:hPhi}
with $\hsigma_t \coloneqq \tPsi_t \ot \ident_{\W_{t-1} \ot \V_{t-1} \ot \cdots \ot \W_1 \ot \V_1}$
$~(t \in \range{1}{T})$.
We may assume, without loss of generality, that each $\tPsi_t$ is
the generalized Bell state $\kket{\I_\Vt}\bbra{\I_\Vt} / N_\Vt \in \Den_{\Vt \ot \Vt}$.
In this case, we have $\Phi_m = \Pi_m / \Pro{t=1}{T} N_\Vt$.
It is easily seen that $\Phi \in \mT$ is a tester with maximally entangled pure states
if and only if $\summ \Phi_m$ is in the unit set
\begin{alignat}{1}
 \S_\Psi &\coloneqq \left\{ \I_\tV / \Pro{t=1}{T} N_\Vt \right\}.
 \label{eq:SPsi}
\end{alignat}
$\S_\Psi$ is obviously a closed convex subset of $\SG$.
In some $\mG$-symmetric problems, we can derive sufficient conditions that
the optimal value of Problem~\eqref{prob:P} with $\S = \SG$ remain unchanged
if $\S$ is replaced by $\S_\Psi$.
\begin{figure}[bt]
 \centering
 \InsertPDF{1.0}{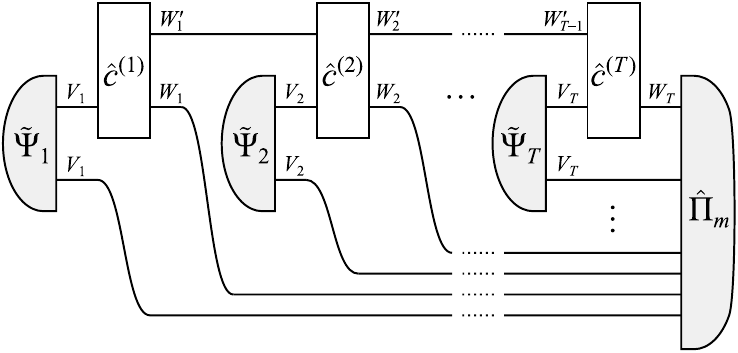}
 \caption{Tester with maximally entangled pure states.
 $\tPsi_1,\dots,\tPsi_T$ are maximally entangled pure states
 and $\{ \hPi_m \}_{m=0}^{M-1}$ is a measurement.}
 \label{fig:tester_Bell}
\end{figure}

\begin{proposition} \label{pro:input_entangle_symT}
 Let $\mG$ be a group whose action on $\Her_\tV$ satisfies
 $g \b \I_\tV = \I_\tV$ $~(\forall g \in \mG)$.
 [Note that if $g \b \Endash$ is expressed in the form of Eq.~\eqref{eq:action_uni} or
 Eq.~\eqref{eq:action_sep}, then $g \b \I_\tV = \I_\tV$ holds.]
 Assume that $\S = \SG$ holds and that Problem~\eqref{prob:P} is $\mG$-symmetric.
 Also, assume that, for each $t \in \range{1}{T}$, there exists a subgroup $\mH^{(t)} \subseteq \mG$
 such that
 \begin{alignat}{1}
  \Trp{\WVT \ot \cdots \ot \Wt} (h \b \Endash)
  &= \Ad_{U'_{h,t}} \ot \ident_{\W_{t-1} \ot \V_{t-1} \ot \cdots \ot \W_1 \ot \V_1}
 \end{alignat}
 for any $h \in \mH^{(t)}$, where $\mH^{(t)} \ni h \mapsto U'_{h,t} \in \Uni_\Vt$ is
 an irreducible projective representation%
 \footnote{$\mH^{(t)} \ni h \mapsto U'_{h,t} \in \Uni_\Vt$ is called irreducible if
 it has only two subrepresentations $\{ 0 \}$ and $\Vt$,
 where a subrepresentation is a subspace $\V$ of $\Vt$ that satisfies
 $U'_{h,t}\ket{x} \in \V$ for any $h \in \mH^{(t)}$ and $\ket{x} \in \V$.}.
 Then, the optimal value of Problem~\eqref{prob:P} remains unchanged if $\S$ is replaced by $\S_\Psi$.
\end{proposition}
\begin{proof}
 Let $P^\opt$ be the optimal value of Problem~\eqref{prob:P}
 and $P_\Psi^\opt$ be that of Problem~\eqref{prob:P} with $\S$ replaced by $\S_\Psi$.
 The assumption $g \b \I_\tV = \I_\tV$ gives
 $g \b \varphi \in \S_\Psi$ $~(\forall \varphi \in \S_\Psi)$.
 Thus, Problem~\eqref{prob:P} with $\S$ replaced by $\S_\Psi$
 is also $\mG$-symmetric.
 Arbitrarily choose $0 < \varepsilon \in \Real_+$; then,
 it is easily seen that there exists $(\chi, q) \in \D$ such that
 $D_{\S_\Psi}(\chi, q) = P_\Psi^\opt + \varepsilon$.
 From Theorem~\ref{thm:sym_chi}, we have
 $D_{\S_\Psi}(\chi^\d, q^\d) \le P_\Psi^\opt + \varepsilon$.
 Let $\X_t \coloneqq \W_{t-1} \ot \V_{t-1} \ot \cdots \ot \W_1 \ot \V_1$ and
 $\chi^\d_t \coloneqq \Trp{\WVT \ot \cdots \ot \Wt} \chi^\d \in \Her_{\Vt \ot \X_t}$
 $~(t \in \range{1}{T})$.
 Assume now that, for each $t \in \range{1}{T}$,
 $\chi^\d_t$ is expressed in the form
 \begin{alignat}{1}
  \chi^\d_t &= \I_\Vt \ot \chi'_t, \quad \chi'_t \in \Her_{\X_t};
  \label{eq:chidt_I}
 \end{alignat}
 then, we can easily check $\chi^\d \in \Lin(\Chn_\tV)$ [see Eq.~\eqref{eq:Comb}].
 From $(\chi^\d, q^\d) \in \D$ and Eq.~\eqref{eq:LinChn_lambda},
 we have $P^\opt \le D_\S(\chi^\d, q^\d) = D_{\S_\Psi}(\chi^\d, q^\d)$,
 which gives $P^\opt \le P_\Psi^\opt + \varepsilon$.
 From $P_\Psi^\opt \le P^\opt$, $P^\opt = P_\Psi^\opt$ must hold.

 It remains to show that, for each $t \in \range{1}{T}$,
 $\chi^\d_t$ is expressed in the form of Eq.~\eqref{eq:chidt_I} .
 Let us arbitrarily choose $s \in \Pos_{\X_t}$ and
 let $\chi^\d_{t,s} \coloneqq \Trp{\X_t}[(\I_\Vt \ot s) \chi^\d_t] \in \Her_\Vt$.
 For any $h \in \mH^{(t)}$, it follows from $h \b \chi^\d = \chi^\d$ that
 $(\Ad_{U'_{h,t}} \ot \ident_{\X_t})(\chi^\d_t) = \chi^\d_t$ holds,
 which gives $\Ad_{U'_{h,t}}(\chi^\d_{t,s}) = \chi^\d_{t,s}$.
 Since the representation $h \mapsto U'_{h,t}$ $~[h \in \mH^{(t)}]$ is irreducible,
 from Schur's lemma (on anti-unitary groups) \cite{Dim-1963},
 $\chi^\d_{t,s}$ must be proportional to $\I_\Vt$.
 Since $\Tr \chi^\d_{t,s} = \braket{s, \Trp{\Vt} \chi^\d_t}$ holds
 from the definition of $\chi^\d_{t,s}$,
 $\chi^\d_{t,s} = \braket{s, \Trp{\Vt} \chi^\d_t} \I_\Vt / N_\Vt$ holds.
 Thus, we have that for any $s' \in \Pos_\Vt$,
 \begin{alignat}{1}
  \braket{s' \ot s, \chi^\d_t} &= \braket{s', \chi^\d_{t,s}}
  = \braket{s, \Trp{\Vt} \chi^\d_t} \braket{s', \I_\Vt / N_\Vt} \nonumber \\
  &= \braket{s' \ot s, \I_\Vt \ot \chi'_t},
  \label{eq:ss_chidt}
 \end{alignat}
 where $\chi'_t \coloneqq \Trp{\Vt} \chi^\d_t / N_\Vt$.
 Since Eq.~\eqref{eq:ss_chidt} holds for any $s$ and $s'$,
 we obtain Eq.~\eqref{eq:chidt_I}.
\end{proof}

Any tester with maximally entangled pure states is non-adaptive.
Thus, Proposition~\ref{pro:input_entangle_symT} implies that if there exists an optimal solution to
Problem~\eqref{prob:P} with $\S$ replaced by $\S_\Psi$,
then an adaptive strategy is not necessary for optimal discrimination in Problem~\eqref{prob:P}.
We will give some applications of this proposition in Subsec.~\ref{subsec:example_input}.

\section{Minimax strategy} \label{sec:minimax}

We now discuss a minimax strategy for a quantum process discrimination problem.
This strategy is useful in particular in the case in which the prior probabilities
of the processes are not known.

\subsection{Formulation}

Let us consider a process discrimination problem in which
the value of an objective function, $Q_k(\Phi)$, depends not only on a tester $\Phi$
but also on some random variable, $k \in \mI_K$.
We want to maximize the average of $Q_k(\Phi)$,
\begin{alignat}{1}
 Q(\mu, \Phi) &\coloneqq \sumk \mu_k Q_k(\Phi),
 \label{eq:fmm}
\end{alignat}
where $\mu \coloneqq \{ \mu_k \}_{k=0}^{K-1}$ is a probability distribution of $k$.
Here, we consider the situation in which the probability distribution $\mu$
is unknown but known to lie in a fixed subset, $\Prob$,
of $\Probmax \coloneqq \left\{ \mu \in \Real_+^K : \sumk \mu_k = 1 \right\}$.
In what follows, assume that $\Prob$ is a nonempty closed convex set.
A natural approach is to maximize the infimum of $Q(\mu, \Phi)$ over $\mu \in \Prob$.
This problem is formulated as
\begin{alignat}{1}
 \begin{array}{ll}
  \mbox{maximize} & \displaystyle \inf_{\mu \in \Prob} Q(\mu, \Phi) \\
  \mbox{subject~to} & \Phi \in \P, \\
 \end{array}
 \tag{$\mathrm{P_{mm}}$} \label{prob:Pm}
\end{alignat}
where $\P$ is defined by Eq.~\eqref{eq:TestercC}.
Assume that, for each $k \in \mI_K$, $Q_k(\Phi)$ is expressed in the form
\begin{alignat}{2}
 Q_k(\Phi) &\coloneqq \summ \braket{\Phi_m, c_{k,m}},
 \label{eq:fmm_k}
\end{alignat}
where $\{ c_{k,m} \}_{(k,m)=(0,0)}^{(K-1,M-1)} \subset \Her_\tV$ are constants.
Note that, when $c_{k,m}$ is expressed in the form $c_{k,m} = c'_{k,m} + d_k u$
with $c'_{k,m} \in \Her_\tV$, $d_k \in \Real$,
and $u \coloneqq \I_\tV / \prod_{t=1}^T N_\Wt$,
we can rewrite $Q_k(\Phi)$ as
\begin{alignat}{1}
 Q_k(\Phi) &= \summ \braket{\Phi_m, c'_{k,m}} + d_k, \quad \forall \Phi \in \P.
\end{alignat}

\begin{ex}[Optimal inconclusive discrimination]
 We can consider a minimax version of Problem~\eqref{prob:PPiInc}.
 Assume that the prior probabilities $p \coloneqq \{ p_r \}_{r=0}^{R-1}$ of the combs
 $\{ \hcE_r \}_{r=0}^{R-1}$ are completely unknown.
 The average success and inconclusive probabilities are, respectively, expressed as
 \begin{alignat}{2}
  \PS(\hPhi; p) &\coloneqq \sumr p_r \mathrm{Pr}(r|\hcE_r),
  &\quad
  \PI(\hPhi; p) &\coloneqq \sumr p_r \mathrm{Pr}(R|\hcE_r),
 \end{alignat}
 where $\mathrm{Pr}(m|\hcE_r) \coloneqq \braket{\hPhi_m, \hcE_r}$.
 Since the constraint $\PI(\hPhi; p) = \pinc$ $~(\forall p \in \Probmax)$ is too tight,
 we relax it to $\PI(\hPhi; p) \le \pinc$ $~(\forall p \in \Probmax)$.
 Let us consider the problem of minimizing the maximum average error probability,
 which is equal to $1 - \PS(\hPhi; p) - \PI(\hPhi; p)$.
 This problem is formulated as
 \begin{alignat}{1}
  \begin{array}{ll}
   \mbox{maximize}
    & \displaystyle \min_{p \in \Probmax} [\PS(\hPhi; p) + \PI(\hPhi; p)] \\
   \mbox{subject~to} & \hPhi \in \mTGh,
    ~\PI(\hPhi; p) \le \pinc ~(\forall p \in \Probmax). \\
  \end{array}
  \label{prob:PPiIncMinimax}
 \end{alignat}
 The second constraint is equivalent to
 $\mathrm{Pr}(R|\hcE_j) \le \pinc$ $~(\forall j \in \mI_R)$,
 and thus this problem is rewritten as Problem~\eqref{prob:Pm} with
 \begin{alignat}{3}
  M &\coloneqq R + 1, &\quad K &\coloneqq R, &\quad J &\coloneqq R, \nonumber \\
  c_{k,m} &\coloneqq (\delta_{m,k} + \delta_{m,R}) \cE_k, &\quad
  \Prob &\coloneqq \Probmax, &\quad \mT &\coloneqq \mTG, \nonumber \\
  a_{j,m} &\coloneqq \delta_{m,R} \cE_j, &\quad
  b_j &\coloneqq \pinc.
 \end{alignat}
 If $T = 1$ and $\V_1 = \Complex$ hold, then this problem is the state discrimination problem
 discussed in Ref.~\cite{Nak-Kat-Usu-2013-minimax}.
 In the special case of $\pinc = 0$, Problem~\eqref{prob:PPiIncMinimax} is rewritten as
 \begin{alignat}{1}
  \begin{array}{ll}
   \mbox{maximize} & \displaystyle \min_{p \in \Probmax} \PS(\hPhi; p) \\
   \mbox{subject~to} & \hPhi \in \mTGh; \\
  \end{array}
  \label{eq:ProbPPiIncMinimaxZero}
 \end{alignat}
 in this case, without loss of generality, we can assume $\hPhi_R = \zero$.
 Problem~\eqref{eq:ProbPPiIncMinimaxZero} corresponds to a minimax version of
 minimum-error discrimination.
\end{ex}

\begin{ex}[Discrimination of sets of combs]
 Let us consider the problem of discriminating $R$ subsets of combs,
 $\{ \hcE_{0,l} \}_{l=0}^{L_0-1}, \{ \hcE_{1,l} \}_{l=0}^{L_1-1}, \dots,
 \{ \hcE_{R-1,l} \}_{l=0}^{L_{R-1}-1}$,
 where $L_0,\dots,L_{R-1}$ are natural numbers.
 Assume that the prior probability, $p_{r,l}$, of each comb $\hcE_{r,l}$
 is unknown.
 We want to maximize the infimum of the average success probability given by
 \begin{alignat}{1}
  \PS'(\hPhi; p) &\coloneqq \sumr \sum_{l=0}^{L_r-1}
  p_{r,l} \braket{\hPhi_r, \hcE_{r,l}},
 \end{alignat}
 where $p \coloneqq \{ p_{r,l} \}_{(r,l)=(0,0)}^{(R-1,L_r-1)} \in \Prob$
 $~(K \coloneqq \sumr L_r)$.
 This problem can be formulated as follows:
 \begin{alignat}{1}
  \begin{array}{ll}
   \mbox{maximize} & \displaystyle \inf_{p \in \Prob} \PS'(\hPhi; p) \\
   \mbox{subject~to} & \hPhi \in \mTh. \\
  \end{array}
  \label{eq:ProbPPiIncMinimaxGroup}
 \end{alignat}
 One can easily verify that this problem is equivalent to Problem~\eqref{prob:Pm} with
 \begin{alignat}{1}
  M &\coloneqq R, \quad K \coloneqq \sumr L_r, \quad J \coloneqq 0,
  \quad c_{k(r,l),m} \coloneqq \delta_{r,m} \cE_{r,l},
 \end{alignat}
 where $k(r,l) \coloneqq \sum_{r'=0}^{r-1} L_{r'} + l$
 $~(r \in \mI_R, l \in \mI_{L_r})$.
 Note that if the prior probabilities are known,
 then this problem can be simply reduced to the problem of discriminating
 $R$ combs $\{ \sum_{l=0}^{L_r-1} p_{r,l} \hcE_{r,l} / p'_r \}_{r=0}^{R-1}$
 with the prior probabilities $\{ p'_r \}_{r=0}^{R-1}$,
 where $p'_r \coloneqq \sum_{l=0}^{L_r-1} p_{r,l}$.
\end{ex}

\subsection{Properties of minimax solutions}

$(\mu^\opt, \Phi^\opt) \in \Prob \times \P$ is called
a \termdef{minimax solution} (or \termdef{saddle point}) if
\begin{alignat}{1}
 Q(\mu^\opt, \Phi) &\le Q(\mu^\opt, \Phi^\opt) \le Q(\mu, \Phi^\opt)
 \label{eq:minimax_saddle}
\end{alignat}
holds for any $\mu \in \Prob$ and $\Phi \in \P$.
We refer to $\Phi^\opt$ as a \termdef{minimax tester}.
Let
\begin{alignat}{1}
 Q^\opt(\mu) &\coloneqq \sup_{\Phi \in \P} Q(\mu, \Phi).
 \label{eq:Qstar}
\end{alignat}
If there exists a minimax solution to $Q$, then
$(\mu^\opt, \Phi^\opt)$ is a minimax solution to $Q$ if and only if
$\Phi^\opt$ is optimal for Problem~\eqref{prob:Pm}
and $\mu^\opt \in \argmin_{\mu \in \Prob} Q^\opt(\mu)$ holds \cite{Ber-2009}.
Also, from Eq.~\eqref{eq:minimax_saddle}, $Q^\opt(\mu^\opt) = Q(\mu^\opt, \Phi^\opt)$ holds.

\begin{remark} \label{remark:minimax}
 Assume that $\P$ is nonempty and closed; then,
 in Problem~\eqref{prob:Pm}, there exists a minimax solution to $Q$.
\end{remark}
\begin{proof}
 $\P$ and $\Prob$ are nonempty compact convex sets.
 $Q(\mu, \Phi)$ is a continuous convex function of $\mu$ for fixed $\Phi$
 and a continuous concave function of $\Phi$ for fixed $\mu$.
 Then, from Ref.~\cite{Eke-Tem-1999} (Chap.~VI, Proposition~2.1),
 there exists a minimax solution to $Q$.
\end{proof}

The following remark states that the problem of finding $Q^\opt(\mu)$ can be formulated as
Problem~\eqref{prob:P}.
\begin{remark}
 For given $\mu \in \Prob$, let $P_\mu^\opt$ be the optimal value of Problem~\eqref{prob:P}
 with $c_m \coloneqq \sumk \mu_k c_{k,m}$; then, $Q^\opt(\mu) = P_\mu^\opt$ holds.
\end{remark}
\begin{proof}
 We have
 \begin{alignat}{2}
  Q^\opt(\mu) &= \sup_{\Phi \in \P} Q(\mu, \Phi)
  &&= \sup_{\Phi \in \P} \sumk \mu_k \summ \braket{\Phi_m, c_{k,m}} \nonumber \\
  &= \sup_{\Phi \in \P} \summ \braket{\Phi_m, c_m} &&= P_\mu^\opt.
 \end{alignat}
\end{proof}

\begin{proposition} \label{pro:minimax_nas}
 $(\mu, \Phi) \in \Prob \times \P$ is a minimax solution to $Q$
 if and only if $Q^\opt(\mu) \le Q(\mu',\Phi)$ holds for any $\mu' \in \Prob$.
\end{proposition}
\begin{proof}
 ``If'':
 Considering the case $\mu = \mu'$, we have $Q^\opt(\mu) \le Q(\mu,\Phi)$.
 Thus, from $Q(\mu,\Phi) \le Q^\opt(\mu)$, $Q(\mu,\Phi) = Q^\opt(\mu)$ must hold.
 Therefore, $(\mu,\Phi)$ is a minimax solution.

 ``Only if'':
 Equation~\eqref{eq:minimax_saddle} gives $Q^\opt(\mu) = Q(\mu,\Phi) \le Q(\mu',\Phi)$
 for any $\mu' \in \Prob$.
\end{proof}

\begin{proposition} \label{pro:minimax_interior}
 Assume that the affine hull of $\Prob$ contains $\Probmax$
 [or, equivalently, the affine hull of $\Prob$ is $(K-1)$-dimensional].
 Also, assume that $\mu$ is a relative interior point of $\Prob$ and that $\Phi \in \P$ holds.
 Then, $(\mu, \Phi)$ is a minimax solution to $Q$
 if and only if $Q_k(\Phi) = Q^\opt(\mu)$ holds for any $k \in \mI_K$.
\end{proposition}
\begin{proof}
 ``If'':
 $Q^\opt(\mu) = \sumk \mu'_k Q_k(\Phi) = Q(\mu',\Phi)$
 $~(\forall \mu' \in \Prob)$ holds.
 Thus, from Proposition~\ref{pro:minimax_nas}, $(\mu, \Phi)$ is a minimax solution.

 ``Only if'':
 Assume by contradiction that there exists $k \in \mI_K$ such that $Q_k(\Phi) \neq Q^\opt(\mu)$.
 In the case of $Q_0(\Phi) = \dots = Q_{K-1}(\Phi)$,
 $Q(\mu,\Phi) = \sumk \mu_k Q_k(\Phi) \neq Q^\opt(\mu)$ holds,
 which contradicts that $(\mu, \Phi)$ is a minimax solution.
 Then, we consider the other case.
 Let us choose $k_0 \in \argmin_{k \in \mI_K} Q_k(\Phi)$
 and $k_1 \in \argmax_{k \in \mI_K} Q_k(\Phi)$;
 then, $Q_{k_0}(\Phi) < Q_{k_1}(\Phi)$ holds.
 Also, let $\mu' \coloneqq \{ \mu_k + \varepsilon(\delta_{k,k_0} - \delta_{k,k_1}) \}_{k=0}^{K-1}$;
 then, since $\mu$ is a relative interior point of $\Prob$,
 $\mu' \in \Prob$ holds for sufficiently small $\varepsilon > 0$.
 We have $Q(\mu',\Phi) - Q(\mu,\Phi) = \varepsilon[Q_{k_0}(\Phi) - Q_{k_1}(\Phi)] < 0$,
 which contradicts that $(\mu, \Phi)$ is a minimax solution.
\end{proof}

In the special case of $\Prob = \Probmax$,
the following proposition and corollary hold.
\begin{proposition} \label{pro:minimax}
 Assume $\mu \in \Prob = \Probmax$ and $\Phi \in \P$.
 The following statements are all equivalent.
 \begin{enumerate}[label=(\arabic*)]
  \item $(\mu, \Phi)$ is a minimax solution to $Q$.
  \item $(\mu, \Phi)$ satisfies
        \begin{alignat}{1}
         Q^\opt(\mu) &= Q(\mu, \Phi), \nonumber \\
         Q_k(\Phi) &\ge Q_{k'}(\Phi), \quad \forall k,k' \in \mI_K
         ~\mbox{s.t.}~\mu_{k'} > 0.
         \label{eq:thm_minimax3}
        \end{alignat}
  \item $(\mu, \Phi)$ satisfies
        \begin{alignat}{1}
         Q_k(\Phi) &\ge Q^\opt(\mu), \quad \forall k \in \mI_K.
         \label{eq:thm_minimax2}
        \end{alignat}
 \end{enumerate}
\end{proposition}
\begin{proof}
 $(1) \Rightarrow (2)$:
 The first line of Eq.~\eqref{eq:thm_minimax3} is obvious.
 Arbitrarily choose $k,k' \in \mI_K$ such that $\mu_{k'} > 0$.
 Let $\mu' \coloneqq \{ \mu_j + \varepsilon(\delta_{j,k} - \delta_{j,k'}) \}_{j=0}^{K-1}$
 with sufficiently small $\varepsilon > 0$;
 then, $\mu' \in \Probmax$ holds.
 Thus, $\varepsilon[Q_k(\Phi) - Q_{k'}(\Phi)] = Q(\mu', \Phi) - Q(\mu, \Phi) \ge 0$,
 i.e., the second line of Eq.~\eqref{eq:thm_minimax3}, holds.

 $(2) \Rightarrow (3)$:
 From the second line of Eq.~\eqref{eq:thm_minimax3},
 $\mu_l Q_k(\Phi) \ge \mu_l Q_l(\Phi)$ holds for any $k,l \in \mI_K$.
 Summing this equation over $l = 0,\dots,K-1$ yields
 \begin{alignat}{1}
  Q_k(\Phi) &\ge \sum_{l=0}^{K-1} \mu_l Q_l(\Phi) = Q(\mu, \Phi)
  = Q^\opt(\mu).
 \end{alignat}

 $(3) \Rightarrow (1)$:
 $Q^\opt(\mu) \le \sum_{k=0}^{K-1} \mu'_k Q_k(\Phi) = Q(\mu', \Phi)$
 holds for any $\mu' \in \Probmax$.
 Thus, from Proposition~\ref{pro:minimax_nas}, $(\mu, \Phi)$ is a minimax solution.
\end{proof}

\begin{cor} \label{cor:minimax_convex}
 Assume $\Prob = \Probmax$.
 A tester is a minimax one of $Q$ if and only if it is optimal for
 the following problem:
 \begin{alignat}{1}
  \begin{array}{ll}
   \mbox{maximize} & \displaystyle Q_\mathrm{min}(\Phi) \coloneqq \min_{k \in \mI_K} Q_k(\Phi) \\
   \mbox{subject~to} & \Phi \in \P. \\
  \end{array}
  \label{eq:minimax_convex}
 \end{alignat}
\end{cor}
\begin{proof}
 ``If'':
 We here replace $\P$ with its closure $\ol{\P}$.
 Let $(\mu^\opt,\Phi^\opt)$ be a minimax solution to $Q$.
 Remark~\ref{remark:minimax} guarantees that such a minimax solution exists.
 Also, let $\tPhi$ be an optimal solution to Eq.~\eqref{eq:minimax_convex}; then,
 $\tPhi$ is also optimal for Eq.~\eqref{eq:minimax_convex} with $\P$ replaced by $\ol{\P}$.
 Thus, $Q_\mathrm{min}(\tPhi) \ge Q_\mathrm{min}(\Phi^\opt) \ge Q^\opt(\mu^\opt)$ holds,
 where the last inequality follows from Eq.~\eqref{eq:thm_minimax2}.
 Therefore, $(\mu^\opt, \tPhi)$ satisfies Statement~(3) of Proposition~\ref{pro:minimax},
 which implies that $(\mu^\opt,\tPhi)$ is a minimax solution to $Q$.

 ``Only if'':
 Let $(\mu^\opt, \Phi^\opt)$ be a minimax solution to $Q$.
 From Eq.~\eqref{eq:thm_minimax2}, we have
 \begin{alignat}{1}
  Q_\mathrm{min}(\Phi^\opt) &\ge Q^\opt(\mu^\opt) = \sup_{\Phi \in \P} Q(\mu^\opt, \Phi)
  \ge \sup_{\Phi \in \P} Q_\mathrm{min}(\Phi),
  \label{eq:minimax_convex2}
 \end{alignat}
 which gives that $\Phi^\opt$ is optimal for Eq.~\eqref{eq:minimax_convex}.
\end{proof}

\subsection{Symmetry}

Using a similar argument as in Sec.~\ref{sec:symmetry},
we can see that if Problem~\eqref{prob:Pm} has a certain symmetry
and at least one minimax solution exists,
then there exists a symmetric minimax solution.

\begin{define}
 Let $\mG$ be a group.
 We will call Problem~\eqref{prob:Pm} \termdef{$\mG$-symmetric} if
 the following conditions hold:
 (a) there exist group actions of $\mG$ on $\mI_M$, $\mI_J$, $\mI_K$, and $\Her_\tV$;
 (b) the action of $\mG$ on $\Her_\tV$ is linearly isometric;
 and (c) Eq.~\eqref{eq:sym_eval_ab} and
 \begin{alignat}{2}
  g \b c_{k,m} &= c_{g \b k, g \b m}, &\quad &\forall g \in \mG, ~k \in \mI_K, ~m \in \mI_M,
  \nonumber \\
  \{ \mu_{g \b k} \}_{k=0}^{K-1} &\in \Prob, &\quad &\forall g \in \mG, ~\mu \in \Prob
  \label{eq:sym_eval_mm}
 \end{alignat}
 hold.
\end{define}

For any group $\mG$ and $\mu \in \Prob$, let
\begin{alignat}{1}
 \mu^\d &\coloneqq \left\{ \mu^\d_k \coloneqq \frac{1}{|\mG|} \sum_{g \in \mG} \mu_{g \b k}
 \right\}_{k=0}^{K-1}.
 \label{eq:mub}
\end{alignat}
We can easily verify $\mu^\d \in \Prob$ and
\begin{alignat}{1}
 \mu^\d_k &= \mu^\d_{g \b k}, \quad \forall g \in \mG, ~k \in \mI_K.
 \label{eq:sym_mu}
\end{alignat}

Analogously to Theorem~\ref{thm:sym_Phi}, the following theorem can be proved.
\begin{thm} \label{thm:sym_minimax}
 Let $\mG$ be a group.
 Assume that Problem~\eqref{prob:Pm} is $\mG$-symmetric;
 then, for any minimax solution $(\mu,\Phi)$ to $Q$,
 $(\mu^\d, \Phi^\d)$ defined by Eqs.~\eqref{eq:Phib} and \eqref{eq:mub} is
 also a minimax solution to $Q$.
\end{thm}
\begin{proof}
 $\Phi^\d \in \P$ holds from Lemma~\ref{lemma:sym_kappa}.
 From Proposition~\ref{pro:minimax_nas},
 it suffices to show $Q^\opt(\mu^\d) \le Q(\mu',\Phi^\d)$ for any $\mu' \in \Prob$.
 In what follows, we show $Q(\mu',\Phi^\d) \ge Q^\opt(\mu)$ $~(\forall \mu' \in \Prob)$
 and $Q^\opt(\mu) \ge Q^\opt(\mu^\d)$.

 First, we show $Q(\mu',\Phi^\d) \ge Q^\opt(\mu)$ $~(\forall \mu' \in \Prob)$.
 We have that for any $\mu' \in \Prob$,

 \begin{alignat}{1}
  Q(\mu',\Phi^\d) &= \sumk \mu'_k \summ \frac{1}{|\mG|} \sumg
  \braket{\inv{g} \b \Phi_{g \b m}, c_{k,m}} \nonumber \\
  &= \frac{1}{|\mG|} \sumg \sumk \mu'_k \summ
  \braket{\Phi_{g \b m}, g \b c_{k,m}} \nonumber \\
  &= \frac{1}{|\mG|} \sumg \sum_{k'=0}^{K-1} \mu'_{\inv{g} \b k'} \sum_{m'=0}^{M-1}
  \braket{\Phi_{m'}, c_{k',m'}} \nonumber \\
  &= \frac{1}{|\mG|} \sumg \sum_{k'=0}^{K-1} \mu'_{\inv{g} \b k'} Q_{k'}(\Phi) \ge Q^\opt(\mu),
 \end{alignat}
 where $m' \coloneqq g \b m$ and $k' \coloneqq g \b k$.
 The inequality follows from
 $\{ \mu'_{\inv{g} \b k'} \}_{k'=0}^{K-1} \in \Prob$
 and $Q(\mu',\Phi) \ge Q^\opt(\mu)$ $~(\forall \mu' \in \Prob)$
 (see Proposition~\ref{pro:minimax_nas}).

 Next, we show $Q^\opt(\mu) \ge Q^\opt(\mu^\d)$.
 Let $\mu^\g \coloneqq \{ \mu^\g_k \coloneqq \mu_{g \b k} \}_{k=0}^{K-1}$; then,
 we have that for any $g \in \mG$,
 \begin{alignat}{1}
  Q^\opt[\mu^\g] &= \sup_{\Phi' \in \P} \sumk \mu_{g \b k}
  \summ \braket{\Phi'_m, c_{k,m}} \nonumber \\
  &= \sup_{\Phi' \in \P} \sum_{k'=0}^{K-1} \mu_{k'}
  \summ \braket{\Phi'_m, c_{\inv{g} \b k',m}} \nonumber \\
  &= \sup_{\Phi' \in \P} \sum_{k'=0}^{K-1} \mu_{k'}
  \sum_{m'=0}^{M-1} \braket{g \b \Phi'_{\inv{g} \b m'}, c_{k',m'}} \nonumber \\
  &\le \sup_{\Phi'' \in \P} \sum_{k'=0}^{K-1} \mu_{k'}
  \sum_{m'=0}^{M-1} \braket{\Phi''_{m'}, c_{k',m'}} \nonumber \\
  &= Q^\opt(\mu),
  \label{eq:sym_minimax_Qopt_eta}
 \end{alignat}
 where $k' \coloneqq g \b k$ and $m' \coloneqq g \b m$.
 The inequality follows from
 $\{ g \b \Phi'_{\inv{g} \b m'} \}_{m'=0}^{M-1} = \Phi^{\prime(\inv{g})} \in \P$
 (see Lemma~\ref{lemma:sym_kappa}).
 Thus, we have
 \begin{alignat}{1}
  Q^\opt(\mu) &\ge \frac{1}{|\mG|} \sumg Q^\opt[\mu^\g]
  = \frac{1}{|\mG|} \sumg \sup_{\Phi' \in \P} \sumk \mu^\g_k Q_k(\Phi') \nonumber \\
  &\ge \sup_{\Phi' \in \P} \frac{1}{|\mG|} \sumg \sumk \mu^\g_k Q_k(\Phi') = Q^\opt(\mu^\d).
 \end{alignat}
\end{proof}

\section{Examples} \label{sec:example}

Using several examples, we show how our approach can be used to
 extract some non-trivial properties of optimal discrimination.

\subsection{Some cases in which a tester with maximally entangled pure states can be optimal}
\label{subsec:example_input}

In this subsection, as applications of Theorem~\ref{thm:opt_two} and
Proposition~\ref{pro:input_entangle_symT},
we provide some examples in which there exists a tester with maximally entangled pure states that is optimal
for Problem~\eqref{prob:P}.

\begin{cor} \label{cor:input_entangle_binary}
 Let us consider the problem of finding minimum-error discrimination of two combs
 $\hcE_0$ and $\hcE_1$ with prior probabilities $p_0$ and $p_1$, respectively.
 There exists a tester with maximally entangled pure states that is optimal for Problem~\eqref{prob:PG}
 if and only if $|\Delta| \in \Lin(\Chn_\tV)$ holds,
 where $\Delta \coloneqq p_0 \cE_0 - p_1 \cE_1$
 and $|\Delta| \coloneqq \sqrt{\Delta^\dagger \Delta}$.
\end{cor}
\begin{proof}
 In the minimum-error discrimination with $R \coloneqq 2$,
 $\mC \coloneqq \mCG$, and $\S \coloneqq \S_\Psi$ [see Eq.~\eqref{eq:SPsi}],
 Problem~\eqref{prob:D} is rewritten as
 \begin{alignat}{1}
  \begin{array}{ll}
   \mbox{minimize} & \displaystyle \Tr \chi / \Pro{t=1}{T} N_\Vt \\
   \mbox{subject~to} & \chi \ge p_0 \cE_0, \quad \chi \ge p_1 \cE_1 \\
  \end{array}
 \end{alignat}
 with $\chi \in \Her_\tV$.
 Let $\chi^\opt$ be its optimal solution
 and $X \coloneqq 2\chi^\opt - (p_0 \cE_0 + p_1 \cE_1)$;
 then, it follows that $X$ minimizes $\Tr X$ subject to $X \ge \Delta$ and $X \ge -\Delta$.
 Thus, we have $X = |\Delta|$ and $\chi^\opt = \frac{1}{2}(p_0 \cE_0 + p_1 \cE_1 + |\Delta|)$.
 Since $\hcE_0$ and $\hcE_1$ are combs,
 $\chi^\opt \in \Lin(\Chn_\tV)$ holds if and only if
 $|\Delta| \in \Lin(\Chn_\tV)$ holds.
 Therefore, Theorem~\ref{thm:opt_two} completes the proof.
\end{proof}
Note that Corollary~3 of Ref.~\cite{Jen-Pla-2016} states that,
in the problem of finding (single-shot) minimum-error discrimination of two channels
$\hLambda_0, \hLambda_1 \in \Chn(\V,\W)$ with prior probabilities $p_0$ and $p_1$, respectively,
there exists a tester with maximally entangled pure states that is optimal if and only if
$\Trp{\W} |p_0 \Lambda_0 - p_1 \Lambda_1| \propto \I_\V$ holds.
One can immediately verify that this is the special case of
Corollary~\ref{cor:input_entangle_binary} with $T \coloneqq 1$,
$\hcE_0 \coloneqq \hLambda_0$, and $\hcE_1 \coloneqq \hLambda_1$.

\begin{cor} \label{cor:input_entangle_HT}
 Let us consider the direct product, $\mG \coloneqq \mH_1 \times \dots \times \mH_T$, of
 some groups $\mH_1,\dots,\mH_T$.
 Assume that the group action of $\mG$ on $\Her_\tV$ is expressed as
 \begin{alignat}{1}
  (h_1,\dots,h_T) \b \Endash &\coloneqq \Ad_{U_{T,h_T} \ot U'_{T,h_T} \ot \cdots
  \ot U_{1,h_1} \ot U'_{1,h_1}}, \nonumber \\
  &\qquad (h_1,\dots,h_T) \in \mG,
 \end{alignat}
 where, for each $t \in \range{1}{T}$, $\mH_t \ni h_t \mapsto U_{t,h_t} \in \Uni_\Wt$ and
 $\mH_t \ni h_t \mapsto U'_{t,h_t} \in \Uni_\Vt$ are
 projective representations of $\mH_t$.
 Also, assume that $\S = \SG$ holds and that Problem~\eqref{prob:P} is $\mG$-symmetric.
 If $h_t \mapsto U'_{t,h_t}$ is irreducible for any $t \in \range{1}{T}$,
 then the optimal value of Problem~\eqref{prob:P}
 remains unchanged if $\S$ is replaced by $\S_\Psi$.
\end{cor}
\begin{proof}
 Proposition~\ref{pro:input_entangle_symT}
 with $\mH^{(t)} \coloneqq \{ (h_1,\dots,h_T) : h_{t'} = e_{t'} ~(\forall t' \neq t), ~h_t \in \mH_t \}$
 concludes the proof, where $e_{t'}$ is the identity element of $\mH_{t'}$.
\end{proof}

We provide two simple applications of this corollary.
Note that each of them is a special case of Problem~\eqref{prob:PG}.
They can be readily extended to Problem~\eqref{prob:P} with $\S = \SG$
(such as the problem shown in Example~\ref{ex:restricted_testers}).

\subsubsection{$T$-shot discrimination of symmetric channels}

Let us first consider the problem of finding optimal inconclusive discrimination
of $R$ channels, $\{ \hLambda_r \}_{r=0}^{R-1}$, discussed in Example~\ref{ex:sym_self}.
Assume that $\Phi^\g \in \mCG$ $~(\forall \Phi \in \mCG)$ holds; then,
this problem is $\mH^T$-symmetric.
It immediately follows from Corollary~\ref{cor:input_entangle_HT}, with $G \coloneqq \mH^T$, that
there exists a tester with maximally entangled pure states that is optimal for Problem~\eqref{prob:PPiInc}
if $h \mapsto \tU_h$ $~(h \in \mH)$ is irreducible.
In what follows, we present two typical examples.

The first example is the case in which $\hLambda_0,\dots,\hLambda_{R-1}$ are
teleportation-covariant channels \cite{Pir-Lau-Ott-Ban-2017,Pir-Lup-2017}.
Let $\mH$ be a group and $\{ \tU_h \}_{h \in \mH}$ be the set of unitary operators
generated by the Bell detection in a teleportation process.
Assume that a collection of channels $\{ \hLambda_r \}_{r=0}^{R-1}$ is teleportation-covariant,
i.e., there exists a projective unitary representation $h \mapsto U_h$ such that
\begin{alignat}{2}
 \hLambda_r &= \Ad_{U_h} \c \hLambda_r \c \Ad_{\tU_h}, &\quad &\forall r \in \mI_R, ~h \in \mH.
\end{alignat}
It is easily seen that
$\Phi^\g \in \mCG$ $~(\forall \Phi \in \mCG)$ holds and
$h \mapsto \tU_h$ $~(h \in \mH)$ is irreducible, and thus
there exists a tester with maximally entangled pure states that is optimal.
Note that its minimum-error version has been discussed in Ref.~\cite{Zhu-Pir-2020}.

The second example is the case in which $T = 1$ holds and
$\hLambda_0,\dots,\hLambda_{R-1} \in \Chn(\V,\W)$ are
unital qubit channels, i.e., unital channels with $\NV = \NW = 2$%
\footnote{A channel $\hLambda \in \Chn(\V,\W)$ is called unital
if $\hLambda(\I_\V / N_\V) = \I_\W / N_\W$
(or, equivalently, $\Trp{\V} \Lambda / N_\V = \I_\W / N_\W$) holds.
Examples of unital channels are mixed unitary qubit channels
and Schur channels \cite{Wat-2018}.}.
For any unital qubit channel $\hLambda \in \Chn(\V,\W)$,
since $\Trp{\W} \Lambda \propto \I_\V$ and $\Trp{\V} \Lambda \propto \I_\W$ hold,
$\Lambda$ is expressed in the form
\begin{alignat}{1}
 \Lambda &=
 \begin{bmatrix}
  s_0 & s_1 & t_1 & t_0 \\
  s_1^* & s_2 & t_2 & -t_1 \\
  t_1^* & t_2^* & s_2 & -s_1 \\
  t_0^* & -t_1^* & -s_1^* & s_0 \\
 \end{bmatrix}
 \label{eq:CLambda_unital}
\end{alignat}
with $s_0,s_2 \in \Real_+$ and $s_1,t_0,t_1,t_2 \in \Complex$.
We can easily verify that such $\Lambda$ satisfies
$\Ad_{\Sa \ot \Sa}(\Lambda) = \Lambda$, where $\Sa$ is the anti-unitary operator
defined by
\begin{alignat}{1}
 \Ad_{\Sa}(x) &= \Ad_S(x^\T), \quad x \in \Her_\V, \nonumber \\
 S &\coloneqq
 \begin{bmatrix}
  0 & 1 \\
  -1 & 0 \\
 \end{bmatrix}.
 \label{eq:Sa}
\end{alignat}
Let us consider a group $\mH \coloneqq \{ e, \th \}$ and its projective representation
$\mH \ni h \mapsto U_h \in \Uni_\V$ with $U_e \coloneqq \I_\V$ and $U_{\th} \coloneqq \Sa$;
then, we have
\begin{alignat}{2}
 \hLambda_r &= \Ad_{U_h} \c \hLambda_r \c \Ad_{U_h^\T}, &\quad &\forall r \in \mI_R, ~h \in \mH.
\end{alignat}
It follows that
$\Phi^\g \in \mCG$ $~(\forall \Phi \in \mCG)$ holds and
the representation $h \mapsto U_h$ is irreducible,
and thus there exists a tester with maximally entangled pure states that is optimal.

\subsubsection{Determination of the modulo sum of independent rotations}

We next consider the problem of determining the modulo sum of $T$ independent rotations.
Let $\mH \coloneqq \{ g_{j,k} \}_{(j,k)=(0,0)}^{(d-1,d-1)}$ be
the generalized Pauli group (or discrete Heisenberg-Weyl group),
whose projective representation is
\begin{alignat}{1}
 \mH \ni g_{j,k} &\mapsto U'_{g_{j,k}} \coloneqq
 \sum_{i=0}^{\NV-1} \exp\left( \i \frac{2\pi ik}{\NV} \right)
 \ket{i \oplus j}\bra{i} \in \Uni_\V,
\end{alignat}
where $\i \coloneqq \sqrt{-1}$, $\V$ is a system, $\oplus$ is addition modulo $\NV$,
and $\{ \ket{i} \}_{i=0}^{\NV-1}$ is the standard basis of $\V$.
$j$ and $k$ can be, respectively, interpreted as the amounts of $\x$- and $\z$- rotations.
Note that this representation is irreducible.
We consider the following process
\begin{alignat}{1}
 \tLambda_{(h_1,\dots,h_T)} &\coloneqq \hLambda^{(T)}_{h_T} \ast \hLambda^{(T-1)}_{h_{T-1}}
 \ast \cdots \ast \hLambda^{(1)}_{h_1},
 \quad (h_1,\dots,h_T) \in \mH^T,
\end{alignat}
where, for each $t \in \range{1}{T}$,
$\{ \hLambda^{(t)}_h \}_{h \in \mH} \subset \Chn(\V,\Wt)$ is
a collection of channels satisfying
\begin{alignat}{2}
 \hLambda^{(t)}_h &= \Ad_{U_{t,h}} \c \hLambda^{(t)}_e \c \Ad_{U'_h},
 &\quad \forall h &\in \mH,
\end{alignat}
$\Wt$ is a system,
and $\mH \ni h \mapsto U_{t,h} \in \Uni_\Wt$ is a projective representation.
Suppose that a process $\tLambda_{(h_1,\dots,h_T)}$ is given,
where $(h_1,\dots,h_T)$ is uniformly randomly chosen from $\mH^T$,
and that we want to determine the modulo sum of $\z$-rotations $\bigoplus_{t=1}^T \z(h_t)$,
where $\z$ is defined as $\z(g_{j,k}) \coloneqq k$ $~(g_{j,k} \in \mH)$.
This problem is formulated as the problem of finding optimal discrimination of the
processes $\{ \tcE_m \}_{m=0}^{\NV-1}$, where
\begin{alignat}{1}
 \tcE_m &\coloneqq \frac{1}{|\mH^T|} \sum_{\left\{ (h_1,\dots,h_T) \in \mH^T: \bigoplus_{t=1}^T \z(h_t) = m \right\}}
 \tLambda_{(h_1,\dots,h_T)}.
\end{alignat}
To simplify the discussion, we here consider the minimum-error strategy,
which is written as Problem~\eqref{prob:PG} with
\begin{alignat}{5}
 M &\coloneqq \NV, &&\quad
 J &\coloneqq 0, &&\quad
 c_m &\coloneqq \Choi_{\tcE_m}.
 \label{eq:param_rotation}
\end{alignat}
Let the group actions of $\mH^T$ on $\mI_M$
and $\Her_\tV$ be, respectively, defined as
\begin{alignat}{1}
 (h_1,\dots,h_T) \b m &\coloneqq
 \left\{ \bigoplus_{t=1}^T \z(h_t) \right\} \oplus m, \nonumber \\
 (h_1,\dots,h_T) \b \Endash &\coloneqq \Ad_{U_{T,h_T} \ot U'_{h_T} \ot \cdots
 \ot U_{1,h_1} \ot U'_{h_1}}
\end{alignat}
for any $(h_1,\dots,h_T) \in \mH^T$; then,
one can easily verify that this problem is $\mH^T$-symmetric.
Thus, there exists a tester with maximally entangled pure states that is optimal.

\subsection{Single-shot discrimination of cyclic unital qubit channels}

It is known that, in several state discrimination problems for highly symmetric states,
their optimal values can be obtained analytically.
Similarly, it is expected that we can analytically obtain the optimal values in
several process discrimination problems with high symmetry.

In this subsection, let us consider the following two problems:
the problem of obtaining single-shot optimal inconclusive discrimination
[i.e., Problem~\eqref{prob:PPiInc}] for $R$ unital qubit channels
$\{ \hLambda_r \}_{r=0}^{R-1} \subset \Chn(\V,\W)$
and its minimax version [i.e., Problem~\eqref{prob:PPiIncMinimax}].
Let $U$ be a unitary operator on $\W$ satisfying $U^R = \I_\W$
and $U^r \neq \I_\W$ for any $1 \le r < R$.
We choose the eigenvectors of $U$ as the standard basis of $\W$.
Assume
\begin{alignat}{2}
 \hLambda_{r \oplus_R 1} &= \Ad_U \c \hLambda_r, &\quad &\forall r \in \mI_R, \nonumber \\
 \hLambda_0(\rho^\T) &= [\hLambda_0(\rho)]^\T, &\quad &\forall \rho \in \Pos_\V,
 \label{eq:ex_sym_Lambda}
\end{alignat}
which means $\Lambda_{r \oplus_R 1} = \Ad_{U \ot \I_\V}(\Lambda_r)$
$~(\forall r \in \mI_R)$ and $\Lambda_0^\T = \Lambda_0$.
Also, assume that, in Problem~\eqref{prob:PPiInc}, the prior probabilities are all equal.

The symmetry expressed by Eq.~\eqref{eq:ex_sym_Lambda}
can be represented by group actions as follows.
Let
\begin{alignat}{1}
 \mG &\coloneqq \{ h^r h_*^k : r \in \mI_R, ~k \in \mI_2 \}
\end{alignat}
be a dihedral group of order $2R$,
generated by a `rotation' $h$ and a `reflection' $h_*$
that satisfy $h^R = e = h_*^2$ and $h_* h h_* = \inv{h}$.
We consider the actions of $\mG$ on $\mI_M$ and $\Her_\WV$ defined by
\begin{alignat}{1}
 h \b m &\coloneqq
 \begin{dcases}
  m \oplus_R 1, & m < R, \\
  R, & m = R, \\
 \end{dcases} \nonumber \\
 h_* \b m &\coloneqq m, \nonumber \\
 h \b x &\coloneqq \Ad_{U \ot \I_\V}(x), \nonumber \\
 h_* \b x &\coloneqq x^\T
\end{alignat}
for any $m \in \mI_M$ and $x \in \Her_\WV$.
Also, in Problem~\eqref{prob:PPiIncMinimax},
the action of $\mG$ on $\mI_K = \mI_J = \mI_R$ is defined by
$h \b r \coloneqq r \oplus_R 1$ and $h_* \b r \coloneqq r$ $~(\forall r \in \mI_R)$.
One can easily verify that Problems~\eqref{prob:PPiInc} and \eqref{prob:PPiIncMinimax} satisfying
Eq.~\eqref{eq:ex_sym_Lambda} are $\mG$-symmetric.

In Problem~\eqref{prob:PPiIncMinimax}, Theorem~\ref{thm:sym_minimax} guarantees that
there exists a minimax solution $(\mu^\d, \Phi^\d)$
satisfying Eqs.~\eqref{eq:sym_Phi} and \eqref{eq:sym_mu}.
This gives $\mu^\d_{r \oplus_R 1} = \mu^\d_r$, i.e.,
the prior probabilities $\mu^\d_0,\dots,\mu^\d_{R-1}$ are all equal.
Thus, Problem~\eqref{prob:PPiIncMinimax} is essentially the same as Problem~\eqref{prob:PPiInc}.
In what follows, we focus on solving Problem~\eqref{prob:PPiInc}.
Note that since each channel is unital, there exists a tester with maximally entangled
pure states that is optimal for Problem~\eqref{prob:PPiInc},
as shown in the previous subsection.
This fact reduces Problem~\eqref{prob:PPiInc} to the corresponding state discrimination problem.
However, solving this state discrimination problem is as hard as solving Problem~\eqref{prob:PPiInc}.

Let us consider Problem~\eqref{prob:DInc}.
We can see that $\lambdaSG(\chi) = \lambdamax(\Trp{\W} \chi)$
holds for any $\chi \in \Her_\WV$,
where $\lambdamax(X)$ is the largest eigenvalue of $X$.
From Corollary~\ref{cor:chi_comb}, without loss of generality,
we assume that an optimal solution, $\chi$, is in $\Lin(\Chn_\WV)$.
Corollary~\ref{cor:sym_chi} asserts that $(\chi^\d,q^\d)$ is also
an optimal solution.
From Eq.~\eqref{eq:sym_chi}, $h \b \chi^\d = \chi^\d$ and $h_* \b \chi^\d = \chi^\d$ hold.
One can also easily check $\chi^\d \in \Lin(\Chn_\WV)$.
Thus, Problem~\eqref{prob:DInc} can be rewritten as
\begin{alignat}{1}
 \begin{array}{ll}
  \mbox{minimize} & \lambdamax(\Trp{\W} \chi) - q \pinc \\
  \mbox{subject~to} & \chi \ge \zeta_0, ~\chi \ge \zeta_1, ~\chi \in \Lin(\Chn_\WV), \\
  & \Ad_{U \ot \I_\V}(\chi) = \chi, ~\chi^\T = \chi \\
 \end{array}
 \label{eq:ProbIncEx}
\end{alignat}
with $\chi \in \Her_\WV$ and $q \in \Real_+$,
where $\zeta_0 \coloneqq \Lambda_0 / R$ and
$\zeta_1 \coloneqq q \sumr \Lambda_r / R = q \sumr \Ad_{U^r \ot \I_\V}(\Lambda_0) / R$.
One should remember that $\zeta_1$ is a function of $q$.
Note that any feasible solution to Problem~\eqref{eq:ProbIncEx}
satisfies $\chi \ge \Lambda_r / R$ $~(\forall r \in \mI_R)$,
which follows from $\chi \ge \zeta_0$ and $\Ad_{U \ot \I_\V}(\chi) = \chi$.

From $\Trp{\W} \Lambda_0 \propto \I_\V$, $\Trp{\V} \Lambda_0 \propto \I_\W$,
and $\Lambda_0^\T = \Lambda_0$,
$\zeta_0$ and $\zeta_1$ can be expressed in the form
\begin{alignat}{2}
 \zeta_0 &=
 \begin{bmatrix}
  s_0 & s_1 & t_1 & t_0 \\
  s_1 & s_2 & t_2 & -t_1 \\
  t_1 & t_2 & s_2 & -s_1 \\
  t_0 & -t_1 & -s_1 & s_0 \\
 \end{bmatrix}, &\quad
 \zeta_1 &= qR
 \begin{bmatrix}
  s_0 & s_1 & 0 & 0 \\
  s_1 & s_2 & 0 & 0 \\
  0 & 0 & s_2 & -s_1 \\
  0 & 0 & -s_1 & s_0 \\
 \end{bmatrix}
 \label{eq:CLambda0_unital}
\end{alignat}
with some $s_k,t_k \in \Real$ $~(k \in \mI_3)$ [see Eq.~\eqref{eq:CLambda_unital}].
They are rewritten as
\begin{alignat}{2}
 \zeta_l &= \Theta
 \begin{bmatrix}
  A_l & B_l \\
  B_l & A_l \\
 \end{bmatrix}
 \Theta^\dagger, &\quad
 &\forall l \in \mI_2,
 \label{eq:CLambda0_unitalJ}
\end{alignat}
where
\begin{alignat}{3}
 \Theta &\coloneqq
 \begin{bmatrix}
  1 & 0 & 0 & 0 \\
  0 & 1 & 0 & 0 \\
  0 & 0 & 0 & 1 \\
  0 & 0 & -1 & 0 \\
 \end{bmatrix}, &\quad
 A_l &\coloneqq
 \begin{bmatrix}
  \ts_{l,0} & \ts_{l,1} \\
  \ts_{l,1} & \ts_{l,2} \\
 \end{bmatrix}, &\quad
 B_l &\coloneqq
 \begin{bmatrix}
  -\tit_{l,0} & \tit_{l,1} \\
  \tit_{l,1} & \tit_{l,2} \\
 \end{bmatrix}
\end{alignat}
and
\begin{alignat}{5}
 \ts_{0,k} &\coloneqq s_k, &\quad
 \tit_{0,k} &\coloneqq t_k, &\quad
 \ts_{1,k} &\coloneqq qRs_k, &\quad
 \tit_{1,k} &\coloneqq 0, &\quad &k \in \mI_3.
\end{alignat}
From $\chi \in \Lin(\Chn_\WV)$, $\Ad_{U \ot \I_\V}(\chi) = \chi$, and $\chi^\T = \chi$,
$\chi$ is expressed in the form
\begin{alignat}{1}
 \chi &=
 \begin{bmatrix}
  x+z & y & 0 & 0 \\
  y & x-z & 0 & 0 \\
  0 & 0 & x-z & -y \\
  0 & 0 & -y & x+z \\
 \end{bmatrix}
 = \Theta (X \oplus X) \Theta^\dagger
 \label{eq:chi_xyz}
\end{alignat}
with some $x,y,z \in \Real$, where
\begin{alignat}{2}
 X &\coloneqq
 \begin{bmatrix}
  x + z & y \\
  y & x - z \\
 \end{bmatrix}.
 \label{eq:chi_zxy_X}
\end{alignat}

$(\chi, q)$ is a feasible solution to Problem~\eqref{eq:ProbIncEx} if and only if
$\chi$ is expressed in the form of Eq.~\eqref{eq:chi_xyz}
and $\chi \ge \zeta_0$ and $\chi \ge \zeta_1$ hold.
Here, to derive a necessary and sufficient condition for $\chi \ge \zeta_l$,
we obtain the eigenvalues of $\chi - \zeta_l$.
From $\det \Theta = 1$, we have
\begin{alignat}{2}
 \det(\chi - \zeta_l - \lambda \I_\WV)
 &= \det
 \begin{bmatrix}
  X - A_l - \lambda \I_\V & - B_l \\
  - B_l & X - A_l - \lambda \I_\V \\
 \end{bmatrix}.
\end{alignat}
It follows that any squared matrices $C$ and $D$ with the same size satisfy
\begin{alignat}{1}
 \det
 \begin{bmatrix}
  C & D \\
  D & C \\
 \end{bmatrix}
 &= \det
 \begin{bmatrix}
  C + D & D + C \\
  D & C \\
 \end{bmatrix}
 = \det
 \begin{bmatrix}
  C + D & \zero \\
  D & C - D \\
 \end{bmatrix} \nonumber \\
 &= \det(C+D) \cdot \det(C-D).
\end{alignat}
Thus, by solving the equation $\det(X - A_l - \lambda \I_\V \pm B_l) = 0$,
we obtain the eigenvalues of $\chi - \zeta_l$ as follows:
\begin{alignat}{1}
 \lambda_{l,k,\pm} &= x - x_{l,k} \pm \sqrt{(y - y_{l,k})^2 + (z - z_{l,k})^2}, \quad k \in \mI_2,
\end{alignat}
where
\begin{alignat}{1}
 x_{l,k} &\coloneqq \frac{1}{2}\left[ \ts_{l,0} + \ts_{l,2} + (-1)^k (\tit_{l,0} - \tit_{l,2}) \right], \nonumber \\
 y_{l,k} &\coloneqq \ts_{l,1} - (-1)^k \tit_{l,1}, \nonumber \\
 z_{l,k} &\coloneqq \frac{1}{2}\left[ \ts_{l,0} - \ts_{l,2} + (-1)^k (\tit_{l,0} + \tit_{l,2}) \right].
\end{alignat}
Since $\lambda_{l,k,+} \ge \lambda_{l,k,-}$ holds for each $k \in \mI_2$,
$\chi \ge \zeta_l$ holds if and only if $\lambda_{l,0,-} \ge 0$ and $\lambda_{l,1,-} \ge 0$ hold.

For each $v \in \Real^3$, let $v_\x, v_\y, v_\z$ be, respectively,
the $\x$-, $\y$-, and $\z$- components of $v$,
and
\begin{alignat}{1}
 \mN_v &\coloneqq \left\{ v' \in \Real^3 : v'_\x - v_\x \ge \sqrt{(v'_\y - v_\y)^2 + (v'_\z - v_\z)^2} \right\}.
\end{alignat}
It follows that $\mN_v$ is a cone with its apex at the point $v$.
Let $u \coloneqq (x, y, z)$ and $\upsilon^{l,k} \coloneqq (x_{l,k}, y_{l,k}, z_{l,k})$;
then, since $\lambda_{l,k,-} \ge 0$ is equivalent to $u \in \mN_{\upsilon^{l,k}}$,
Problem~\eqref{eq:ProbIncEx} is rewritten as
\begin{alignat}{1}
 \begin{array}{ll}
  \mbox{minimize} & 2u_\x - q \pinc \\
  \mbox{subject~to} & u \in \mN_{\upsilon^{0,0}} \cap \mN_{\upsilon^{0,1}} \cap \mN_{\upsilon^{1,1}(q)},
 \end{array}
 \label{eq:ProbUnital_u}
\end{alignat}
where we use $\upsilon^{1,0} = \upsilon^{1,1}$, which is given by
$\tit_{1,0} = \tit_{1,1} = \tit_{1,2} = 0$.
To emphasize that $\upsilon^{1,1}$ is a function of $q$,
we denote it by $\upsilon^{1,1}(q)$.
It is easily seen that
the optimal value of Problem~\eqref{eq:ProbUnital_u} is equal to
\begin{alignat}{1}
 P^\opt(\pinc) &\coloneqq \inf_{q \in \Real_+} [2u^\opt_\x(q) - q\pinc],
 \label{eq:unital_Popt}
\end{alignat}
where, for each $q$, $u^\opt_\x(q)$ is the $\x$-component of the point
$u^\opt(q) \in \mN_{\upsilon^{0,0}} \cap \mN_{\upsilon^{0,1}} \cap \mN_{\upsilon^{1,1}(q)}$
that has the minimum $\x$-component.
Note that Eq.~\eqref{eq:unital_Popt} implies that
$- P^\opt(\pinc)$ is the Legendre transformation of $2u^\opt_\x(q)$.

We should note that Problem~\eqref{eq:ProbUnital_u} can also be expressed as
\begin{alignat}{1}
 \begin{array}{ll}
  \mbox{minimize} & \Tr X - q \pinc \\
  \mbox{subject~to} & X \ge \nu^{0,0}, ~X \ge \nu^{0,1}, ~X \ge \nu^{1,1}
 \end{array}
 \label{eq:ProbUnital}
\end{alignat}
with two-dimensional symmetric matrix $X$ given by Eq.~\eqref{eq:chi_zxy_X}
and $q \in \Real_+$, where
\begin{alignat}{1}
 \nu^{l,k} &\coloneqq
 \begin{bmatrix}
  x_{l,k} + z_{l,k} & y_{l,k} \\
  y_{l,k} & x_{l,k} - z_{l,k} \\
 \end{bmatrix}.
 \label{eq:X_nu}
\end{alignat}
If $q$ is fixed, then Problem~\eqref{eq:ProbUnital}
can be regarded as the dual of a qubit state discrimination problem
and thus can be analytically solved \cite{Ha-Kwo-2013}.
One can interpret Problem~\eqref{eq:ProbUnital_u} as
the geometrical representation of Problem~\eqref{eq:ProbUnital}
(see \cite{Dec-Ter-2010,Nak-Kat-Usu-2015-inc}).

As a simple example, we now consider the case $s_1 = t_1 = 0$.
Note that this case is equivalent to the case in which
$\hLambda_0$ is a Pauli channel.
Since $\upsilon^{l,k}_\y = y_{l,k} = 0$ holds, we need only to consider
in the plane $\y = 0$.
Figure~\ref{fig:result_xz} shows a geometrical representation of
Problem~\eqref{eq:ProbUnital_u} in the case of
$R = 3$, $s_0 = t_0 = 0.3 / R$, $s_2 = 0.7 / R$, and $t_2 = 0.1 / R$.
\begin{figure}[tb]
 \centering
 \InsertPDF{0.6}{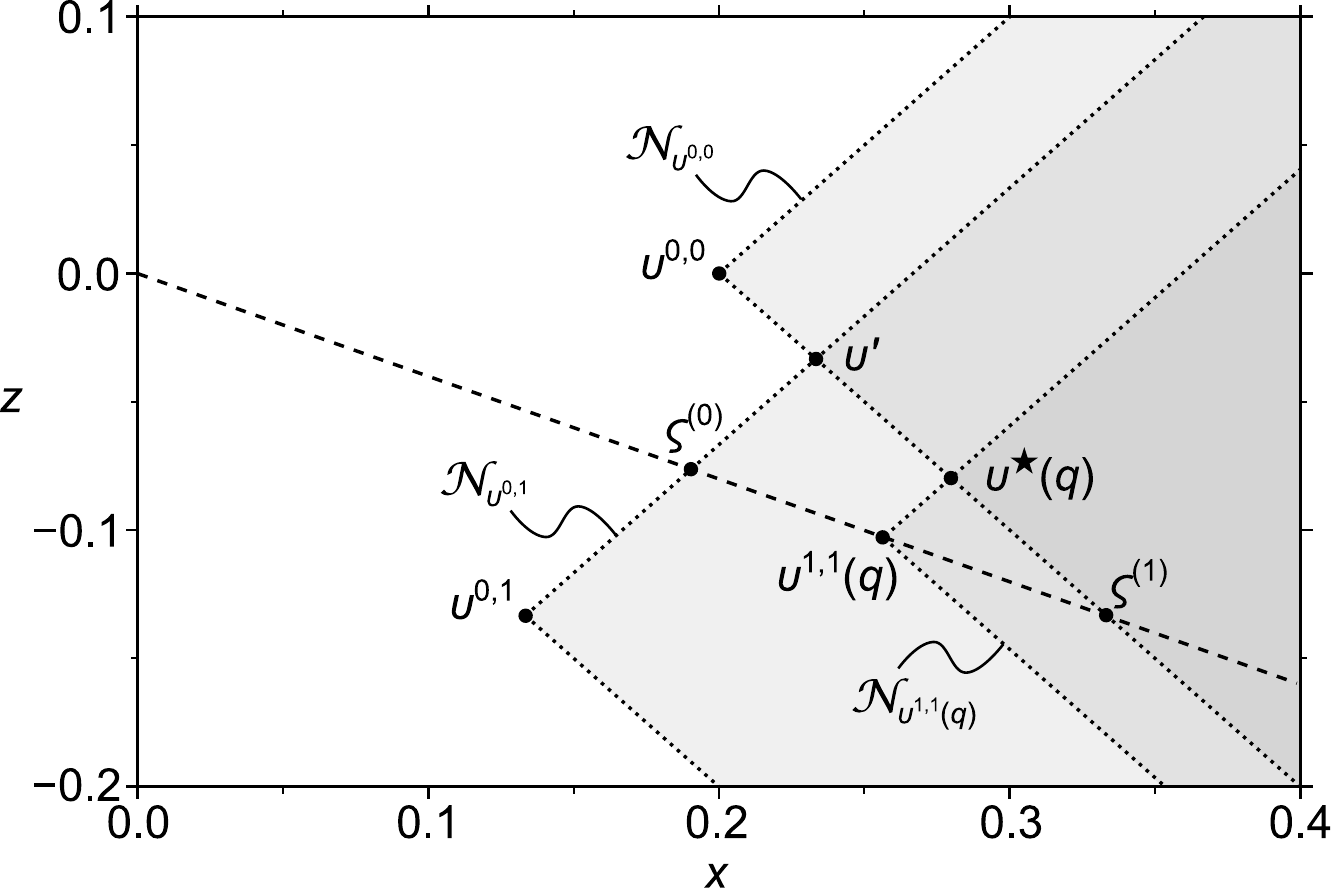}
 \caption{Geometrical representation of Problem~\eqref{eq:ProbUnital_u}
 in the case of $R = 3$, $s_0 = t_0 = 0.3 / R$, $s_1 = t_1 = 0$, $s_2 = 0.7 / R$,
 and $t_2 = 0.1 / R$.
 $\upsilon^{0,0} = (0.2,0,0)$, $\upsilon^{0,1} = (2/15,0,-2/15)$,
 and $\upsilon^{1,1}(q) = (0.5q,0,-0.2q)$ are in the plane $\y = 0$, and so is $u^\opt(q)$.
 $\upsilon^{1,1}(q)$ lies on the dashed straight line.
 The three cones $\mN_{\upsilon^{0,0}}$, $\mN_{\upsilon^{0,1}}$, and $\mN_{\upsilon^{1,1}(q)}$
 are shaded in gray.
 Note that for any $v \in \Real^3$ with $v_\y = 0$,
 $\mN_v$ is represented as the set $\{ (x',z') : x' - v_\x \ge |z' - v_\z| \}$
 in the plane $\y = 0$.}
 \label{fig:result_xz}
\end{figure}
Let $\upsilon'$ be the element of $\mN_{\upsilon^{0,0}} \cap \mN_{\upsilon^{0,1}}$
that has the minimum $\x$-component,
$q_0$ be the maximum value of $q$ satisfying $\upsilon' \in \mN_{\upsilon^{1,1}(q)}$,
and $q_1$ be the minimum value of $q$ satisfying $\upsilon^{1,1}(q) \in \mN_{\upsilon'}$.
Also, let $\varsigma^{(0)} \coloneqq \upsilon^{1,1}(q_0)$
and $\varsigma^{(1)} \coloneqq \upsilon^{1,1}(q_1)$.
Then, we can easily verify
\begin{alignat}{1}
 u^\opt(q) &=
 \begin{dcases}
  \upsilon', & q \le q_0, \\
  \upsilon' + \frac{q - q_0}{q_1 - q_0} [\varsigma^{(1)} - \upsilon'], & q_0 < q < q_1, \\
  \upsilon^{1,1}(q), & q \ge q_1. \\
 \end{dcases}
 \label{eq:unital_uopt}
\end{alignat}
Note that $\upsilon'$ and $\varsigma^{(1)}$ can be easily obtained from
$s_0$, $s_2$, $t_0$, and $t_2$.
Moreover, from Eq.~\eqref{eq:unital_Popt}, we have
\begin{alignat}{1}
 P^\opt(\pinc) &=
 \begin{dcases}
  2\upsilon'_\x - q_0 \pinc, & \pinc < p_0, \\
  2\varsigma^{(1)}_\x - q_1 \pinc, & \mbox{otherwise}, \\
 \end{dcases} \nonumber \\
 p_0 &\coloneqq
 \begin{dcases}
  \frac{2\varsigma^{(1)}_\x - 2\upsilon'_\x}{q_1 - q_0}, & q_1 \neq q_0, \\
  1, & \mbox{otherwise}. \\
 \end{dcases}
 \label{eq:unital_Popt_ex}
\end{alignat}
$2u^\opt_\x(q)$ and $P^\opt(\pinc)$ are shown in Fig.~\ref{fig:resultP}.
\begin{figure}[tb]
 \centering
 \InsertPDF{0.6}{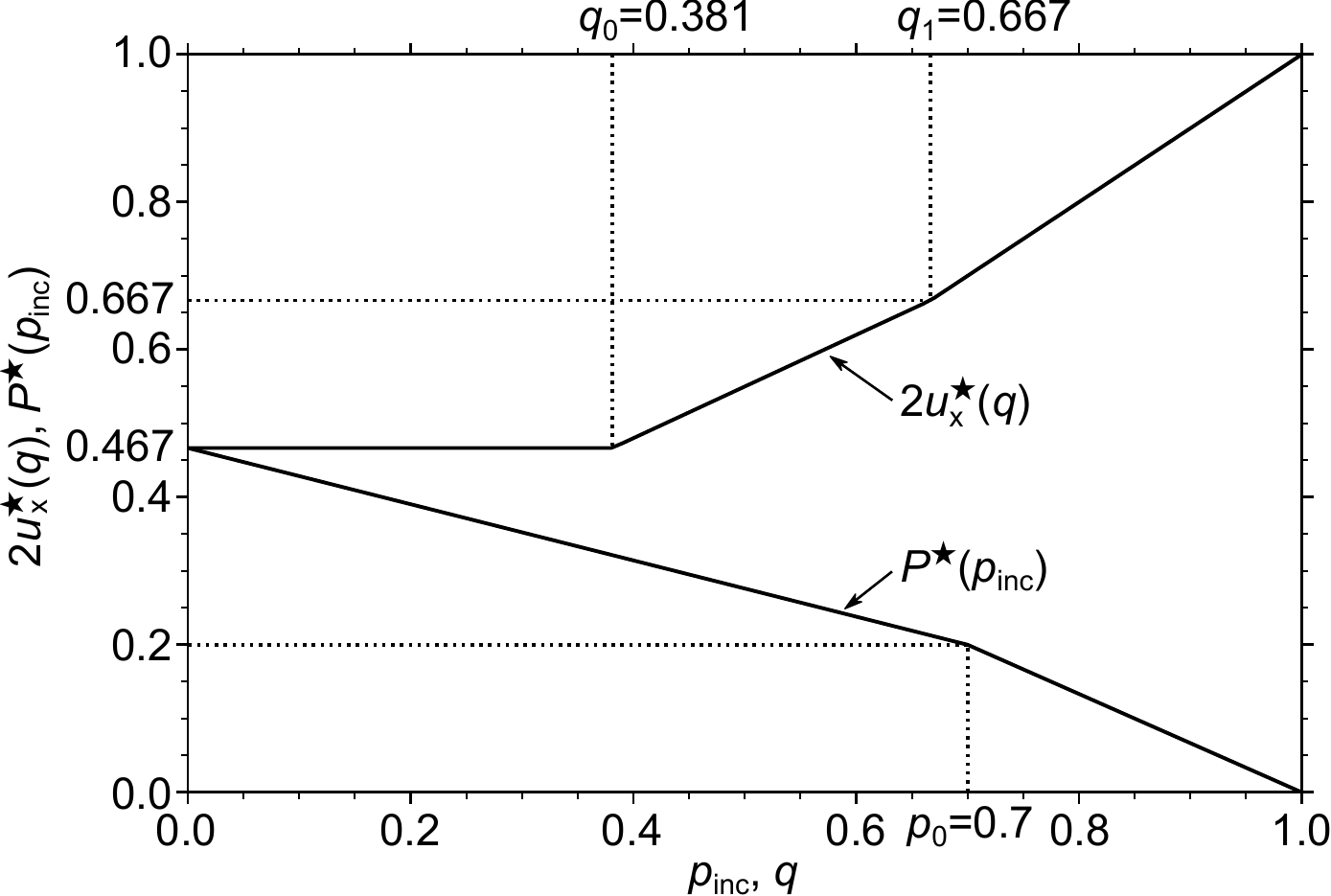}
 \caption{$2u^\opt_\x(q)$ and $P^\opt(\pinc)$ in the same conditions of Fig.~\ref{fig:result_xz}.}
 \label{fig:resultP}
\end{figure}
As seen in Eq.~\eqref{eq:unital_uopt}, $2u^\opt_\x(q)$ can be generally represented
with three line segments corresponding to $q \le q_0$, $q_0 < q < q_1$,
and $q \ge q_1$.
Note that since $\upsilon^{1,1}_\x(q) = q/2$ holds,
$2u^\opt_\x(q) = 2\upsilon^{1,1}_\x(q) = q$ holds when $q \ge q_1$.
Also, as seen in Eq.~\eqref{eq:unital_Popt_ex},
$P^\opt(\pinc)$ can be generally represented with two line segments
corresponding to $\pinc < p_0$ and $\pinc \ge p_0$.

\section{Conclusion}

We have studied a generalized problem of discriminating quantum processes
each of which can consist of several time steps and can have an internal memory.
This problem can be formulated as a convex problem with a quantum tester.
We first showed that the optimal values of this problem and its Lagrange dual problem
coincide (i.e., the strong duality holds).
Based on this result, necessary and sufficient conditions for a tester to be optimal were provided.
Necessary and sufficient conditions that the optimal value
remain unchanged even when a certain additional constraint is imposed were also given.
We next showed that for a problem that is symmetric with respect to given group actions,
there exists an optimal solution having the same type of symmetry.
Moreover, we discussed a minimax strategy for a generalized process discrimination problem.

Process discrimination problems can be interpreted as an extension of
state discrimination problems.
In state discrimination, the formulation of the problem as a convex problem
is useful for developing analytical and numerical techniques,
such as deriving analytical expressions for optimal measurements,
developing numerical algorithms for efficiently obtaining optimal solutions,
finding near-optimal measurements (e.g., a square root measurement),
and obtaining upper and lower bounds on optimal values.
We expect that our results will allow us to extend these techniques to
a broad class of process discrimination problems.

\begin{acknowledgments}
 We are grateful to O.~Hirota and T.~S.~Usuda for useful discussions.
 This work was supported by JSPS KAKENHI Grant Number JP19K03658.
\end{acknowledgments}

\appendix

\section{Proof of Theorem~\ref{thm:dual}} \label{append:dual}

Let $P^\opt$ and $D^\opt$ be, respectively, the optimal values of Problems~\eqref{prob:P}
and \eqref{prob:D}.
We consider the following Lagrangian associated with Problem~\eqref{prob:P}:
\begin{alignat}{1}
 L(\Phi,\varphi,\chi,q) &\coloneqq \summ \braket{\Phi_m,c_m} + \Braket{\varphi - \summ \Phi_m,\chi}
 - \sumj q_j \eta_j(\Phi) \nonumber \\
 &= \braket{\varphi,\chi} + \sumj q_j b_j - \summ \braket{\Phi_m,\chi - z_m(q)},
 \label{eq:L}
\end{alignat}
where $\Phi \in \mC$, $\varphi \in \S$, $\chi \in \Her_\tV$,
and $q \coloneqq \{q_j\}_{j=0}^{J-1} \in \Real_+^J$.
It follows that
\begin{alignat}{1}
 \inf_{\chi,q} L(\Phi,\varphi,\chi,q) &=
 \begin{dcases}
  \summ \braket{\Phi_m,c_m}, & \Phi \in \ol{\P}, ~\varphi = \summ \Phi_m \\
  -\infty, & \mathrm{otherwise},
 \end{dcases} \nonumber \\
 \sup_{\Phi,\varphi} L(\Phi,\varphi,\chi,q) &=
 \begin{dcases}
  D_\S(\chi,q), & (\chi,q) \in \D, \\
  \infty, & \mathrm{otherwise}
 \end{dcases}
\end{alignat}
holds.
Thus, from the max-min inequality, we have
\begin{alignat}{1}
 P^\opt &= \sup_{\Phi,\varphi} \inf_{\chi,q} L(\Phi,\varphi,\chi)
 \le \inf_{\chi,q} \sup_{\Phi,\varphi} L(\Phi,\varphi,\chi) = D^\opt.
 \label{eq:L_maxmin}
\end{alignat}

It remains to show the strong duality.
In the case of $D^\opt = -\infty$, the strong duality obviously holds from $P^\opt = D^\opt = -\infty$.
Now, we consider the other case.
It suffices to show that there exists $\Phi^\opt \in \ol{\P}$ such that $P(\Phi^\opt) \ge D^\opt$,
in which case, from $P^\opt \ge P(\Phi^\opt)$, we have $P^\opt = D^\opt$.
We consider the set
\begin{alignat}{1}
 \mZ &\coloneqq \left\{ \left( \{y_m + z_m(q) - \chi\}_{m=0}^{M-1}, D_\S(\chi,q) - d \right) :
 (\chi,y,d,q) \in \mZ_0 \right\} \nonumber \\
 &\qquad \subset \Her_\tV^M \times \Real,
\end{alignat}
where $y \coloneqq \{ y_m \}_{m=0}^{M-1}$ and
\begin{alignat}{1}
 \mZ_0 &\coloneqq \left\{ \left( \chi,y,d,q \right)
 \in \Her_\tV \times \mC^* \times \Real \times \Real_+^J : d < D^\opt \right\}.
\end{alignat}
One can easily verify that $\mZ$ is a nonempty convex set.
We can show $(\{\zero\},0) \not\in \mZ$.
Indeed, for any $(\chi,y,d,q) \in \mZ_0$ such that
$y_m + z_m(q) - \chi = \zero$ $~(\forall m)$,
since $\{\chi-z_m(q)\}_m = y \in \mC^*$ [i.e., $(\chi,q) \in \D$] holds,
$D_\S(\chi,q) - d \ge D^\opt - d > 0$ must hold.
From the separating hyperplane theorem \cite{Dha-Dut-2011},
there exists $(\{\Psi_m\}_{m=0}^{M-1},\alpha) \neq (\{\zero\},0) \in \Her_\tV^M \times \Real$
such that
\begin{alignat}{1}
 &\summ \braket{\Psi_m, y_m + z_m(q) - \chi} + \alpha [D_\S(\chi,q) - d] \ge 0, \nonumber \\
 &\qquad \forall (\chi,y,d,q) \in \mZ_0.
 \label{eq:separation0}
\end{alignat}
By substituting $y_m = \kappa y'_m$ $~(\kappa \in \Real_+, \{y'_m\}_m \in \mC^*)$
into Eq.~\eqref{eq:separation0} and taking the limit $\kappa \to \infty$,
we obtain $\{\Psi_m\}_m \in \mC$.
Also, we have $\alpha \ge 0$ in the limit $d \to -\infty$.
We can show $\alpha > 0$.
[Indeed, assume by contradiction that $\alpha = 0$.
Substituting $\chi = \kappa \I_\tV$ $~(\kappa \in \Real_+)$ into Eq.~\eqref{eq:separation0}
and taking the limit $\kappa \to \infty$ gives $\summ \Tr \Psi_m \le 0$.
From $\{\Psi_m\}_m \in \mC \subseteq \Pos_\tV^M$, $\Psi_m = \zero$ holds for any $m \in \mI_M$.
This contradicts $(\{\Psi_m\}_m,\alpha) \neq (\{\zero\},0)$.]
Let $\Phi^\opt_m \coloneqq \Psi_m / \alpha$; then, Eq.~\eqref{eq:separation0} is rewritten by
\begin{alignat}{1}
 &\summ \braket{\Phi^\opt_m, y_m + z_m(q) - \chi} + D_\S(\chi,q) - d \ge 0, \nonumber \\
 &\qquad \forall (\chi,y,d,q) \in \mZ_0.
 \label{eq:separation}
\end{alignat}
Substituting $\chi = \kappa \chi'$ $~(\kappa \in \Real_+, \chi' \in \Her_\tV)$
and $q_j = 0$ into Eq.~\eqref{eq:separation} and taking the limit $\kappa \to \infty$
yields $\lambdaS(\chi') \ge \summ \braket{\Phi^\opt_m,\chi'}$ $~(\forall \chi' \in \Her_\tV)$.
This implies $\summ \Phi^\opt_m \in \S$.
[Indeed, assume by contradiction that $\summ \Phi^\opt_m$ is not in $\S$;
then, from separating hyperplane theorem, there exists $\chi' \in \Her_\tV$ such that
$\braket{\phi,\chi'} < \braket{\summ \Phi^\opt_m,\chi'}$ $~(\forall \phi \in \S)$,
which contradicts $\lambdaS(\chi') \ge \summ \braket{\Phi^\opt_m,\chi'}$.]
Thus, $\Phi^\opt \in \ol{\T}$ holds [see Eq.~\eqref{eq:T}].
By substituting $q_j = \kappa \delta_{j,j'}$ $~(j' \in \mI_J)$ into Eq.~\eqref{eq:separation},
we have $\eta_{j'}(\Phi^\opt) \le 0$ in the limit $\kappa \to \infty$.
Thus, $\Phi^\opt \in \ol{\P}$ holds from Eq.~\eqref{eq:TestercC2}.
By substituting $y_m = \zero$, $\chi = \zero$, and $q_j = 0$ into Eq.~\eqref{eq:separation}
and taking the limit $d \to D^\opt$,
we have $P(\Phi^\opt) = \summ \braket{\Phi^\opt_m,c_m} \ge D^\opt$.
\QED

\section{Proof of Proposition~\ref{pro:chi_comb_feasible}} \label{append:chi_comb}

Before proving Proposition~\ref{pro:chi_comb_feasible}, we first show the following lemma.
\begin{lemma} \label{lemma:dual}
 For any $\chi \in \Her_\tV$,
 $\lambdaSG(\chi)$ is equal to the optimal value of the following optimization problem:
 \begin{alignat}{1}
  \begin{array}{ll}
   \mbox{minimize} & \omega_0 \\
   \mbox{subject~to} & \Trp{\WT} \chi \le \I_\VT \ot \omega_{T-1}, \\
   & \Trp{\Wt} \omega_t \le \I_\Vt \ot \omega_{t-1} ~(\forall t \in \range{1}{T-1}) \\
  \end{array}
  \label{prob:D0}
 \end{alignat}
 with $\left\{ \omega_t \in \Her_{\WVt \ot \cdots \ot \W_1 \ot \V_1} \right\}_{t=0}^{T-1}$
 (note that $\omega_0 \in \Real$ holds).
\end{lemma}
\begin{proof}
 Let us consider the following Lagrangian associated with Problem~\eqref{prob:D0}:
 \begin{alignat}{1}
  L_0(\tau,\omega) &\coloneqq \omega_0 + \sum_{t=1}^{T-1}
  \Braket{\tau_t, \Trp{\Wt} \omega_t - \I_\Vt \ot \omega_{t-1}} \nonumber \\
  &\quad + \Braket{\tau_T, \Trp{\WT} \chi - \I_\VT \ot \omega_{T-1}} \nonumber \\
  &= \Braket{1 - \Tr \tau_1, \omega_0}
  + \sum_{t=1}^{T-1} \Braket{\I_\Wt \ot \tau_t - \Trp{\V_{t+1}} \tau_{t+1}, \omega_t}
  \nonumber \\
  &\quad + \Braket{\I_\WT \ot \tau_T, \chi},
  \label{eq:L0}
 \end{alignat}
 where $\tau \coloneqq \{ \tau_t \in \Pos_{\Vt \ot \W_{t-1} \ot \V_{t-1} \ot \cdots \ot \W_1 \ot \V_1} \}_{t=1}^T$
 and $\omega \coloneqq \left\{ \omega_t \in \Her_{\WVt \ot \cdots \ot \W_1 \ot \V_1}
 \right\}_{t=0}^{T-1}$.
 Due to the max-min inequality, we have
 \begin{alignat}{1}
  \sup_{\tau} \inf_{\omega} L_0(\tau,\omega) &\le \inf_{\omega} \sup_{\tau} L_0(\tau,\omega).
  \label{eq:L0_maxmin}
 \end{alignat}
 From the second equation of Eq.~\eqref{eq:L0},
 it is straightforward to derive that if $\omega$ is a feasible solution to
 Problem~\eqref{prob:D0},
 then $\sup_{\tau} L_0(\tau,\omega) = \omega_0$, otherwise $\infty$.
 Thus, the right-hand side of Eq.~\eqref{eq:L0_maxmin} is equal to
 the optimal value of Problem~\eqref{prob:D0}, denoted by $D_0^\opt$.
 Similarly, it follows from the last equation of Eq.~\eqref{eq:L0}
 that the left-hand side of Eq.~\eqref{eq:L0_maxmin} is equal to
 the optimal value of the following problem:
 \begin{alignat}{1}
  \begin{array}{ll}
   \mbox{maximize} & \braket{\I_\WT \ot \tau_T, \chi} \\
   \mbox{subject~to} & \Trp{\Vt} \tau_t = \I_{\W_{t-1}} \ot \tau_{t-1} ~(\forall t \in \range{2}{T}), \\
   & \Tr \tau_1 = 1 \\
  \end{array}
  \label{prob:P0}
 \end{alignat}
 with $\tau$.
 The constraint is equivalent to $\tau_T \in \OtT \Chn_{\Vt \ot \W_{t-1}}$
 (with $\W_0 \coloneqq \Complex$),
 or, equivalently, $\I_\WT \ot \tau_T \in \SG$.
 Thus, the optimal value is $\sup_{\varphi \in \SG} \braket{\varphi, \chi} = \lambdaSG(\chi)$.
 To prove $D_0^\opt = \lambdaSG(\chi)$, it suffices to show that Slater's condition holds.
 Let $\tau' \coloneqq \{ \tau'_t \}_{t=1}^T$ with
 $\tau'_1 \coloneqq \I_{\V_1} / N_{\V_1}$ and
 $\tau'_t \coloneqq \I_\Vt / N_\Vt \ot \I_{\W_{t-1}} \ot \tau'_{t-1}$ $~(t \in \range{2}{T})$;
 then, $\tau'$ is a feasible solution to
 Problem~\eqref{prob:P0} and $\tau'_t$ is positive definite for each $t \in \range{1}{T}$,
 which implies that Slater's condition holds.
\end{proof}

We are now ready to prove Proposition~\ref{pro:chi_comb_feasible}.
Arbitrarily choose $(\chi', q) \in \D$.
Let $\{ \omega'_t \}_{t=0}^{T-1}$ be an optimal solution to
Problem~\eqref{prob:D0} with $\chi \coloneqq \chi'$.
Also, let
\begin{alignat}{2}
 \omega_0 &\coloneqq \omega'_0, \nonumber \\
 \omega_t &\coloneqq \omega'_t + \frac{\I_\Wt}{N_\Wt}
 \ot (\I_\Vt \ot \omega_{t-1} - \Trp{\Wt} \omega'_t),
 &~ &\forall t \in \range{1}{T-1}, \nonumber \\
 \chi &\coloneqq \chi' + \frac{\I_\WT}{N_\WT} \ot (\I_\VT \ot \omega_{T-1} - \Trp{\WT} \chi').
 \label{eq:chi_I}
\end{alignat}
We have $\omega_t \ge \omega'_t$ $~(t \in \range{1}{T-1})$ and $\chi \ge \chi'$,
which follows from
$\I_\Vt \ot \omega_{t-1} \ge \I_\Vt \ot \omega'_{t-1} \ge \Trp{\Wt} \omega'_t$ and
$\I_\VT \ot \omega_{T-1} \ge \I_\VT \ot \omega'_{T-1} \ge \Trp{\WT} \chi'$.
Thus,
\begin{alignat}{1}
 \summ \braket{\Phi_m, \chi - z_m(q)} &\ge \summ \braket{\Phi_m, \chi' - z_m(q)} \ge 0,
 \quad \forall \Phi \in \mC,
\end{alignat}
which gives $\{\chi - z_m(q)\}_{m=0}^{M-1} \in \mC^*$
[i.e., $(\chi, q) \in \D$].
From Eq.~\eqref{eq:chi_I}, we have
\begin{alignat}{1}
 \Trp{\Wt} \omega_t &= \I_\Vt \ot \omega_{t-1}, \quad \forall t \in \range{1}{T-1}, \nonumber \\
 \Trp{\WT} \chi &= \I_\VT \ot \omega_{T-1}.
 \label{eq:chi_I2}
\end{alignat}
Assume now that $\chi \in \Lin(\Chn_\tV)$ holds, i.e.,
$\chi$ is expressed in the form $\chi = \beta_+ \chi_+ - \beta_- \chi_-$
$~(\beta_\pm \in \Real_+, ~\chi_\pm \in \Chn_\tV)$;
then, $\omega_0 = \beta_+ - \beta_-$ obviously holds.
From Eq.~\eqref{eq:Phim_c_1}, we have
$\lambdaSG(\chi) = \beta_+ - \beta_- = \omega_0 = \omega'_0 = \lambdaSG(\chi')$.

To complete the proof, we have to show $\chi \in \Lin(\Chn_\tV)$.
Let $u_t \coloneqq \I_{\WVt \ot \cdots \ot \W_1 \ot \V_1} / \prod_{t'=1}^t N_{\W_{t'}}$,
$\chi^+ \coloneqq \chi + p u_T$, and $\omega_t^+ \coloneqq \omega_t + p u_t$
$~(t \in \range{0}{T-1})$,
where $p \in \Real_+$ is taken to be sufficiently large such that
$\chi^+ \ge \zero$, $\omega_t^+ \ge \zero$ $~(\forall t \in \range{1}{T-1})$,
and $\omega_0^+ > 0$.
From Eq.~\eqref{eq:chi_I2} and $\Trp{\Wt} u_t = \I_\Vt \ot u_{t-1}$,
we have $\Trp{\Wt} \omega_t^+ = \I_\Vt \ot \omega_{t-1}^+$
$~(\forall t \in \range{1}{T-1})$
and $\Trp{\WT} \chi^+ = \I_\VT \ot \omega_{T-1}^+$,
which gives $\chi^+ / \omega_0^+ \in \Chn_\tV$ [see Eq.~\eqref{eq:Comb}].
From $u_T \in \Chn_\tV$, $\chi = \chi^+ - pu_T \in \Lin(\Chn_\tV)$ holds.
\QED

\section{Proof of Theorem~\ref{thm:Phi_nas}} \label{append:Phi_nas}

We will prove it using the proof of Theorem~\ref{thm:dual} in Appendix~\ref{append:dual}.
Arbitrarily choose $\Phi \in \P$ and $(\chi, q) \in \D$.
We consider a sequence $\{ \varphi_n \in \S \}_{n=1,2,\dots}$ such that
$\lim_{n \to \infty} \braket{\varphi_n,\chi} = \lambdaS(\chi)$.
From Eq.~\eqref{eq:L}, we have
\begin{alignat}{1}
 D_\S(\chi, q) - P(\Phi) &= - \sumj q_j \eta_j(\Phi)
 + \summ \braket{\Phi_m, \chi - z_m(q)} \nonumber \\
 &\quad + \left[ \lambdaS(\chi) - \summ \braket{\Phi_m, \chi} \right]
 \label{eq:P_DP_gap}
\end{alignat}
in the limit $n \to \infty$.
Since each term on the right-hand side of Eq.~\eqref{eq:P_DP_gap} is always nonnegative,
$P(\Phi) = D_\S(\chi, q)$ holds
[i.e., $\Phi$ and $(\chi, q)$ are, respectively, optimal for
Problems~\eqref{prob:P} and \eqref{prob:D}] if and only if
Eq.~\eqref{eq:channel_nas} holds.
\QED

\bibliographystyle{apsrev4-1}
%


\end{document}